\documentclass[twoside,11pt]{article}
\usepackage{blindtext}
\usepackage{jmlr2e}

\usepackage{lastpage}
\jmlrheading{23}{2022}{1-\pageref{LastPage}}{6/20; Revised
8/22}{9/22}{20-643}{Oscar Hernan Madrid Padilla, Yi Yu and Carey E.~Priebe}
\ShortHeadings{Change point in dependent dynamic RDPG}{Madrid Padilla, Yu and Priebe}

\usepackage{amsmath, bbm}
\usepackage{hyperref}
\usepackage{graphicx}
\usepackage{verbatim}

\usepackage{algorithm}
\usepackage{algpseudocode}
\usepackage[dvipsnames]{xcolor}
\usepackage{tikz}
\usetikzlibrary{arrows.meta}
\usetikzlibrary{decorations.pathreplacing}
\usetikzlibrary{positioning, shapes.geometric}

\algnewcommand\algorithmicinput{\textbf{INPUT:}}
\algnewcommand\INPUT{\item[\algorithmicinput]}
\algnewcommand\algorithmicoutput{\textbf{OUTPUT:}}
\algnewcommand\OUTPUT{\item[\algorithmicoutput]}

\DeclareMathAlphabet{\mathpzc}{OT1}{pzc}{m}{it}

\usepackage{cleveref}
\allowdisplaybreaks

\DeclareMathOperator*{\argmax}{argmax}
\DeclareMathOperator*{\argsup}{argsup}
\newtheorem{model}{Model}
\newtheorem{assumption}{Assumption}

\firstpageno{1}

\begin{document}

\title{Change point localization in dependent \\dynamic nonparametric random dot product graphs}

\author{\name Oscar Hernan Madrid Padilla \email oscar.madrid@stat.ucla.edu \\
       \addr Department of Statistics\\
       University of California\\
       Los Angeles, CA 90095-1554, USA
       \AND
       \name Yi Yu \email yi.yu.2@warwick.ac.uk \\
       \addr Department of Statistics\\
       University of Warwick\\
	   Coventry, CV4 7AL, UK
	   \AND
	   \name Carey E.~Priebe \email cep@jhu.edu \\
	   \addr Department of Applied Mathematics and Statistics \\
	   Johns Hopkins University \\
	   Baltimore, MD 21218-2682, USA}

\editor{David Blei}

\maketitle

\begin{abstract}
In this paper, we study the offline change point localization problem in a sequence of dependent nonparametric random dot product graphs.  To be specific, assume that at every time point, a network is generated from a nonparametric random dot product graph model \citep[see e.g.][]{athreya2017statistical}, where the latent positions are generated from unknown underlying distributions.  The underlying distributions are piecewise constant in time and change at unknown locations, called change points.  Most importantly, we allow for dependence among networks generated between two consecutive change points.  This setting incorporates edge-dependence within networks and temporal dependence between networks, which is the most flexible setting in the published literature.

To accomplish the task of consistently localizing change points, we propose a novel change point detection algorithm, consisting of two steps.  First, we estimate the latent positions of the random dot product model, our theoretical result being a refined version of the state-of-the-art results, allowing the dimension of the latent positions to diverge.  Subsequently, we construct a nonparametric version of the CUSUM statistic \citep[e.g.][]{Page1954, padilla2019optimal} that allows for temporal dependence.  Consistent localization is proved theoretically and supported by extensive numerical experiments, which illustrate state-of-the-art performance.  We also provide in depth discussion of possible extensions to give more understanding and insights.

\end{abstract}

\begin{keywords}
  Dependent dynamic networks, Nonparametric random dot product graph models, Change point localization.
\end{keywords}

\section{Introduction}

Computationally-efficient and theoretically-justified change point localization methods that can handle new data types are in high demand, due to technological advances in a broad range of application areas including finance, biology, social sciences, to name only a few.  The literature on change point detection is extensive, including the univariate mean case \citep[e.g.][]{FrickEtal2014, fryzlewicz2014wild, wang2018univariate}, the high-dimensional mean case \citep[e.g.][]{wang2016high, cho2016change}, the robust mean case \citep[e.g.][]{fearnhead2018changepoint, pein2017heterogeneous}, the covariance case \citep[e.g.][]{AueEtal2009, wang2017optimal, avanesov2018change}, the univariate nonparametric case \citep[e.g.][]{zou2014nonparametric, padilla2019optimal}, and the multivariate nonparametric case \citep[e.g.][]{arlot2012kernel, matteson2014nonparametric, garreau2018consistent, padilla2019optimal_2}. 
 
In this paper we are concerned with offline change point localization in dynamic networks.  Let $\{A(t)\}_{t=1}^T \subset \{0,1\}^{n \times n}$ be a sequence of adjacency matrices generated from a sequence of distributions $\{\mathcal{L}_t\}_{t = 1}^T$, such that for an unknown sequence of change points $\{\eta_k\}_{k = 1}^K \subset \{2, \ldots, T\}$ with $1 = \eta_0 < \eta_1 < \ldots < \eta_K \leq  T < \eta_{K+1} = T+1$, we have that
	\[
		\mathcal{L}_{t-1} \neq \mathcal{L}_t, \quad \text{if and only if} \quad t \in \{\eta_1, \ldots, \eta_{K}\}.
	\]
	The goal is to estimate the change point collection $\{\eta_k\}_{k = 1}^K$ accurately.

The model described can be used to study various application problems.  For instance, in epidemiology, studying the dynamic networks formed by human interaction can facilitate the detection of transmissible diseases outbreaks; in finance, dynamic networks formed by within-companies transactions over time may possess abrupt changes which indicate abnormal market behaviours; in neuroscience, we provide a detailed real data example in the context of detecting changes in the neuronal activity in \Cref{sec-real-data}.  In response to the growing demand from applications, there has been recently an increasing interest in the literature studying the model described above.  \cite{wang2018optimal} considered an independent sequence of inhomogeneous Bernoulli networks and presented a nearly optimal change point localization algorithm, accompanied with a phase transition phenomenon.     \cite{zhao2019change} assumed an independent sequence of graphon models with independent edges and proposed consistent yet optimal localization result.  Other network change point papers include \cite{WangEtal2014}, \cite{CribbenYu2017}, \cite{LiuEtal2018}, \cite{ChuChen2017}, \cite{Mukherjee}, among others.  We would like to mention that both \cite{CribbenYu2017} and \cite{LiuEtal2018} have exploited the eigenvectors information to conduct change point detection, but both lack theoretical results.  Our paper, to the best of our knowledge, is the first to provide theoretical justifications for eigenvector-based change point detection methods.  More in-depth comparisons with \cite{wang2018optimal} will be conducted later in the paper.

\subsection{Random dot product graph models}\label{sec-rdphm}

Different from the aforementioned papers, in order to allow for dependence among edges, we assume that at every time point, the network is generated from a random dot product graph \citep[e.g.][]{young2007random, athreya2017statistical}.  We formally define the model in Definitions~\ref{def-1} and \ref{def-3}, which are both based on \cite{athreya2017statistical}.

\begin{definition}[Inner product distribution]\label{def-1}
	 Let $F$ be a probability distribution whose support is given by $\mathcal{X}_F \subset \mathbb{R}^{d}$.  We say that $F$ is a $d$-dimensional inner product distribution on $\mathbb{R}^d$ if for all $x, y \in \mathcal{X}_F$, it holds that $x^{\top}y \in [0, 1]$.	
\end{definition}

\begin{definition}[Random dot product graph with distribution $F$] \label{def-3} 
Let $F$ be an inner product distribution with $\{X_i\}_{i = 1}^n \stackrel{\mbox{i.i.d.}}{\sim} F$.  Let $X = (X_1, \ldots, X_n)^{\top} \in \mathbb{R}^{n \times d}$.  Suppose $A$ is a random adjacency matrix given by 
	\begin{equation}\label{eq-def-rdpg-1}
		\mathbb{P}\left\{A \mid X\right\} = \prod_{1 \leq i < j \leq n} (X_i^{\top}X_j)^{A_{ij}}(1 - X_i^{\top}X_j)^{1 - A_{ij}}.
	\end{equation}
	We write $A \sim \mathrm{RDPG}(F, n)$.
\end{definition}

We would like to make a few comments regarding random dot product graph models.  For first time reading, one can safely skip this and jump to \Cref{sec-contributions}.

\medskip
\noindent \textbf{Equivalence of distributions}

It can be seen from  Definition \ref{def-3} that the latent positions come into play only through their inner products, i.e.~we have
	\[
		A_{ij} \sim \mathrm{Ber}(X_i^{\top}X_j), \quad 1 \leq i, j \leq n.
	\]
	This means that one can apply any orthonormal rotations to all the latent positions and retain the same distribution of $A$.  In light of this rotational invariance, we define the equivalence of inner product distributions below, which is also from \cite{athreya2017statistical}.

\begin{definition}[Equivalence of inner product distributions]\label{def-equivalence} 
	If both $F(\cdot)$ and $G(\cdot)$ are  inner product distributions defined on $\mathbb{R}^d$, and there exists an orthogonal operator $U: \, \mathbb{R}^d \to \mathbb{R}^d$ such that $F = G \circ U$, then we say $F$ and $G$ are equivalent. 
\end{definition}

\medskip
\noindent \textbf{Community structures}

The random dot product graph is a generalization of the stochastic block model \citep{HollandEtal1983}, where the latent positions $X$ are assumed to be fixed and satisfy
	\[
		XX^{\top} = ZQZ^{\top},
	\]
	where $Z \in \{0, 1\}^{n \times d}$ is a membership matrix,  with each row  consisting  of one and only one entry being 1 and $Q \in [0, 1]^{d \times d}$ is a connectivity matrix encoding the edge probabilities.  
	
One may be puzzled by the observation that under  Definition \ref{def-3}, we have that for any $(i, j) \in \{1, \ldots, n\}^2$, $i \neq j$,
	\[
		\mathbb{E}(A_{ij}) = \mathbb{E}(X_i^{\top}X_j) = \mathbb{E}(X_1^{\top}X_2),
	\]	
	where the second identity follows from the fact that within a network the latent positions are i.i.d., and therefore one loses the community structure and connections from the stochastic block model.  
	
This observation is due to the randomness of the latent positions.  To enforce a version of ``communities'' under  Definition \ref{def-3}, one may introduce a membership vector and treat the distribution $F$ as a mixture distribution.  To be specific, we have an alternative to Definition \ref{def-3} below.	

\begin{definition} \label{def-alt-rdpg} 
Let $\tau_1,\ldots,\tau_n$ be i.i.d.~random variables satisfying
	\[
		\mathbb{P}\{\tau = m\} = \pi_m, \quad \pi_m \geq 0, \, m \in \{1, \ldots, M\}, \, \sum_{m = 1}^M \pi_m = 1,
	\]
	where $M$ is a positive integer.  Let $\{F_m\}_{m = 1}^M$ be a sequence of $d$-dimensional inner product distributions.  Assume that
	\[
		X_i \mid \tau_i \stackrel{ind.}{\sim} F_{\tau_i}, \quad i = 1, \ldots, n.
	\]
	Let $X = (X_1, \ldots, X_n)^{\top} \in \mathbb{R}^{n \times d}$.  Suppose $A$ is a random adjacency matrix given by 
	\[
		\mathbb{P}\left\{A \mid X\right\} = \prod_{1 \leq i < j \leq n} (X_i^{\top}X_j)^{A_{ij}}(1 - X_i^{\top}X_j)^{1 - A_{ij}}.
	\]
	We write $A \sim \mathrm{RDPG}(F, n)$, where 
	\[
		F = \sum_{m = 1}^M \pi_m F_{m}.
	\]
\end{definition}

We remark that Definition \ref{def-alt-rdpg} is a special case of  Definition \ref{def-3}. Therefore the theoretical results based on  Definition \ref{def-3} also hold for Definition  \ref{def-alt-rdpg}.  The vector $\tau$ prompts the vertex correspondence in a dynamic network.  For instance, one may assume a sequence of RDPG($F, n$) using Definition \ref{def-alt-rdpg}, with latent positions drawn independently and the membership vector unchanged.  There are also other variants.  For instance, one may also assume instead that the membership vector $\tau$ is fixed.

\subsection{List of contributions}\label{sec-contributions}

We highlight the contributions of this paper.

First of all, we propose a novel algorithm for change point localization in dependent dynamic random dot product graph models, see \Cref{algorithm:WBS}.  This proceeds by first estimating the latent positions $\{\widehat{X}_i(t)\}_{i = 1, t = 1}^{n, T}$.  However,  due to the latent positions' rotational-invariance properties discussed in \Cref{sec-rdphm}, one pertaining challenge in the RDPG literature is to match the rotations of the latent position estimators of different networks \citep[e.g.][]{athreya2017statistical, cape2019two}.  We propose a novel way to get around this issue with matching. Specifically,  we define $\widehat{Y}^t_{ij} = (X_i(t))^{\top}X_j(t)$, and construct a Kolmogorov--Smirnov CUSUM statistic \citep{padilla2019optimal} based on $\{\widehat{Y}^t_{ij}: (i, j) \in \{(l, n/2+l), \, l = 1, \ldots, n/2\}, t = 1, \ldots, T\}$.  One may question the power of the Kolmogorov--Smirnov distance, but it allows for more general distributions for latent positions. Among those distributions stochastic block models are special cases.  One may also question the effectiveness of using only a subset of all the possible edges, we will discuss in \Cref{sec-possible-extensions} that in terms of order, this is in fact the same as using all possible edges.

Secondly, under an appropriate signal-to-noise ratio condition, we prove that our proposed method  (\Cref{algorithm:WBS}) can estimate the number and locations of change points consistently, which will be formally stated in \Cref{sec:main_theorem}.  It is worth mentioning that \Cref{thm-main} handles the situation where there exists dependence across time and among edges.  This is not shown in the existing network change point detection literature.  To be more specific, the dependence among edges are imposed by assuming the latent positions are random and the edges are conditionally independent given the latent positions.  Our proposed method is also robust to some model mis-specification, see the discussions following \Cref{thm-main} for details.

Thirdly, we provide in-depth discussions on the characterization of jumps in \Cref{sec:jumps}.  Note that the data we have are a collection of adjacency matrices. However, as stated in Definition \ref{def-3}, the data generating mechanism depends on latent positions' distributions $F$s.  A natural question is whether the changes in $F$ will lead to the changes in the distributions of the adjacency matrices, and if so, whether we can characterize the changes.  The results we developed in \Cref{sec:jumps} are interesting \emph{per se}, and can shed light on network testing problems.
 
Lastly, numerical experiments provide ample evidence on the strength of our proposed approach. In particular, we  highlight the advantage of our method in scenarios with dependent networks.

\medskip

The rest of the paper is organized as follows.  \Cref{sec:model} provides the formal problem setup and our proposed method in detail, including discussions on possible extensions.  The characterization of the distributional changes and statistical guarantees for our approach are collected in \Cref{sec:theory}.  We conclude with numerical experiments in \Cref{sec:experiments} and final discussions in \Cref{sec-conclusion}.  Technical details are deferred to the Appendix.

\section{Methodology}
\label{sec:model}

\subsection{Setup}
We first formally state the full model descriptions.
\begin{model}\label{assume:model-rdpg}
	Let $\{A(1),\ldots, A(T)\} \subset \mathbb{R}^{n\times n}$ be a sequence of adjacency matrices of random dot product graphs, satisfying the following.
\begin{enumerate}
	\item (\textbf{Random dot product graphs}) For any $t \in \{1, \ldots, T\}$, it holds that
		\[
			\mathbb{P}\left\{A(t) \mid X(t)\right\} = \prod_{1 \leq i < j \leq n} (X_i(t)^{\top}X_j(t))^{A_{ij}(t)}(1 - X_i(t)^{\top}X_j(t))^{1 - A_{ij}(t)},
		\]
		where $X(t) = (X_1(t), \ldots, X_n(t))^{\top} \in \mathbb{R}^{n \times d}$ satisfies the following.
		
		There exists a sequence $1 = \eta_0 < \eta_1 < \ldots < \eta_K \leq T   < \eta_{K+1} = T+1$ of time points, called change points.  For $k \in \{0, \ldots, K\}$, we have that
		\[
			X_i(\eta_k) \in \mathbb{R}^d \stackrel{\mathrm{ind}}{\sim} F_{\eta_{k}}, \quad i = 1, \ldots, n,
		\]
		and for $t \in \{\eta_k + 1, \ldots, \eta_{k+1}-1\}$, we have that
		\begin{equation}\label{eq-rho}
			X_i(t) \begin{cases}
 				= X_i(t-1), & \mbox{with probability } \rho, \\
 				\stackrel{\mathrm{ind}}{\sim} F_{\eta_{k}}, & \mbox{with probability } 1-\rho,
 			\end{cases}
		\end{equation}		
		and with $F_{t}$'s satisfying Definition \ref{def-1}. Throughout, we write  $P_t = X(t)X(t)^{\top}$ for the matrix of latent link probabilities at time $t \in \{1, \ldots, T\}$.  
			
	\item (\textbf{Minimal spacing}) The minimal spacing between two consecutive change points satisfies 
		\[
			\min_{k = 1, \ldots, K+1} \{\eta_k-\eta_{k-1}\} = \Delta > 0.
		\]
	\item (\textbf{Minimal jump size}) For each $k \in \{0, \ldots, K\}$ and for any $X, Y \stackrel{\mbox{i.i.d.}}{\sim} F_{\eta_k}$, denote 
		\[
			G_{\eta_k}(z) = \mathbb{P}\left\{X^{\top} Y \leq z\right\}, \quad z \in [0, 1].
		\]
		The magnitudes of the changes in the data generating distribution are such that
		\begin{equation}\label{eq-jump-define}
			\kappa = \min_{k = 1, \ldots, K+1} \kappa_k = \min_{k = 1, \ldots, K} \sup_{z \in [0, 1]} |G_{\eta_k}(z) - G_{\eta_{k-1}}(z)| > 0.
		\end{equation}
	\item Assume that for every $k \in \{0, \ldots, K\}$ and  $i \in \{1,\ldots,n\}$, 
		\[
			\mathbb{E}\left\{X_i(\eta_k)X_i(\eta_k)^{\top}\right\} = \Sigma_k \in \mathbb{R}^{d \times d},
		\]	
	where $\Sigma_k$ has eigenvalues $\mu^k_1 \geq \cdots \geq \mu^k_d > 0$, with $\{\mu^k_l, \, k = 0, \ldots, K, \, l = 1, \ldots, d\}$ all being universal constants.  
\end{enumerate}
\end{model}

In \Cref{assume:model-rdpg}, between two consecutive change points, the latent positions are dependent with exponentially decaying correlations; and for latent positions drawn at time points separated by change points, they are independent.  If $\rho = 0$ in \eqref{eq-rho}, then all the latent positions are independent, which implies that the adjacency matrices are independent.

The distributional changes occurring at change points are quantified through cumulative distribution functions $\{G_{\eta_k}\}$ defined in \Cref{assume:model-rdpg}(3).  Intuitively, since the unconditional distributions of $\{A(t)\}$ are completely characterized by the joint distributions of $\{X_i(t)^{\top}X_j(t)\}$, it is natural to quantify the changes with respect $\{G_{\eta_k}\}$.  (A more detailed discussion on this can be found in \Cref{sec:jumps}.)  In particular, the changes are measured by the Kolmogorov--Smirnov distance in \eqref{eq-jump-define}, since the Kolmogorov--Smirnov distance does not require assumptions about the moments of the distributions, or about their discrete/continuous nature.  With the stochastic block model being a special case of the random dot product graph, the distributions thereof are point-mass distributions, which handicaps the adoption of other (potentially more powerful) distribution distances, including the total variation distance.

\Cref{assume:model-rdpg}(4) is imposed to guarantee that the latent link probabilities satisfy $\mathrm{rank}(P_t) = d$ with high probability.  Without this full-rank assumption, assuming that $r = \mathrm{rank}(\Sigma_k) < d$  implies that there exists a rank-$r$ subspace which characterises the latent positions, and the effective dimension is $r$ instead of $d$.  For simplicity, we assume that $\Sigma_k$ is of full rank.
 
\subsection{Methods}\label{sec-methods}

To arrive at our construction, we start by defining the main statistic, and its population version.  Without loss of generality, we assume that the number of nodes $n$ is an even integer.  If $n$ is odd, then we randomly ignore a certain but fixed node and all edges connecting to it throughout the whole procedure.

\begin{definition}[CUSUM statistics]\label{def-cusum-3}
Let $\mathcal{O} = \{(i, n/2+i), \, i = 1, \ldots, n/2\}$.
\begin{itemize}
\item (Sample version) With $\{A(t)\}_{t = 1}^T \subset \mathbb{R}^{n \times n}$, let
	\[
		\widehat{X}(t) =  U_A(t)\Lambda_A(t)^{1/2},
	\]
	where $U_A(t) \in \mathbb{R}^{n \times d}$ is an orthogonal matrix with columns being the leading $d$ eigenvectors of $A(t)$, and $\Lambda_A(t) \in \mathbb{R}^{d \times d}$ is a diagonal matrix with entries being the largest $d$, in absolute value,  eigenvalues of $A(t)$.  
	
  For any $t \in \{1, \ldots, T\}$ and $(i, j) \in \mathcal{O}$, let
	\[
		\widehat{Y}_{ij}^t = \widehat{X}_i(t)^{\top}\widehat{X}_j(t),
	\]
	where $\widehat{X}_i(t)^{\top}$ is the $i$th row of $\widehat{X}(t)$.  For any integer triplet $(s, t, e)$, $0 \leq s < t < e \leq T$ and $z \in \mathbb{R}$, we define the CUSUM statistic as 
	\begin{align}\label{eq-cusum-def-ddd}
		& D^t_{s, e}(z) = \Bigg|\sqrt{\frac{2(e-t)}{n(e-s)(t-s)}}\sum_{k = s+1}^t \sum_{(i, j) \in \mathcal{O}} \mathbbm{1}\{\widehat{Y}^k_{ij} \leq z\} \nonumber \\
		& \hspace{3cm} - \sqrt{\frac{2(t-s)}{n(e-s)(e-t)}}\sum_{k = t+1}^e \sum_{(i, j) \in \mathcal{O}} \mathbbm{1}\{\widehat{Y}^k_{ij} \leq z\}\Bigg|,
	\end{align}
	and 
	\[
		D^t_{s, e} = \sup_{z \in [0, 1]}|D^t_{s, e}(z)|.
	\]

\item (Population version) With $\{A(t)\}_{t = 1}^T \subset \mathbb{R}^{n \times n}$, recall that  $P_t =  X(t) X(t)^{\top} $ and write 
	\[
		X(t) =  U_P(t)\Lambda_P(t)^{1/2},
	\]
	where $U_P(t) \in \mathbb{R}^{n \times d}$ is an orthogonal matrix with columns being the  $d$ eigenvectors of $P_t$ with largest  absolute eigenvalues, and $\Lambda_P(t) \in \mathbb{R}^{d \times d}$ is a diagonal matrix with entries being the leading $d$ eigenvalues of $P_t$.  
	
For any $t \in \{1, \ldots, T\}$ and $(i, j) \in \mathcal{O}$, let
	\begin{equation}\label{eq-def-Y}
		Y_{ij}^t = X_i(t)^{\top}X_j(t),
	\end{equation}
	where $X_i(t)^{\top}$ is the $i$th row of $X$.  For any integer triplet $(s, t, e)$, $0 \leq s < t < e \leq T$ and $z \in \mathbb{R}$, we define the CUSUM statistic as 
	\begin{align*}
		\widetilde{D}^t_{s, e}(z) & = \Bigg|\sqrt{\frac{2(e-t)}{n(e-s)(t-s)}}\sum_{k = s+1}^t \sum_{(i, j) \in \mathcal{O}} \mathbb{E}\left(\mathbbm{1}\{Y^k_{ij} \leq z\}\right) \\
		& \hspace{2cm} - \sqrt{\frac{2(t-s)}{n(e-s)(e-t)}}\sum_{k = t+1}^e \sum_{(i, j) \in \mathcal{O}} \mathbb{E}\left(\mathbbm{1}\{Y^k_{ij} \leq z\}\right)\Bigg|
	\end{align*}
	and 
	\[
		\widetilde{D}^t_{s, e} = \sup_{z \in [0, 1]}|\widetilde{D}^t_{s, e}(z)|.
	\]
\end{itemize}
\end{definition}

We remark that in Definition \ref{def-cusum-3}, if the $d$th and $(d+1)$th eigenvalues share the same value, then one can randomly pick an eigenvector to construct $\widehat{X}, X \in \mathbb{R}^{n \times d}$.  In addition, we do not require a specific order of the eigenvectors in constructing $\widehat{X}$ and $X$.

Recall that the distributions of the latent positions are equivalent up to a rotation, see Definition \ref{def-equivalence}.  To avoid extra efforts in matching the rotations when comparing two latent position distributions, we resort to the inner products of latent positions instead of latent positions itself.  We explain this via \eqref{eq-def-Y}.  For any orthogonal matrix $U \in \mathbb{R}^{d \times d}$, it holds that 
	\[
		Y_{ij}^t = (X_i(t))^{\top}X_j(t) = (UX_i(t))^{\top}UX_j(t).
	\]
	
With Definition \ref{def-cusum-3}, we arrive at our proposed procedure \Cref{algorithm:WBS} that builds on the wild binary segmentation algorithm \citep{fryzlewicz2014wild}.  The method requires first estimating the latent positions, a subroutine shown in \Cref{algorithm:PCA} \citep[adjacency spectral embedding, see e.g.][]{sussman2012consistent}.  Note that this only needs to be done once regardless of the choice of the tuning parameter $\tau$, and is parallelizable.  Since the complexity of the truncated principal component analysis is of order $O(d n^2)$, \Cref{algorithm:PCA} has the computational cost of order $O(T d n^2)$. Once the latent positions are estimated, we run the remaining steps in \Cref{algorithm:WBS}, which amounts to running Algorithm 2 in \cite{padilla2019optimal}.  For a fixed $\tau$ which leads to $\widetilde{K}$ change points, we have the computational complexity of order $O(\widetilde{K} M T n \log(n))$, which translates to $O(T d n^2 + \widetilde{K} MT n \log(n))$ for the overall cost of \Cref{algorithm:WBS}, where $M$ is the number of random intervals drawn in \Cref{algorithm:WBS}.

\begin{algorithm}[htbp!]
\begin{algorithmic}
	\INPUT Matrix $A \in \mathbb{R}^{n \times n}$, and  tuning parameter $d \in \mathbb{Z}_+$.
	\State $A = (v_1, \ldots, v_n) \mathrm{diag}(\lambda_1, \ldots, \lambda_n) (v_1, \ldots, v_n)^{\top}$, where $|\lambda_1| \geq \ldots \geq |\lambda_n|$.
	\State $X   \leftarrow     (v_1, \ldots, v_d)\mathrm{diag}(|\lambda_1|^{1/2}, \ldots, |\lambda_d|^{1/2})$
	\OUTPUT $X$
\caption{ScaledPCA $(A, d)$}
\label{algorithm:PCA}
\end{algorithmic}
\end{algorithm} 

\begin{algorithm}[htbp!]
\begin{algorithmic}
	\INPUT A sample $\{A(t)\}_{t = s+1}^{e} \subset \mathbb{R}^{n \times n}$, collection of intervals $\{ (\alpha_m,\beta_m)\}_{m=1}^M$, tuning parameters $d \in \mathbb{Z}_+$, and $\tau > 0$.
	\For{$t = s + 1, \ldots, e$}
		\State $X(t) \leftarrow \mathrm{ScaledPCA}(A(t), d)$
	\EndFor
	\For{$m = 1, \ldots, M$}  
		\State $(s_m, e_m) \leftarrow [s, e]\cap [\alpha_m, \beta_m]$
		\If{$e_m - s_m > 1$}
			\State $b_{m} \leftarrow \argmax_{s_m + 1 \leq t \leq e_m - 1}   D_{s_m, e_m}^{t}$
			\State $a_m \leftarrow D_{s_m, e_m}^{b_{m}}$
		\Else 
			\State $a_m \leftarrow -1$	
		\EndIf
	\EndFor
	\State $m^* \leftarrow \argmax_{m = 1, \ldots, M} a_{m}$
	\If{$a_{m^*} > \tau$}
		\State add $b_{m^*}$ to the set of estimated change points
		\State NonPar-RDPG-CPD$((s, b_{m*}),\{ (\alpha_m,\beta_m)\}_{m=1}^M, \tau)$
		\State NonPar-RDPG-CPD$((b_{m*}+1,e),\{ (\alpha_m,\beta_m)\}_{m=1}^M,\tau) $

	\EndIf  
	\OUTPUT The set of estimated change points.
\caption{NonPar-RDPG-CPD $((s, e), \{ (\alpha_m,\beta_m)\}_{m=1}^M, \tau)$}
\label{algorithm:WBS}
\end{algorithmic}
\end{algorithm} 

In every network, there are $n(n-1)/2$ observations, but note that in Definition \ref{def-cusum-3}, we in fact only use $n/2$ of them.  This is for technical convenience, since due to the choice of $\mathcal{O}$, we obtain independent observations within one network.  We acknowledge that there are other variants of this treatment.  For instance, instead using a fixed choice of $\mathcal{O}$, one can do multiple random sub-samplings and combine the results; one can also gather all the observations and create a $U$-statistic instead.  In \Cref{sec-possible-extensions}, we will show that in terms of rate, using $n/2$ edges  is as effective as using all possible edges. 

\section{Theory}
\label{sec:theory}

In this section, we provide the statistical guarantees for \Cref{algorithm:WBS}. In order to enhance the theoretical understanding, we take a step back and understand how the jump defined in \eqref{eq-jump-define} through the cumulative distribution functions of the inner products can be related to the jumps in terms of the distributions of the adjacency matrices.  The main results which provide theoretical guarantees of our algorithm are collected in \Cref{thm-main}.

\subsection{Characterizations of the changes} \label{sec:jumps}

We summarize the notation below and consider two different sets of models.

\begin{model}\label{mod-comp}
	We assume the following two independent models:
	\[
		\{A_{ij}, 1 \leq i <  j \leq n\} | \{X_i\}_{i = 1}^n \stackrel{\mathrm{ind}}{\sim} \mathrm{Ber}(X_i^{\top}X_j), \quad X_i \stackrel{\mathrm{ind}}{\sim} F \in \mathbb{R}^d;
	\]
	and
	\[
		\{\widetilde{A}_{ij}, 1 \leq i < j \leq n\} | \{\widetilde{X}_i\}_{i = 1}^n \stackrel{\mathrm{ind}}{\sim} \mathrm{Ber}(\widetilde{X}_i^{\top} \widetilde{X}_j), \quad \widetilde{X}_i \stackrel{\mathrm{ind}}{\sim} \widetilde{F} \in \mathbb{R}^d. 
	\]	
	For $i \neq j$, the cumulative distribution functions of $X_i^{\top} X_j$ and $\widetilde{X}_i^{\top} \widetilde{X}_j$ are denoted by $G(\cdot)$ and $\widetilde{G}(\cdot)$, respectively.  We further write $\mathcal{L}$ and $\widetilde{\mathcal{L}}$ for the joint unconditional distributions of $\{A_{ij}, 1 \leq i <  j \leq n\}$ and $\{\widetilde{A}_{ij}, 1 \leq i < j \leq n\}$, respectively.
\end{model}

The rest of this subsection is summarized in \Cref{fig-flowchart}.  The notation $A \Rightarrow B$ means $A$ implies $B$. 

\begin{figure}
\centering
\begin{tikzpicture}[roundnode/.style={circle, draw=Fuchsia!60, fill=Fuchsia!5, very thick, minimum size=7mm}, squarednode/.style={rectangle, draw=CarnationPink!60, fill=CarnationPink!5, very thick, minimum size=7mm}, ball/.style={ellipse, minimum width=2cm, minimum height=1cm, draw=Fuchsia!60, fill=Fuchsia!5, very thick}]

%Nodes
\node [squarednode] (G) {$G \neq \widetilde{G}$};
\node [squarednode] (F) [left of = G, xshift = 6cm] {$F \neq \widetilde{F}$};
\node [squarednode] (L) [left of = F, xshift = 6cm, yshift = -1.5cm] {$\mathcal{L} \neq \widetilde{\mathcal{L}}$};
\node [squarednode] (F2) [below of = F, yshift = -2cm] {$\substack{\mbox{First } n-1 \mbox{ moments of}\\ F \mbox{ and } \widetilde{F} \mbox{ are not identical}}$};

\node  [ball] (l1) [left of = G, xshift = 3.5cm, yshift = -0.7cm] {Lemma \ref{lem-ff-gg}};
\node  [ball] (l3) [left of = F, xshift = 3.5cm, yshift = -0.7cm] {Lemma \ref{lem-one-network}};
\node  (a1) [left of = F, xshift = 3.5cm, yshift = 0.5cm] {$+$        \Cref{assume-comp}};
\node [ball, below of = L, yshift = -0.5cm, rotate = 10, anchor = north, xshift = -1cm]{Lemma \ref{lem-dist-moment}};

%Lines
\draw [double, ->, thick] (G.east) to [out= 0, in=180] (F.west);
\draw [double, ->, thick] (F.east) to [out= 0, in=90] (L.north);
\draw [double, <->, thick] (L.south) to [out= 270, in=0] (F2.east);

\end{tikzpicture}
\caption{Flowchart of \Cref{sec:jumps}.  The notation $A \Rightarrow B$ means $A$ implies $B$. }\label{fig-flowchart}
\end{figure}
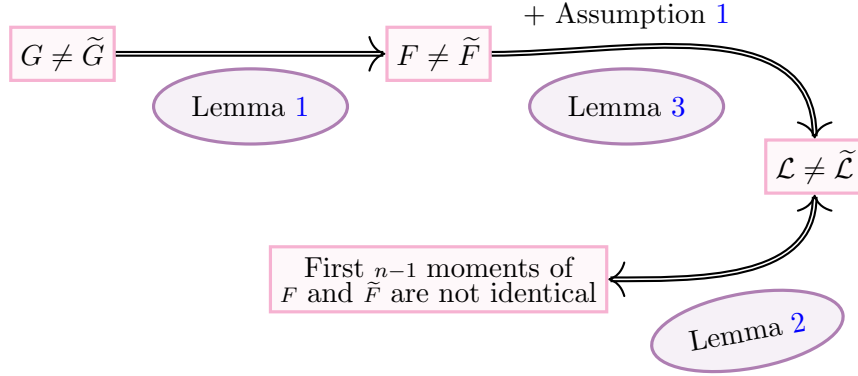

\begin{lemma}\label{lem-ff-gg}
	With the notation in \Cref{mod-comp}, if $F = \widetilde{F}$, then $G = \widetilde{G}$.
\end{lemma}

This follows automatically from the definitions, and is equivalent to the claim that if $G \neq \widetilde{G}$ then $F \neq \widetilde{F}$, which implies that \eqref{eq-jump-define} is equivalent to 
	\[
		F_{\eta_k} \neq F_{\eta_{k-1}}, \quad k \in \{1, \ldots, K\}.
	\]

However, $F \neq \widetilde{F}$ does not imply $\mathcal{L} \neq \widetilde{\mathcal{L}}$.  As a simple toy example, consider $F$ and $\widetilde{F}$ to be defined in  Definition \ref{def-1}, with the same mean but different variances, and $n = 2$.  Then $F \neq \widetilde{F}$ but $\mathcal{L} = \widetilde{\mathcal{L}}$.  \ref{lem-dist-moment} below shows that $\mathcal{L}$ is determined by the first $n-1$ moments of $F$.

\begin{lemma}\label{lem-dist-moment}
	Under \Cref{mod-comp}, we have that $\mathcal{L} = \widetilde{\mathcal{L}}$ if and only if there exists an orthogonal operator $U \in \mathbb{R}^{d \times d}$, such that if $d = 1$, 
	\[
	\mathbb{E}_F(X_1^k) = \mathbb{E}_{\widetilde{F}}\{(U\widetilde{X}_1)^k\}, \quad k = 1, \ldots, n-1,
	\]
	if $d > 1$
	\[
	\mathbb{E}_F\left(\prod_{l = 1}^d X_{1, l}^{k_l}\right) = \mathbb{E}_{\widetilde{F}}\left\{\prod_{l = 1}^d(U\widetilde{X}_1)_l^{k_l}\right\}, \quad k_l \in \mathbb{Z}, \quad k_l \geq 0, \quad \sum_{l = 1}^d k_l = k, \quad k = 1, \ldots, n-1,
	\]
	where $X_{1, l}$ and $(U\widetilde{X}_1)_l$ are the $l$th coordinates of the $X_1$ and $U\widetilde{X}_1$.
\end{lemma}

It can be seen from  Lemma \ref{lem-dist-moment} that the unconditional distribution of the data matrix is determined by the first $n-1$ moments of the underlying distribution $F$.  Unfortunately, without additional assumptions, the first $n-1$ moments do not determine the distribution \citep[e.g.][]{heyde1963property}\footnote{We are grateful to Richard J. Samworth for this reference and constructive discussions.}.  This means that only assuming \eqref{eq-jump-define} can not guarantee that the data matrices $A$  and $\tilde{A}$ have different distributions.  

The final claim we make in this subsection is that under some additional but weak conditions, we will be able to guarantee that  $\mathcal{L} \neq \widetilde{\mathcal{L}}$.

\begin{assumption}\label{assume-comp}
	Under \Cref{mod-comp}, let 
	\[
	\kappa_0 = \sup_{z \in [0, 1]} |G(z) - \widetilde{G}(z)|.
	\]
	It holds that
	\[
	\kappa_0\sqrt{n} > 3 \sqrt{\log(n)}.
	\]	
\end{assumption}

\begin{lemma}\label{lem-one-network}
	Assume that \Cref{mod-comp} and \Cref{assume-comp} hold. Then we have that 
	\[
	\mathcal{L} \neq \widetilde{\mathcal{L}}.
	\]
\end{lemma}

Lemma \ref{lem-one-network} suggests that under \Cref{assume-comp}, $G \neq \widetilde{G}$ implies $\mathcal{L} \neq \widetilde{\mathcal{L}}$.  This enhances the rationale of imposing the distributional changes occurring at the change points on the differences on $G$, as detailed in \Cref{assume:model-rdpg}(4).  \Cref{assume-comp} is a weak assumption, which will be further elaborated in \Cref{sec:main_theorem}.  The proofs of Lemmas~\ref{lem-dist-moment} and \ref{lem-one-network} are collected in \Cref{sec-app-jumps}.

\subsection{Consistent estimation of change points}
\label{sec:main_theorem}

We first state a signal-to-noise ratio condition below.

\begin{assumption}[Signal-to-noise ratio]\label{assume-snr}
There exists a universal constant $C_{\mathrm{SNR}} > 0$, such that there exists a diverging sequence $a_T \to \infty$, as $T \to \infty$, satisfying
	\[
		\kappa\sqrt{\Delta n (1 - \rho)} >  C_{\mathrm{SNR}} \sqrt{T} \max\{\sqrt{d\log(n \vee T)}, \, d^{3/2}\} + a_T.
	\]
\end{assumption}

To better understand \Cref{assume-snr}, we would like to use Assumptions 2 and 3 in \cite{wang2018optimal} as benchmarks, since \cite{wang2018optimal} studied a simpler problem assuming independence within and across networks, and showed a phase transition phenomenon in the minimax sense.  However, we would like to emphasize that comparing \Cref{assume-snr} and Assumptions 2 and 3 in \cite{wang2018optimal} is comparing apples and oranges, to some extent.   Even though the jump size $\kappa$ are defined differently in these two papers, both take values in $(0, 1]$.  The parameter $\rho$ in this paper indicates the correlation between networks, while the parameter $\rho$ in \cite{wang2018optimal} represents the entrywise sparsity.  For simplicity, we let $\rho = 1$ in \cite{wang2018optimal} for this discussion.	

One key difference is that in \Cref{assume-snr}, the required signal-to-noise ratio is inflated by $\sqrt{1-\rho}$.  We might view  this as the effective  sample size  being   shrunk from $\Delta$ to $(1-\rho)\Delta$, due to the dependence across time.  In \Cref{assume:model-rdpg}, we do not allow $\rho = 1$, but allow $\rho \to 1$, as long as \Cref{assume-snr} holds. In the extreme case that $\rho = 1$, between two consecutive change points, there is essentially only one observation.  As long as \Cref{assume-comp} holds, Lemma \ref{lem-one-network} shows that the distributions of the adjacency matrices before and after change points are different, which implies that one can identify the change points with probability 1.

Another difference is that in our paper, the signal-to-noise ratio is inflated by $\sqrt{T}$ compared to \cite{wang2018optimal}.  This is due to the fact that we estimate the latent positions separately for every single network, while the graphons were estimated based on a version of sample average of the adjacency matrices in \cite{wang2018optimal}.  The reason we estimate the positions separately roots in the difficulty of deriving theoretical properties of eigenvectors of a sample average matrices.  More discussions on this can be found in \Cref{sec-possible-extensions}.

We allow the dimensionality $d$ to diverge, provided that \Cref{assume-snr} holds.  The dimensionality $d$ is essentially the low rank condition imposed in \cite{wang2018optimal}.  The upper bound on the rank $r$ in \cite{wang2018optimal} comes into play with the term $\sqrt{r}$, while we have $d^{3/2}$ here.  The difference again is rooted in the estimation of the latent positions, although we do not claim optimality here.  

The sequence $a_T$ can diverge at any arbitrarily slow rate.  We will explain the role of $a_T$ after we state \Cref{thm-main}.

Finally, we make connections between Assumptions~\ref{assume-comp} and \ref{assume-snr}.  Recall that we use \Cref{assume-comp} in Lemma \ref{lem-one-network}, where only one observation is available for each distribution, i.e.~$\Delta = 1$, $\rho = 0$ and $T = 2$.  Ignoring the universal constants, the only difference left between Assumptions~\ref{assume-comp} and \ref{assume-snr} is the term $d^{3/2}$.  Of course, if $d = O(1)$, then this is also a universal constant, and there is no difference left.  The interesting thing happens when $d$ is allowed to diverge faster than the poly-logarithm term.  \Cref{assume-comp} is required to differentiate two different distributions, which roughly speaking is related to a testing task; while \Cref{assume-snr} is used below in \Cref{thm-main} with the purpose of consistent localization, which is an estimation problem.  To this end, the extra $d^{3/2}$ in \Cref{assume-snr} is a piece of evidence that estimation is a harder problem than a testing one.

\begin{theorem}\label{thm-main}
	Let data be from \Cref{assume:model-rdpg} and satisfy \Cref{assume-snr}. Assume the following.
	\begin{itemize}
	\item The tuning parameter $\tau$ in \Cref{algorithm:WBS} satisfies	
		\begin{equation}\label{eq-tau-cond}
			c_{\tau, 1} T^{1/2} \max\{\sqrt{d\log(n \vee T)}, \, d^{3/2}\} < \tau < c_{\tau, 2} \kappa\sqrt{\Delta n (1 - \rho)},
		\end{equation}
		where $c_{\tau, 1}, c_{\tau,2} > 0$ are universal constants depending on all the universal constants in \Cref{assume:model-rdpg} and \Cref{assume-snr}.
	\item The tuning parameter $d$ in Algorithms~\ref{algorithm:PCA} and \ref{algorithm:WBS} are the true dimension $d$ of the latent positions. 
	\item The intervals satisfy 
		\begin{equation}\label{eq-thm1-statement-cr}
			\max_{m = 1, \ldots, M} (\alpha_m - \beta_m) \leq C_R \Delta,
		\end{equation}
		where $C_R > 3/2$ is a universal constant.
	\end{itemize}
	
Let $\{\widehat{\eta}_k\}_{k = 1}^{\widehat{K}}$ be the output of \Cref{algorithm:WBS}.  We have that
	\begin{align*}
		& \mathbb{P}\left\{\widehat{K} = K, \quad |\widehat{\eta}_k - \eta_k| \leq C_{\epsilon} \frac{T \max\{d\log(n \vee T), \, d^3\}}{\kappa_k^2 n (1 - \rho)}, \, \forall k \right\} \\
		& \hspace{2cm}\geq 1 - C(n \vee T)^{-c} - CTe^{-n} - \exp\left(\log(T/\Delta) - (4C_R)^{-1}T^{-1}M\Delta\right),
	\end{align*}
	where $C, c >0$ are universal constants depending only on the other universal constants.
\end{theorem}

The proof of \Cref{thm-main} can be found in \Cref{sec-app-main}, following two sets of lemmas -- technical details on estimating the latent positions and on change point analysis, collected in Appendices~\ref{sec-app-prop} and \ref{sec-app-cpd}, respectively.

Suppose that $Te^{-n} \to 0$ and that $M$ satisfies
	\begin{equation} \label{choice_of_m}
	\frac{T/\Delta \log\left(T/\Delta\right)}{M} \to 0.
	\end{equation}
Then	it can be seen from \Cref{thm-main} that with probability tending to 1, as $T$ diverges, we have that $\widehat{K} = K$ and 
	\begin{align*}
		\max_{k = 1, \ldots, K}\frac{|\widehat{\eta}_k - \eta_k|}{\Delta} \leq \max_{k = 1, \ldots, K}  C_{\epsilon} \frac{T \max\{d\log(n \vee T), \, d^3\}}{\Delta \kappa_k^2 n (1 - \rho)}  \leq  C_{\epsilon} \frac{T \max\{d\log(n \vee T), \, d^3\}}{\Delta \kappa^2 n (1 - \rho)} \to 0,
	\end{align*}
	where the second inequality follows from the definition of $\kappa$ and the convergence follows from \Cref{assume-snr}, with the aid of an arbitrarily diverging sequence $a_T$.  This implies that the change point estimators we obtain are consistent, with a vanishing localization rate.  Since the quantity $M$ only appears in the probability, we remark that the larger $M$ is, the more likely that our estimators would perform satisfactorily, while the higher the computational cost is.
	
\Cref{algorithm:WBS} in fact can handle networks of varying size. For instance, if we do not allow for the dependence across time, then \Cref{thm-main} holds provided that all network sizes are of the same order, which amounts to $c_1 n \leq n_t \leq c_2 n$, $t = 1, \ldots, T$, for universal constants $c_1, c_2 > 0$.

In view of \Cref{thm-main}, the two most important tuning parameters in \Cref{algorithm:WBS} are the threshold $\tau$ and the dimension $d$.

The threshold $\tau$ is set to satisfy \eqref{eq-tau-cond}.  The upper and lower bounds in \eqref{eq-tau-cond} are the lower bound on the signals and the upper bound on the noise, in a large probability event, respectively.  If the input $\tau$ is larger than the upper bound in \eqref{eq-tau-cond}, then one may not be able to detect all change points; while if the input is smaller than the lower bound, then there is the risk of falsely detecting change points.  Note that both the upper and lower bounds involve unknown constants.  We will discuss the practical guidance in choosing $\tau$ in \Cref{sec:experiments}.
	
In \Cref{thm-main}, we assume that the input $d$ should be the true dimension.  This is a seemingly strong condition.  We would like to comment on this from a few different angles.
	\begin{itemize}
	\item In the context of stochastic block models, which are simpler than the RDPG models, the parameter $d$ is a lower bound on the number of communities.  To estimate the number of communities in a stochastic block model is yet open, despite a tremendous amount of efforts \citep[e.g.][]{bickel2016hypothesis, lei2016goodness, chen2018network, li2016network, saldana2017many}.  We do not intend to propose a method to estimate the dimensionality here, but in practice, one could resort to the aforementioned papers.  
	\item Without a theoretically-justified method to estimate $d$, we need to discuss on the potential misspecification.  If one overestimates $d$, i.e.~with an input $d_1 > d$, then our method can still consistently estimate the change points under \Cref{assume-snr}, with a sufficiently large constant $C_{\mathrm{SNR}}$.  This is due to the fact our statistic is a function of inner products of latent position estimators.  Overestimating $d$ will only add extra noise which is in fact of the same order of the noise introduced when estimating the latent positions with true dimension $d$.  
	\item Another possible misspecification is underestimating the dimension $d$, i.e.~the input of the algorithms is $d_2 < d$.  This is a more damning issue than overestimating $d$, however it does not necessarily lead to inconsistent change point estimators.  Now we assume a toy example where the true dimension $d = 3$.  Recall the definition on the jump size $\kappa$ that
		\begin{align*}
			\kappa & = \min_{k = 1, \ldots, K} \sup_{z \in [0, 1]} |G_{\eta_k}(z) - G_{\eta_{k-1}}(z)| \\
			& = \min_{k = 1, \ldots, K} \sup_{z \in [0, 1]} \left|\mathbb{P}_{\eta_k}\left\{X^{\top} Y \leq z\right\} - \mathbb{P}_{\eta_{k-1}}\left\{X^{\top} Y \leq z\right\}\right| \\
			& = \min_{k = 1, \ldots, K} \sup_{z \in [0, 1]} \left|\mathbb{P}_{\eta_k}\left\{ \sum_{i = 1}^3 X_i Y_i \leq z\right\} - \mathbb{P}_{\eta_{k-1}}\left\{\sum_{i = 1}^3 X_i Y_i \leq z\right\}\right|.
		\end{align*}
		If we underestimate $d$ and we miss out the third dimension, our \textit{de facto} jump size becomes
		\[
			\kappa_1 = \min_{k = 1, \ldots, K} \sup_{z \in [0, 1]} \left|\mathbb{P}_{\eta_k}\left\{ \sum_{i = 1}^2 X_i Y_i \leq z\right\} - \mathbb{P}_{\eta_{k-1}}\left\{\sum_{i = 1}^2 X_i Y_i \leq z\right\}\right|.
		\]
		Provided that the signal-to-noise ratio condition holds for $\kappa_1$, i.e.
		\[
			\kappa_1\sqrt{\Delta n (1 - \rho)} >  C_{\mathrm{SNR}} \sqrt{T} \max\{\sqrt{\log(n \vee T)}, \, d^{3/2}\} + a_T,
		\]
		with the notation defined in \Cref{assume-snr}, \Cref{thm-main} still holds.  In general, underestimating the dimension $d$ decreases the true jump sizes at the change points.  Identical arguments to those in \Cref{thm-main} can lead to consistent detection of change points, whose decreased jump sizes satisfy the signal-to-noise ratio condition \Cref{assume-snr}.
	\end{itemize}
	
On a different note, without assuming \eqref{eq-thm1-statement-cr}, and using the trivial bound $C_R \leq T/\Delta$, it can be shown that we will achieve a larger localization error. The resulting rate inflates  that  of \Cref{thm-main} by a factor of polynomials of $T/\Delta$.  

\subsection{Possible extensions}\label{sec-possible-extensions}

There are three aspects of the methods proposed in \Cref{sec-methods} that might not seem to be satisfactory at first sight.  In this subsection, we discuss possible extensions.  Readers who are not familiar with the area, may safely skip this subsection during the first time reading.

\medskip
\noindent \textbf{From dense to sparse networks}

The networks we are dealing with in this paper are dense, i.e.~the average degrees are of order of the network size.  In order to allow for sparse networks, one might wish to replace \eqref{eq-def-rdpg-1} in Definition \ref{def-1} with the following
	\[
		\mathbb{P}\left\{A \mid X\right\} = \prod_{1 \leq i < j \leq n} (\alpha X_i^{\top}X_j)^{A_{ij}}(1 - \alpha X_i^{\top}X_j)^{1 - A_{ij}},
	\]
	where $\alpha = \alpha(n) \in (0, 1]$.  
	
If $\alpha$ is known, then one could simply replace the definition of $\widehat{X}$ in Definition \ref{def-cusum-3} with 
	\[
		\widehat{X}(t) =  \alpha^{-1/2} U_A(t)\Lambda_A(t)^{1/2}.
	\]
	The signal-to-noise ratio and the localization errors will change correspondingly by a polynomial factor of $\alpha$, following the identical derivations.  
	
If $\alpha$	is unknown but satisfies $\alpha n \gtrsim \log(n)$, then one could use graphon estimation methods, e.g.~the universal singular value thresholding (USVT) method \citep{chatterjee2015matrix}, to first produce a USVT estimator of each $P(t)$, namely $\widehat{A}(t)$.  The quantity $\widehat{X}(t)$ can be defined to be  
	\[
		\widehat{X}(t) =  U_{\widehat{A}}(t)\Lambda_{\widehat{A}}(t)^{1/2},
	\]
	and the rest of the algorithm remains the same.  The localization rate would change from 
	\[
		\frac{T \max\{d\log(n \vee T), \, d^3\}}{\kappa^2 n (1 - \rho)} \quad \mbox{to} \quad \frac{T \max\{\sqrt{n /\alpha}, \, d^3\}}{\kappa_k^2 n (1 - \rho)},
	\]
	by simply using Theorem~1 in \cite{xu2018rates} instead of Lemma \ref{lem-three-e} in controlling large probability events and following all the rest of our proofs.  The term $d\log(n \vee T)$ in the upper bound is due to the fact that conditional on the latent positions, the entries in the upper triangular matrix of $A(t)$'s are independent.  This is not true for the USVT estimator $\widehat{A}$, and therefore the difference between the two different rates is not merely multiplying a polynomial factor of the sparsity parameter $\alpha$. 

\medskip
\noindent \textbf{From individual estimation to a bulk estimation}

In \Cref{algorithm:WBS}, we estimate every individual network separately, which results in the polynomial dependence on $T$ in both the signal-to-noise ratio and the localization rate.  One natural question would be if there is a way to conduct \Cref{algorithm:PCA} to a bulk of adjacency matrices at once in order to improve the statistical accuracy and computational efficiency.

There are two possible extensions.  One is to use the omnibus embedding proposed in \cite{levin2017central2} and the other is to conduct \Cref{algorithm:PCA} to a sample average of the adjacency matrices.  Either way is suffered from the lack of some critical theoretical understanding.  When the bulk of adjacency matrices used to construct either the omnibus matrix or the sample average matrix, are not generated from the same set of latent positions, the behaviours of the sample eigenvectors remain unknown.      In change point detection, one needs to deal with intervals containing adjacency matrices coming from different latent positions.  Without knowing how the eigenvectors would behave, eigenvector-based change points detection methods would not be able utilize bulk of adjacency matrices.  In fact, this is also the reason that the methods proposed in \cite{CribbenYu2017} and \cite{LiuEtal2018} lack theoretical guarantees.

\medskip
\noindent \textbf{From using $n/2$ edges to all edges}

In \Cref{algorithm:WBS}, we only use $n/2$ out of $n(n-1)/2$ edges for technical convenience.  In our choice, all the edges are conditionally independent given the latent positions and the concentration inequalities are easier to handle.  

One extreme is to use all possible edges such that $D^t_{s, e}(z)$ defined in Definition \ref{def-cusum-3} is a $U$-statistic.  To be specific, \eqref{eq-cusum-def-ddd} is replaced by 
	\begin{align*}
		D^t_{s, e}(z) & = \Bigg|\sqrt{\frac{2(e-t)}{n(n-1)(e-s)(t-s)}}\sum_{k = s+1}^t \sum_{i = 1}^{n-1}\sum_{j = i+1}^n  \mathbbm{1}\{\widehat{Y}^k_{ij} \leq z\} \\
		& \hspace{2cm} - \sqrt{\frac{2(t-s)}{n(n-1)(e-s)(e-t)}}\sum_{k = t+1}^e \sum_{i = 1}^{n-1}\sum_{j = i+1}^n \mathbbm{1}\{\widehat{Y}^k_{ij} \leq z\}\Bigg|.
	\end{align*}
	Using the Hoeffding theorem \citep[Theorem~5.2 in][]{hoeiffding1948class}, we can see that in our cases, the variance of using all edges and that of only using $n/2$ edges are of the same order.  This means that using all edges will not improve the statistical accuracy (in terms of rates) but creates extra computational burden.

\medskip  		
\noindent \textbf{Generalised random dot product graph}

The random dot product graph models have a generalisation, namely generalised random dot product graph models \citep[GRDPG,][]{rubin2017statistical}.  Recalling \Cref{assume:model-rdpg}, GRDPG assumes that
	\[
		\mathbb{P}\left\{A(t) \mid X(t)\right\} = \prod_{1 \leq i < j \leq n} (X_i(t)^{\top}I_{p, q}X_j(t))^{A_{ij}(t)}(1 - X_i(t)^{\top} I_{p, q}X_j(t))^{1 - A_{ij}(t)},
	\]
	where $I_{p, q} = \mathrm{diag}(1, \ldots, 1, -1, \ldots, -1)$ with $p$ ones and $q$ minus ones.  We remark that the algorithms and theoretical results developed in this paper for RDPGs also hold for GRDPGs.

\section{Numerical Experiments} \label{sec:experiments}

\subsection{Simulations}\label{sec:simulations}

We now assess the performance of our proposed estimator NonPar-RDPG-CPD (\Cref{algorithm:WBS}) in different scenarios, and compare our results with those produced by the network binary segmentation (NBS) algorithm \citep{wang2018optimal} and the modified neighbourhood smoothing (MNBS) algorithm \citep{zhao2019change} \footnote{Code implementing our method can be found in \url{https://github.com/hernanmp/RDPG}.  The algorithms are now included in the R package \texttt{changepints} \citep{r-changepoint-package}.}.  The measurements we adopt are the absolute error $|\widehat{K} - K|$, where $\widehat{K}$ and $K$ are the numbers of the change point estimators and the true change points, respectively, and the one-sided Hausdorff distance defined as
	\[
		d(\widehat{\mathcal{C}}| \mathcal{C}) = \max_{\eta \in \mathcal{C} } \min_{ x\in \widehat{\mathcal{C}}   } |x - \eta|,
	\]
	where $\mathcal{C}$ is the set of true change points, and $\widehat{\mathcal{C}}$ is the set of estimated change points. We also consider the metric $d(\mathcal{C}|\widehat{\mathcal{C}})$.  For Hausdorff distances, we report the medians over 100 Monte Carlo simulations, and for $|\widehat{K} - K|$, we report the means over 100 Monte Carlos trials. By convention, if $\widehat{\mathcal{C}} = \emptyset$, we define $d(\widehat{\mathcal{C}}| \mathcal{C}) = \infty$ and $d(\mathcal{C}|\widehat{\mathcal{C}}) = -\infty$.

\textbf{Choice of tuning parameters.}  Recall that NonPar-RDPG-CPD involves three tuning parameters: (1) the threshold $\tau$ for declaring change points, (2) the number of random intervals $M$ and (3) the dimension of the embedding $d$.  
\begin{itemize}
\item [(1)]	 We choose $\tau$ based on the model selection criteria from \cite{zou2014nonparametric}.  To be specific, we stack all the $\widehat{Y}_{ij}^t$ into one matrix $\widehat{Z} \in \mathbb{R}^{T \times n/2}$.  For any given $\tau$, we let $\{\hat{\eta}_k, k = 0, \ldots, \widehat{K}(\tau)+1\}$ be the set of change points estimated by NonPar-RDPG-CPD, with $\hat{\eta}_0(\tau)=0$ and $\hat{\eta}_{\widehat{K}(\tau)+1}(\tau)=T$.  Define
	\begin{align*}
		\mathrm{BIC}_{\tau} & =  \sum_{l=1}^{n/2}  \Bigg[\sum_{k=0}^{ \hat{K}(\tau)  } \sum_{t= 2}^{T-1}  \{\hat{\eta}_{k+1}(\tau) -  \hat{\eta}_{k} (\tau)\}  \frac{\widehat{H}_{k tl}  \log (\widehat{H}_{ktl})  +  (1-\widehat{H}_{ktl}) \log (1- \hat{H}_{ktl}  ) }{ t(T-t)}  \\
		& \hspace{1cm}+ \xi \widehat{K}(\tau)  \Bigg],
	\end{align*}
	where $\widehat{H}_{ktl} =  \widehat{H}_{ \hat{\eta}_k(\tau) :  \hat{\eta}_{k+1}(\tau)  }^{ l   }(\widehat{Z} _{t,l} )$ with  $\widehat{H}_{ \hat{\eta}_k(\tau) :  \hat{\eta}_{k+1}(\tau)  }^{ l   }$ being the empirical cumulative distribution function of the observations  $\{ \widehat{Z}_{t,l}\}_{\hat{\eta}_k(\tau) \leq   t\leq \hat{\eta}_{k+1}(\tau) -1}$ and $\xi = \log^{2.1}(n)/5$. The metric $\text{BIC}_{\tau}$ is  constructed  by calculating along  each column of $\widehat{Z} $ the BIC-type scores defined in Equation (2.4) in \cite{zou2014nonparametric}, and then aggregating the scores to produce $\text{BIC}_{\tau}$. We select the model with $\tau$ that minimizes $   \text{BIC}_{\tau}$. 
\item [(2)] As for input representing the dimension of the latent positions $d$, we set $d = 10$ throughout this section, with varying $d$ scenarios discussed in \Cref{sec:sensitivity}.  In general, we find the procedure very robust to the choice of $d$ provided it is no smaller than the true dimension of the latent positions.  This supports our discussions on misspecification after \Cref{thm-main}. 	
\item [(3)] We also set $M = 120$, which is large enough for the various settings considered to perform well.
\end{itemize}

As for the competitor NBS, we follow the proposal by the authors in \cite{wang2018optimal} setting $\tau$ to be of order $n\log^2(T)$.  For the competitor MNBS, we use the default  choice of its tuning parameters with code generously provided by the authors of \cite{zhao2019change}.  
	To be specific, the scaling parameter for the threshold is set as $D_0 = 0.25$, the constant $B_0$ for the neighborhood size is chosen as $B_0 = 3$, the threshold size is set as $\delta_0 =0.1$ and the local window size is set as $h =\sqrt{T}$.

\textbf{Disclaimer}: We would like to emphasize that the comparisons to the competitors might not be fair, due to the fact that the tuning parameter choosing schemes in \cite{zhao2019change} and \cite{wang2018optimal} are not meant for dependent networks.

We construct four different models, in each of which, $T = 150$ and $K= 2$.  The locations of the change points are evenly spaced, giving rise to three  disjoint intervals $\mathcal{A}_1 = [1, 50]$, $\mathcal{A}_2 = [51, 100]$ and $\mathcal{A}_3 = [101,150]$.  As for the sizes of networks, we consider $n \in \{100, 200, 300\}$.

\paragraph{Scenario  1.} \textbf{Stochastic block models}.  We construct two matrices of probabilities, $P, Q \in \mathbb{R}^{n \times n}$.  The matrix $P$ satisfies 
	\[
		P_{i,j} = \begin{cases}
			0.5,  & i, j \in \mathcal{B}_l, \, l \in \{1, \ldots, 4\},\\
			0.3,  & \text{otherwise},
		\end{cases}
	\]	
	where $\mathcal{B}_1, \ldots, \mathcal{B}_4$ are evenly sized communities of nodes that form a partition of $\{1, \ldots, n\}$.  The matrix $Q$ satisfies
	\[
		Q_{i,j} = \begin{cases}
			0.45,  & i,j \in \mathcal{B}_l, \, l \in \{1, \ldots, 4\},\\
			0.2,  & \text{otherwise}.
		\end{cases}
	\]
	We then construct a sequence of matrices $\{E(t)\}_{t=1}^T \subset \mathbb{R}^{n \times n} $ such that
	\[
		E_{i, j}(t) = \begin{cases}
			P_{i,j},  & t\in \mathcal{A}_1 \cup \mathcal{A}_3,\\
			Q_{i,j},  & \text{otherwise},
		\end{cases}
	\]
	for every $i, j \in \{1, \ldots, n\}$.
 
The data are then generated with a correlation parameter $\rho \in \{0, 0.5, 0.9\}$.  Specifically, for any $\rho$, we have  $A_{i,j}(1) \sim \mathrm{Ber}(P_{i,j}(1))$, and between two consecutive change points,
	\[
		A_{i,j}(t+1) \sim \begin{cases}
			\mathrm{Ber}((1-  E_{i,j}(t+1)) \rho + E_{i,j}(t+1)), & A_{i,j}(t) = 1,\\
			\mathrm{Ber}((E_{i,j}(t+1))(1-\rho)),  & A_{i,j}(t)= 0,\\
		\end{cases}
	\] 
	for  $1 \leq i <  j \leq  n$.

\paragraph{Scenario 2.} We first generate 
	\[
		X_i(t) \stackrel{\mbox{ind}}{\sim} \text{Uniform}[0.2, 0.8],  \quad i = 1, \ldots, n, \, t \in \mathcal{A}_1 \cup  \mathcal{A}_3.
	\]
	Then for any $\varepsilon \in \{0.05, 0.15, 0.3\}$, we generate
	\[
		X_i(t) = \begin{cases}
			Z_i(t) + 0.2, &  i \in\{1, \ldots, \lfloor n \varepsilon \rfloor \}, \\
			Z_i(t), & \text{otherwise},
		\end{cases} 
	\]
	where  $Z_i(t) \stackrel{\mbox{ind}}{\sim} \text{Uniform}[0.2, 0.8]$ for $i \in \{1, \ldots, n\}$ and $t \in \mathcal{A}_2$. Then we generate $A_{i,j}(t) \sim \text{Ber}(X_i(t)X_j(t ))$. 

\paragraph{Scenario 3.} For $t \in \{1, 101\}$, we generate $Z_i(t) \stackrel{\mbox{ind}}{\sim} \mathcal{N}(0, I_3)$, and for  $t \in \mathcal{A}_1 \cup \mathcal{A}_3 \backslash \{1, 101\}$, we generate
	\[
		Z_i(t)\begin{cases}
		     	\stackrel{\mathrm{ind}}{\sim}	\mathcal{N}(0,I_3), & \text{with probability } 0.9,\\
		=	Z_i(t-1), & \text{with probability } 0.1. 
		\end{cases} 
	\]
	We then set
	\[
		P_{i,j}(t) = \frac{\exp\left\{Z_i(t)^{\top}Z_{j}(t)\right\}}{1 + \exp\left\{Z_i(t)^{\top}Z_{j}(t)\right\}}.
	\]
	Furthermore, we generate $P_{i,j}(51) \sim \text{Beta}(100, 100)$, and for $t \in \{52, \ldots, 100\}$ we generate
	\[
		P(t) \begin{cases}
			= P(t-1),  &  \text{with probability } 0.9,\\
			\sim \text{Beta}(100,100), & \text{with probability } 0.1.\\
		\end{cases}
	\]
	Once the mean matrices $\{P(t)\}_{t=1}^T \mathbb{R}^{n\times n}$ have been constructed, we independently draw $A_{i,j}(t) \sim \text{Ber}(P_{i, j}(t))$, for all $i, j \in \{1, \ldots, n\}$ and $t \in \{1, \ldots, T\}$.

\paragraph{Scenario 4.}  For $t \in \{1,101\}$ we generate $X_t \in \mathbb{R}^5$ as
	\[
		X_i(t) \sim  \text{Dirichlet}(1,1,1,1,1), 
	\]
	for all $i \in \{1, \ldots, n\}$.  Then for $t \in \mathcal{A}_1 \cup \mathcal{A}_3 \backslash \{1, 101\}$, 
	\[
		X_i(t) \begin{cases}
			= X_i(t-1),  & \text{with probability } 0.9, \\
			\sim \text{Dirichlet}(1,1,1,1,1) & \text{otherwise,}  \\
		\end{cases} 
	\]  
	for all $i \in \{1, \ldots, n\}$. We also have 
	\[
		X_i(51) \sim \begin{cases}
			\text{Dirichlet}(500,500,500,500,500), & i \in \{1,\ldots, \lfloor n \varepsilon \rfloor\}, \\
			\text{Dirichlet}(1,1,1,1,1), & i \in \{\lfloor n \varepsilon  \rfloor +1 , \ldots, n \}, \\
		\end{cases} 
	\]   
	and for $t \in \mathcal{A}_2 \backslash\{51\}$,
	\[
		X_i(t) \begin{cases}
			= X_i(t-1),  &  \text{with probability } 0.9,   \\
			\sim \text{Dirichlet}(500,500,500,500,500), & \text{with probability } 0.1\, \text{if } i \in \{1,\ldots, \lfloor n \varepsilon  \rfloor \}, \\
			\sim \text{Dirichlet}(1,1,1,1,1), & \text{with probability } 0.1, \,\text{if } i \in \{\lfloor n \varepsilon  \rfloor + 1, \ldots, n \}, \\
		\end{cases} 
	\]  
	for all $i \in \{1, \ldots, n\}$, where $\varepsilon \in \{0.05,0.15,0.3\}$.

\begin{figure}[t!]
	\begin{center}
		\includegraphics[width = 0.45\textwidth]{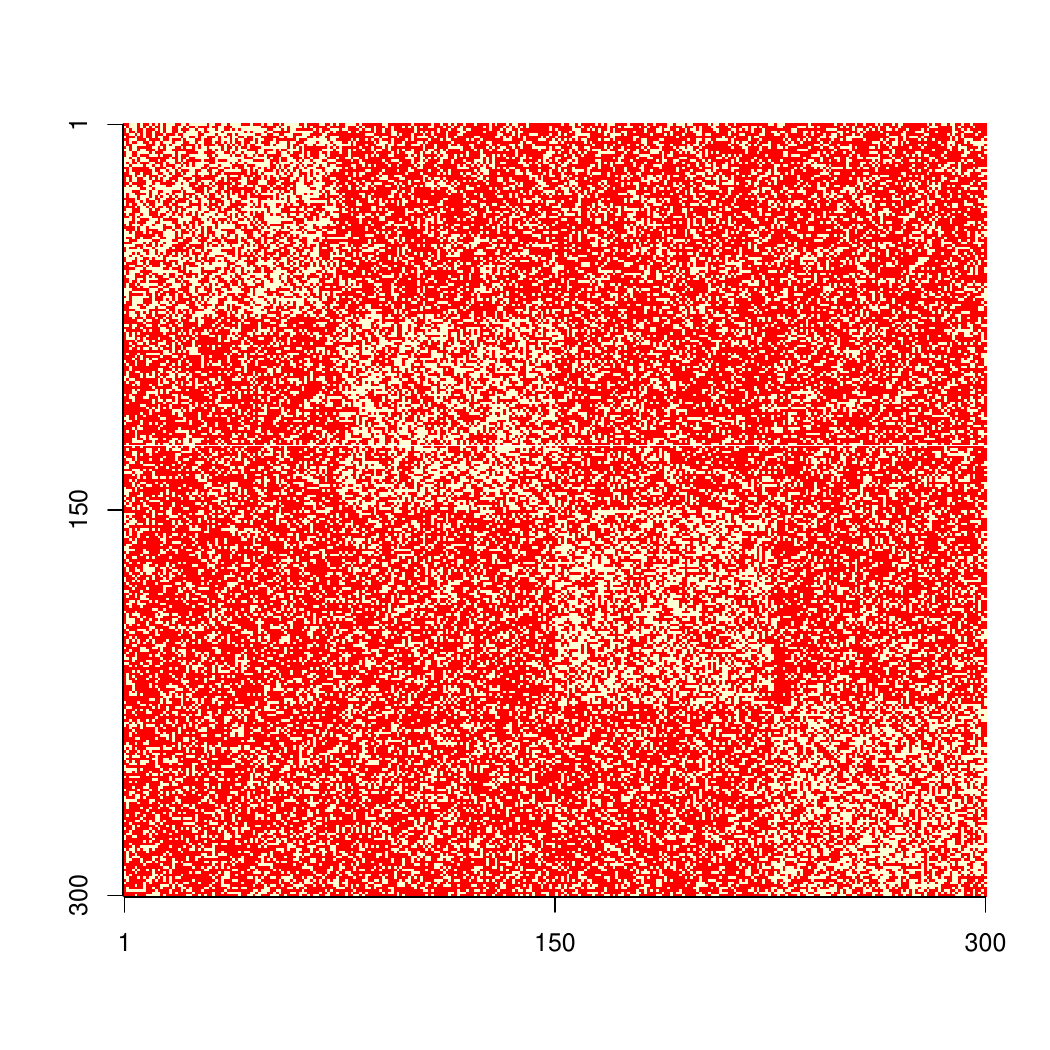} 
		\includegraphics[width = 0.45\textwidth]{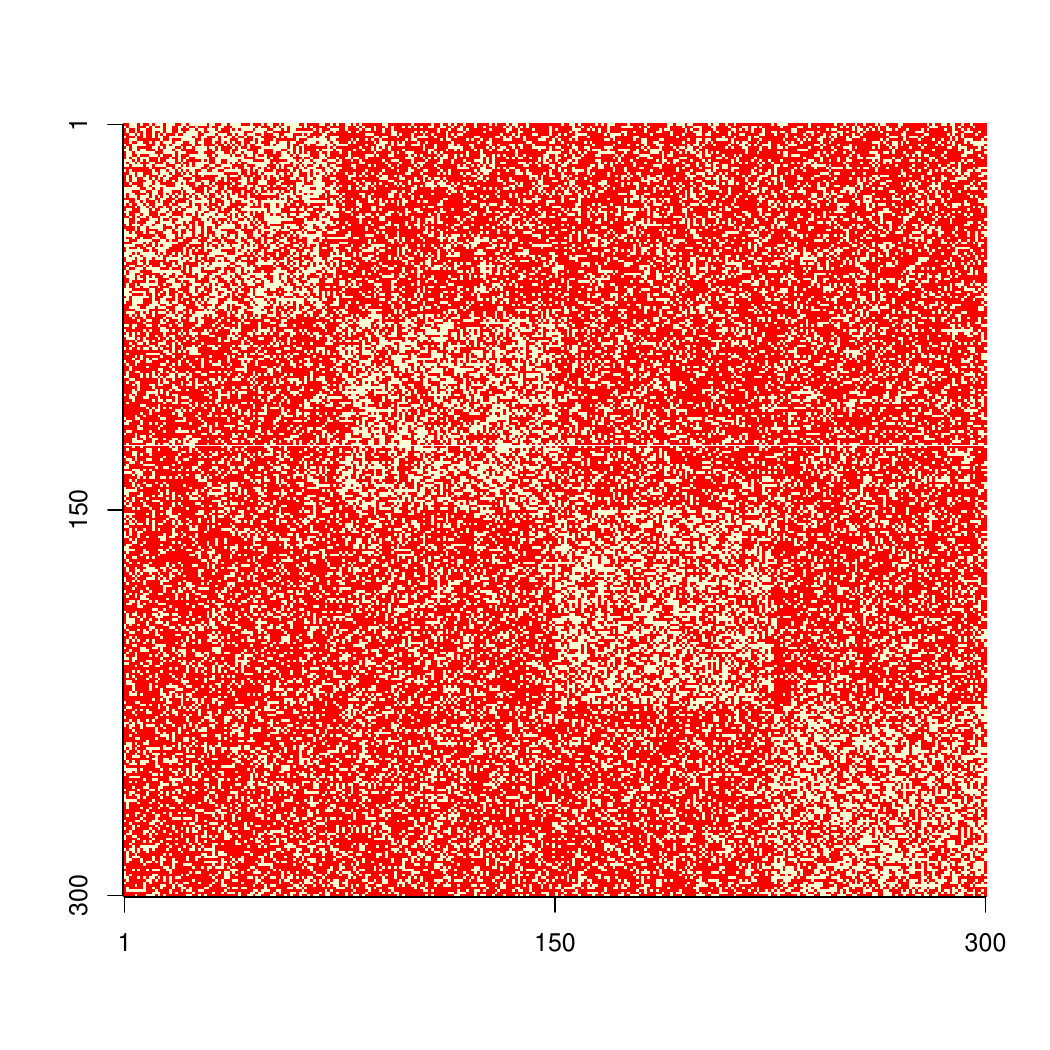}
		\includegraphics[width = 0.45\textwidth]{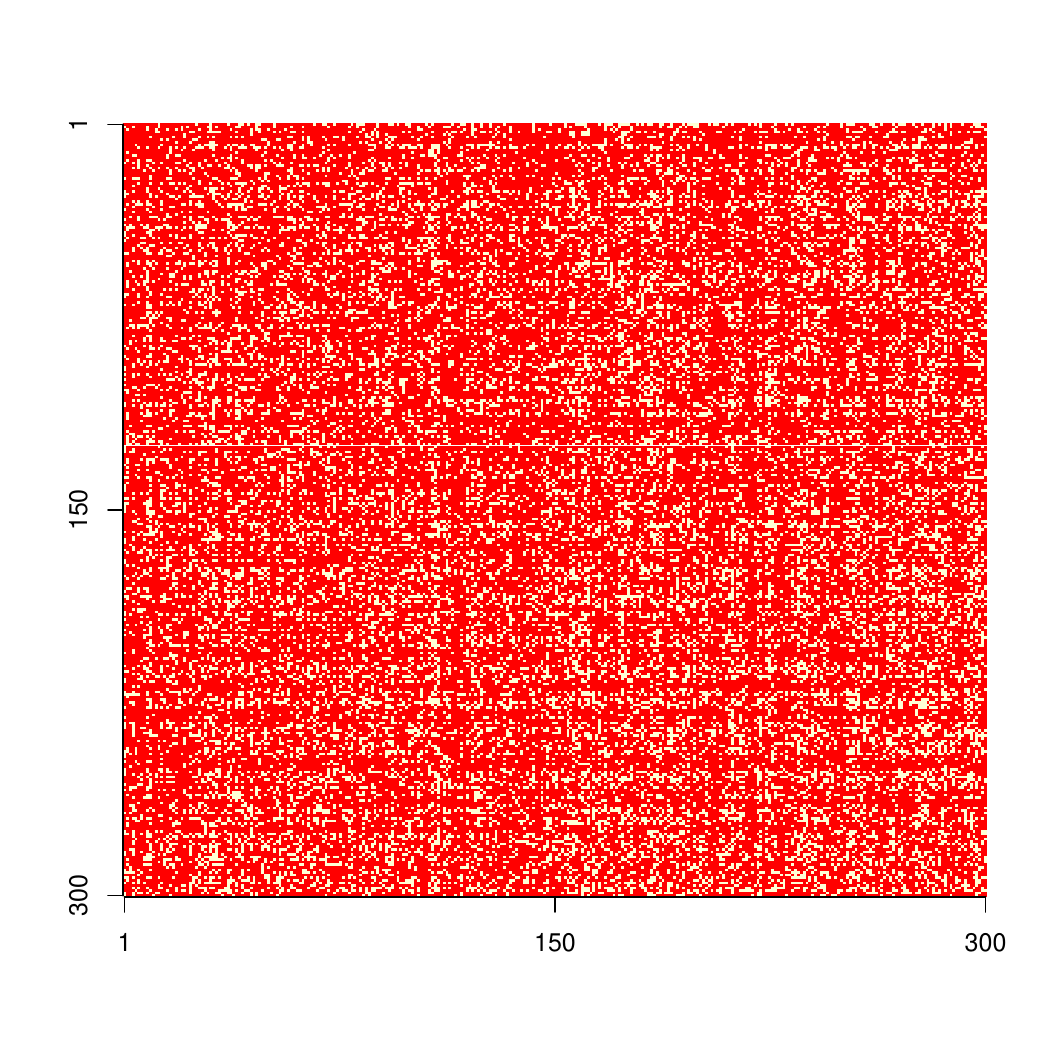} 
		\includegraphics[width = 0.45\textwidth]{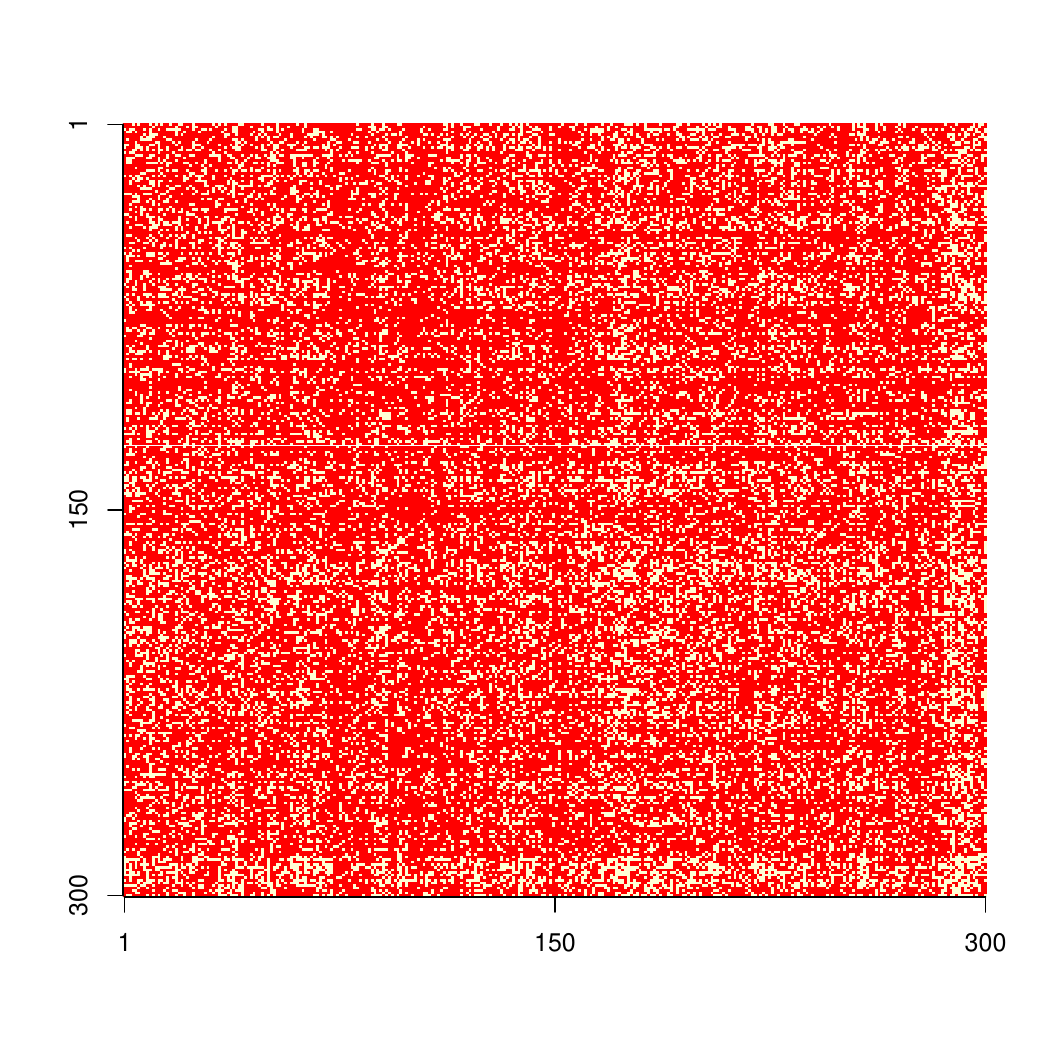}
		\caption{\label{fig1}  The top row shows two instances of data generated in \textbf{Scenario 1}.  The left panel corresponds to $A(t)$ for $t$ before the first change point, and the right panel to $A(t)$ between the first and   second change points.  The bottom row shows the corresponding  plots for \textbf{Scenario 2} with $\varepsilon =0.05$.}
	\end{center}
\end{figure}

\begin{figure}[t!]
	\begin{center}
		\includegraphics[width = 0.45\textwidth]{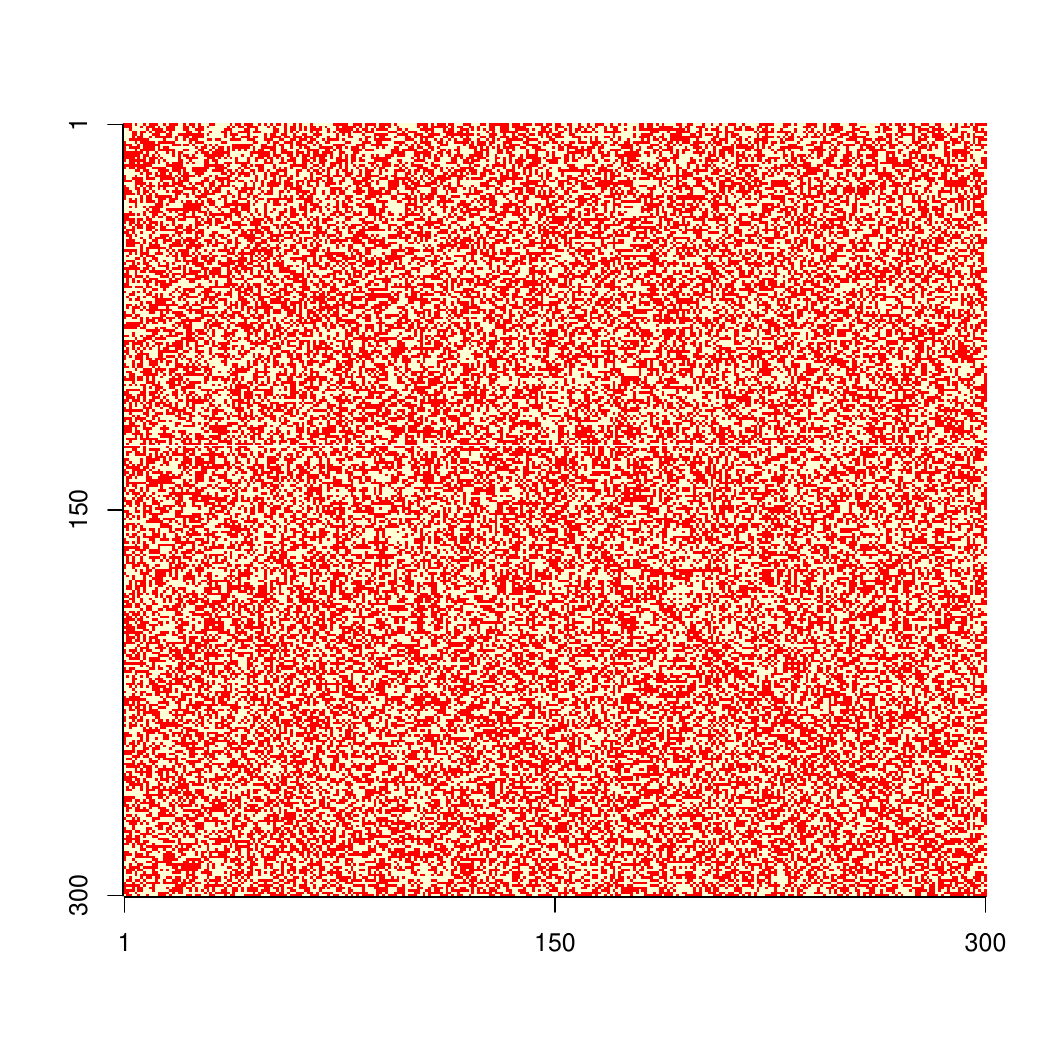} 
		\includegraphics[width = 0.45\textwidth]{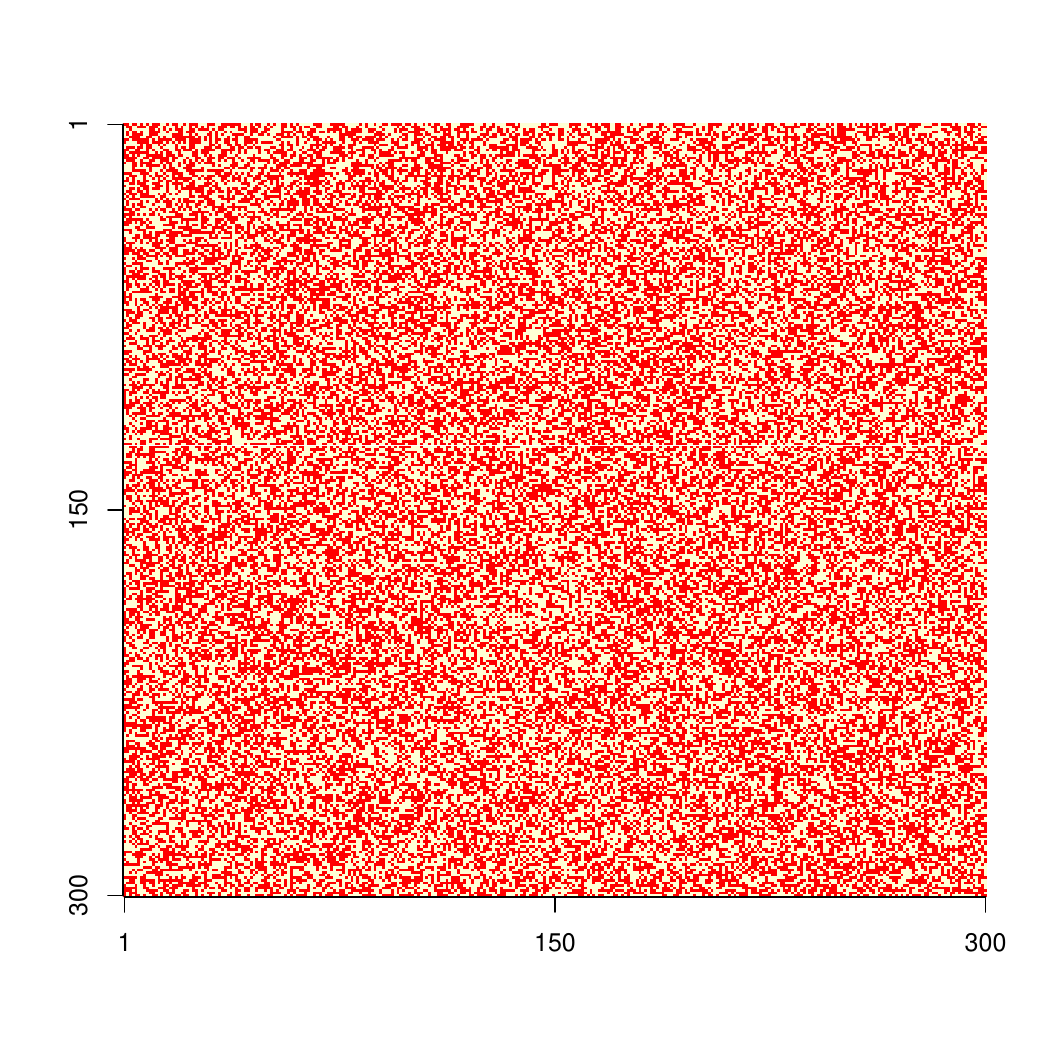}
		\includegraphics[width = 0.45\textwidth]{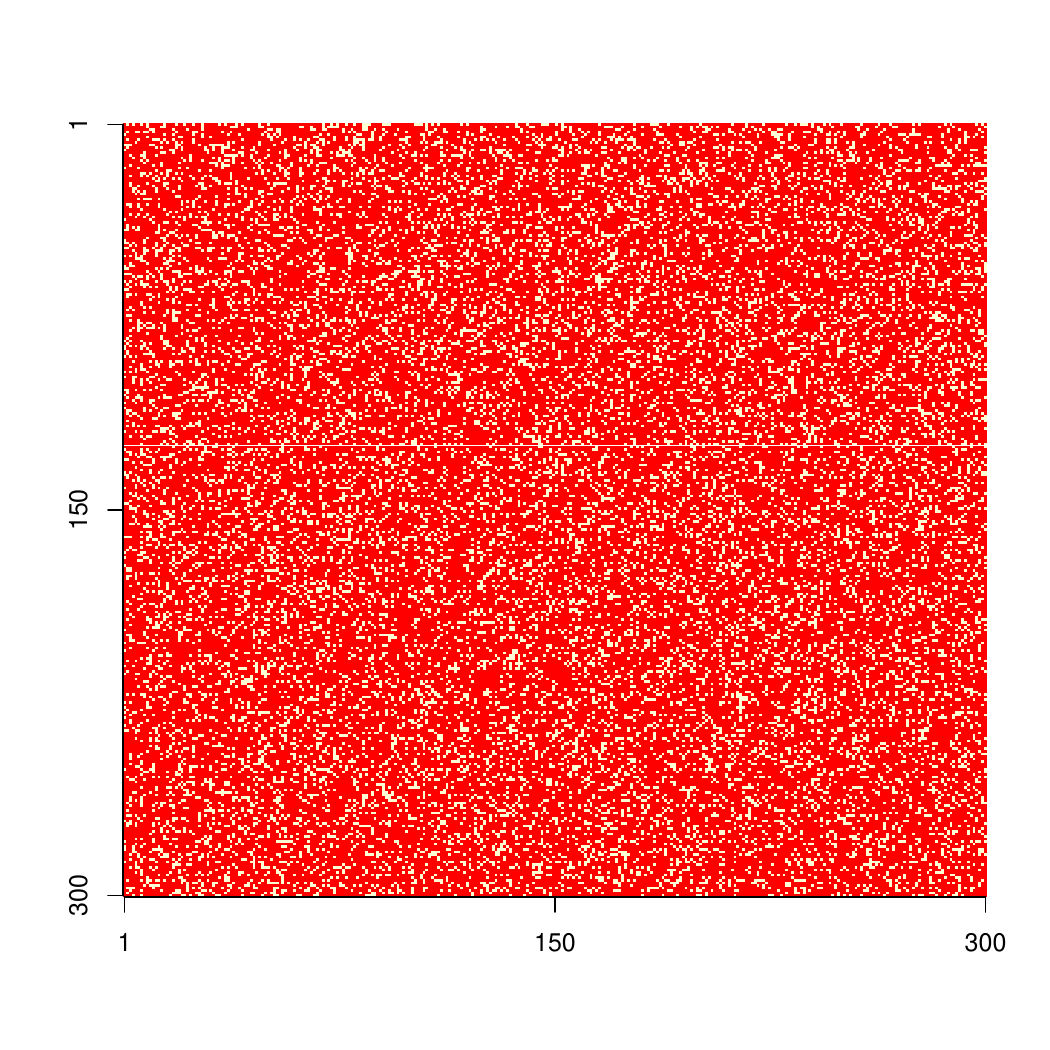} 
		\includegraphics[width = 0.45\textwidth]{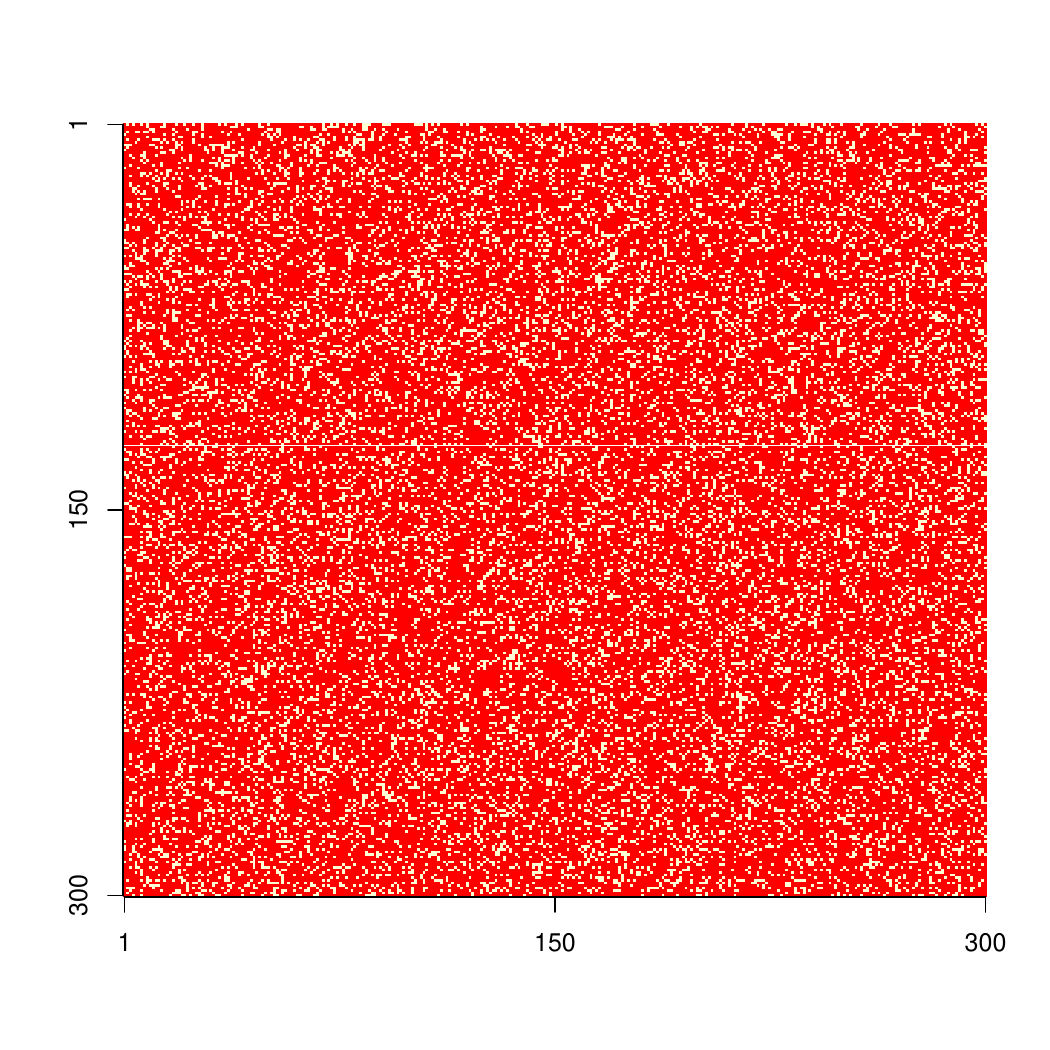}
		\caption{\label{fig2} The top row shows two instances of data generated in \textbf{Scenario 3}.  The left panel corresponds to $A(t)$ for $t$ before the first change point, and the right panel to $A(t)$ between the first and   second change points.  The bottom row shows the corresponding  plots for \textbf{Scenario 4} with $\varepsilon =0.05$.}
	\end{center}
\end{figure}

\begin{table}[t!]
	\centering
	\caption{\label{tab1} Scenario 1}
	\medskip
	\setlength{\tabcolsep}{18pt}
	\begin{small}
		\begin{tabular}{ rrrrrr}
			\hline
			Method &	$n$ & $\rho$  & $\vert  K - \widehat{K}\vert$   & $d(\widehat{\mathcal{C}}| \mathcal{C})$  & $d( \mathcal{C}|\widehat{\mathcal{C}})$ \\ 
			\hline
		NonPar-RDPG-CPD & 300 & 0 & 0.1 &  \textbf{1.0} &  1.0 \\
		NBS	& 300& 0&  \textbf{0.0}  & \textbf{1.0} & 1.0\\			
				MNBS	& 300& 0& 1.16&  50.0  &  
				\textbf{0.0}\\				
			\hline		
		NonPar-RDPG-CPD & 200 & 0 &   \textbf{0.0} & \textbf{1.0}& \textbf{1.0}\\
NBS	& 200&0&\textbf{0.0} & \textbf{1.0} & \textbf{1.0}\\		
MNBS	& 200&0&   1.92& $\inf$ & $-\inf$ \\
			\hline
		NonPar-RDPG-CPD & 100 &0  &    0.2   &  \textbf{1.0} & 1.0\\
NBS	& 100&0&\textbf{0.0} &   \textbf{1.0}&  1.0\\			
MNBS	& 100&0& 0.84 &50.0 & \textbf{0.0}\\
			\hline
			NonPar-RDPG-CPD & 300 & 0.5 & \textbf{0.0} &  \textbf{0.0}   &\textbf{0.0}  \\
			NBS	& 300&0.5 & 21.2 &  1.0  & 43.0 \\			
		MNBS	& 300&0.5 &\textbf{0.0} & \textbf{0.0}  & \textbf{0.0} \\	
			\hline		
			NonPar-RDPG-CPD & 200 & 0.5 &0.04& \textbf{0.0}  &\textbf{0.0}  \\
			NBS	& 200&0.5 &21.3 &  1.0  & 4.30 \\
		MNBS	& 200&0.5 &\textbf{0.0}  &\textbf{0.0}   &\textbf{0.0}  \\					
			\hline
			NonPar-RDPG-CPD & 100 &0.5  & 0.16   &  \textbf{0.0}  &\textbf{0.0}  \\
			NBS	& 100&0.5&  21.3 &  1.0   & 42.0 \\			
		MNBS	& 100&0.5&\textbf{0.12} &   \textbf{0.0}   &\textbf{0.0} \\			
			\hline
						NonPar-RDPG-CPD & 300 & 0.9 & \textbf{0.0} &\textbf{0.0} &\textbf{0.0}\\
			NBS	& 300&0.9 &21.0 & 1.0 & 43.0\\	
		MNBS	& 300&0.9 & 3.12&  \textbf{0.0}  &36.0 \\			
			\hline		
			NonPar-RDPG-CPD & 200 & 0.9 &\textbf{0.0} & \textbf{0.0}& \textbf{0.0}\\
			NBS	& 200&0.9 &21.0 &  1.0  & 43.0 \\		
		MNBS	& 200&0.9 &2.88 & \textbf{0.0} &35.0 \\	
			\hline
			NonPar-RDPG-CPD & 100 &0.9  & \textbf{0.0}  &\textbf{1.0}&\textbf{1.0} \\
			NBS	& 100&0.9&21.04 &1.0 &43.0\\
		MNBS	& 100&0.9&3.28 &0.0 &35.0\\
		\end{tabular}
	\end{small}
\end{table}

\begin{table}[t!]
	\centering
	\caption{\label{tab2} Scenario 2}
	\medskip
	\setlength{\tabcolsep}{18pt}
	\begin{small}
		\begin{tabular}{ rrrrrr}
			\hline
			Method &	$n$ &  $\varepsilon$  & $\vert  K - \widehat{K}\vert$   & $d(\widehat{\mathcal{C}}| \mathcal{C})$  & $d( \mathcal{C}|\widehat{\mathcal{C}})$ \\ 
			\hline
			NonPar-RDPG-CPD & 300 & 0.3 & \textbf{0.04}   & \textbf{0.0} &\textbf{0.0}\\
			NBS	& 300& 0.3  &0.28 &1.0   & 1.0 \\			
	     MNBS	& 300& 0.3  &0.76 &\textbf{0.0}   & 21.0 \\	
			\hline	
		NonPar-RDPG-CPD & 200 & 0.3 & \textbf{0.0}   & \textbf{0.0} &\textbf{0.0}\\
			NBS	& 200& 0.3  & 0.32& 1.0  &1.0  \\	
		MNBS	& 200& 0.3  & 0.48&   \textbf{0.0}  & 1.0 \\			
			\hline
	NonPar-RDPG-CPD & 100 & 0.3 & \textbf{0.08}  & 3.0 &3.0 \\
NBS	& 100& 0.3  & \textbf{0.08} & 1.0&\textbf{1.0}\\
MNBS	& 100& 0.3  & 0.64& \textbf{0.0} & 18.0 \\								
			\hline
		NonPar-RDPG-CPD & 300 & 0.15 & \textbf{0.0}  &2.0  &2.0\\
			NBS	& 300& 0.15  &0.4 & 1.0 & \textbf{1.0}\\	
     	MNBS	& 300& 0.15  &0.76& \textbf{0.0}&21.0 \\			
			\hline
			NonPar-RDPG-CPD & 200 & 0.15 & \textbf{0.04}  &  3.0  & 3.0\\
			NBS	& 200& 0.15  & 0.28& 1.0 &\textbf{1.0}\\
		MNBS	& 200& 0.15  &0.76 &\textbf{0.0}   & 20.0\\			
			\hline
		NonPar-RDPG-CPD & 100 & 0.15 &\textbf{0.28}   &    4.0  &  10.0\\
			NBS	& 100& 0.15  & 0.32&   \textbf{1.0} &\textbf{1.0} \\
		MNBS	& 100& 0.15  &0.48&  \textbf{1.0} &  5.0\\	
			\hline
	NonPar-RDPG-CPD & 300 & 0.05 & \textbf{0.72}  &  36.0 & \textbf{5.0}\\
NBS	& 300& 0.05  &0.84 & \textbf{1.0} &  9.0 \\	
MNBS	& 300& 0.05  &\textbf{1.24}& \textbf{1.0}& 21.0 \\			
\hline
NonPar-RDPG-CPD & 200 & 0.05 & 0.64  &  37.0   & \textbf{6.0}\\
NBS	& 200& 0.05  &0.76 & \textbf{3.0} &11.0\\
MNBS	& 200& 0.05  &\textbf{0.6} &  4.0   &8.0 \\			
\hline
NonPar-RDPG-CPD & 100 & 0.05 &\textbf{0.72}   & \textbf{19.0}&\textbf{15.0}\\
NBS	& 100& 0.05  &1.4&  $\inf$  & $-\inf$ \\
MNBS	& 100& 0.05  &1.88&  $\inf$  &$-\inf$ \\
		\hline
		\end{tabular}
	\end{small}
\end{table}

\begin{table}[t!]
	\centering
	\caption{\label{tab3} Scenario 3}
	\medskip
	\setlength{\tabcolsep}{18pt}
	\begin{small}
		\begin{tabular}{ rrrrr}
			\hline
			Method &	$n$  & $\vert  K - \widehat{K}\vert$   & $d(\widehat{\mathcal{C}}| \mathcal{C})$  & $d( \mathcal{C}|\widehat{\mathcal{C}})$ \\ 
			\hline
			NonPar-RDPG-CPD & 300 &\textbf{0.24}  &\textbf{0.0} & \textbf{0.0} \\
			NBS	& 300& 15.04 &1.0 & 43.0\\
			MNBS	& 300&0.84 &25 &36 \\				
			\hline		
			NonPar-RDPG-CPD & 200 & \textbf{0.08} &\textbf{0.0} &\textbf{0.0}  \\
		NBS	& 200& 14.4 &43.0 &1.0\\
		MNBS	& 200&0.84 &23 &36 \\	
			\hline
			NonPar-RDPG-CPD & 100 & \textbf{0.52} &\textbf{3.0}& \textbf{5.0} \\
		NBS	& 100&13.96 &1.0   &43.0\\
		MNBS	& 100&1.16 & 23& 35\\			
			\hline
		\end{tabular}
	\end{small}
\end{table}

Examples of matrices $A(t)$ generated in each  scenario are depicted in Figures~\ref{fig1}-\ref{fig2}.  We can see qualitative differences among \textbf{Scenarios 1}-\textbf{4}.  In particular, \textbf{Scenario 1} produces adjacency matrices with  block structure.  Interpretation is less clear for the other models, but we see that \textbf{Scenario 3} seems to generate more dense graphs than \textbf{Scenarios 2} and \textbf{4}.

Results comparing NonPar-RDPG-CPD with NBS are provided in Tables \ref{tab1}-\ref{tab4}.  We observe that, overall, NonPar-RDPG-CPD provides generally reliable estimation of the number of change points and their locations. 

In \textbf{Scenario 1}  with $\rho = 0$, a model where the marginal distributions of $A(t)$ only change in mean, we see from \Cref{tab1} that NBS outperforms our proposed approach.  This does not come as a surprise since NBS is designed to   detect change points in mean.  However, as $\rho$ increases and the number of samples decreases, the most robust method seems to be NonPar-RDPG-CPD. 

\textbf{Scenario 2} poses an interesting example where the behaviour of only a fraction of nodes in the network changes at the change points.  Furthermore, the data are generated under an RDPG model.  As shown in \Cref{tab2}, NonPar-RDPG-CPD seems to be the best method for estimating the number of change points.  A possible explanation is that the underlying changes in the distributions of $A(t)$ not only occur at the level of the means, and hence the NBS might not be the ideal for this scenario even though it outperforms MNBS in this framework.  Our method was constructed under the assumption of the RDPG   model.

To assess the robustness of our method to misspecification, we can look at the performance of our method in the context of \textbf{Scenario 3}  which is not an RDPG.  Interestingly, \Cref{tab3} shows that NonPar-RDPG-CPD is the best in this  model with MNBS coming in second.  In contrast, NBS suffers greatly, overestimating the number of change points.  This makes sense since  between change points, the latent positions $X(t)$ remain constant with probability 0.9 and change with probability 0.1.  Hence, some of these changes in $X(t)$ could be confused as change points by NBS.

Finally,  \textbf{Scenario 4} consists of an example of \Cref{assume:model-rdpg}.  However, similarly as \textbf{Scenario 2},  the change points correspond to shifts in the behaviour of only some of the nodes in the network.  In particular,  \Cref{tab4} suggests that our method performs reasonably well,  improving its  performance when the signal-to-noise ratio  increases. This is different from the NBS which once again tends to overestimate the number of change points.  As for the MNBS, we see that this method is unable to detect the change points in this example. 

\begin{figure}[pb!]
	\begin{center}
		\includegraphics[width=0.24\textwidth]{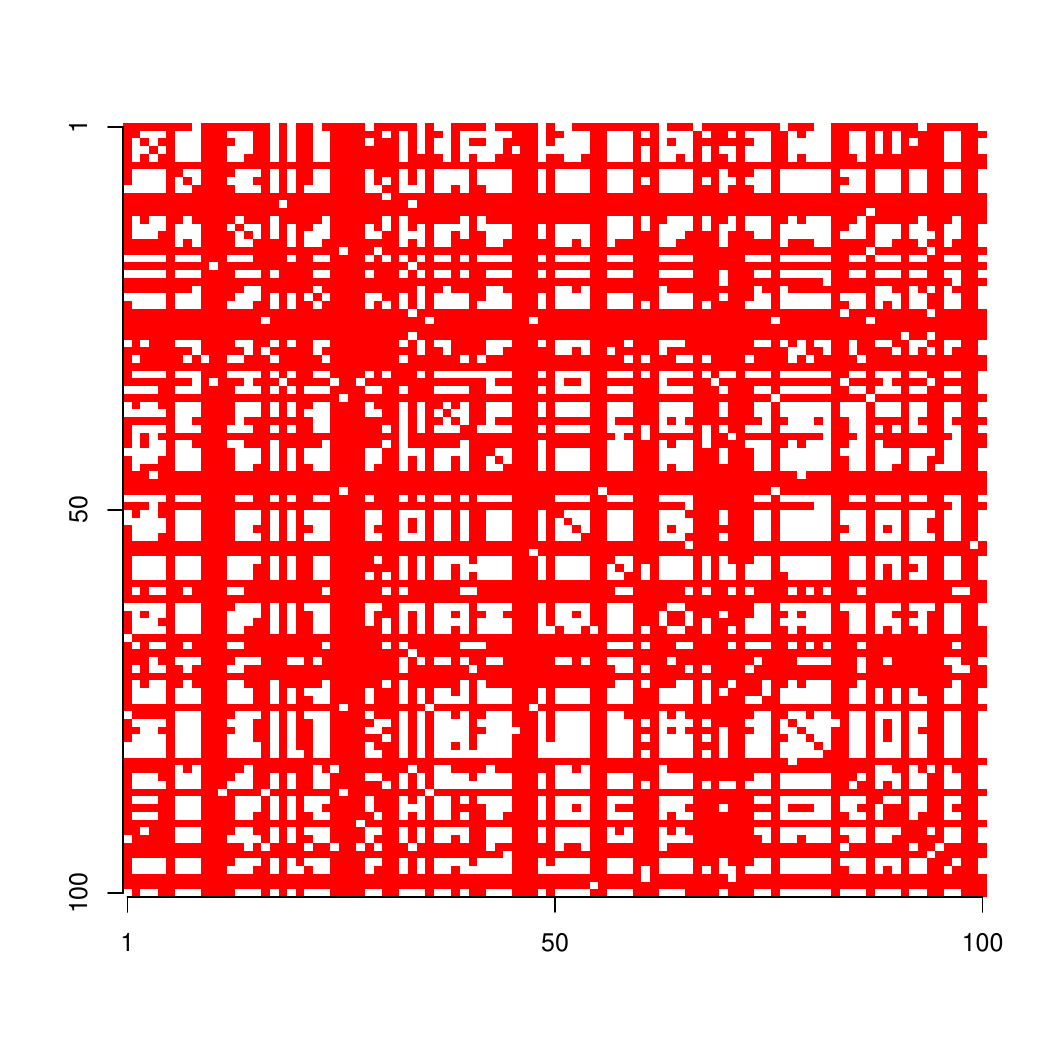} 
		\includegraphics[width=0.24\textwidth]{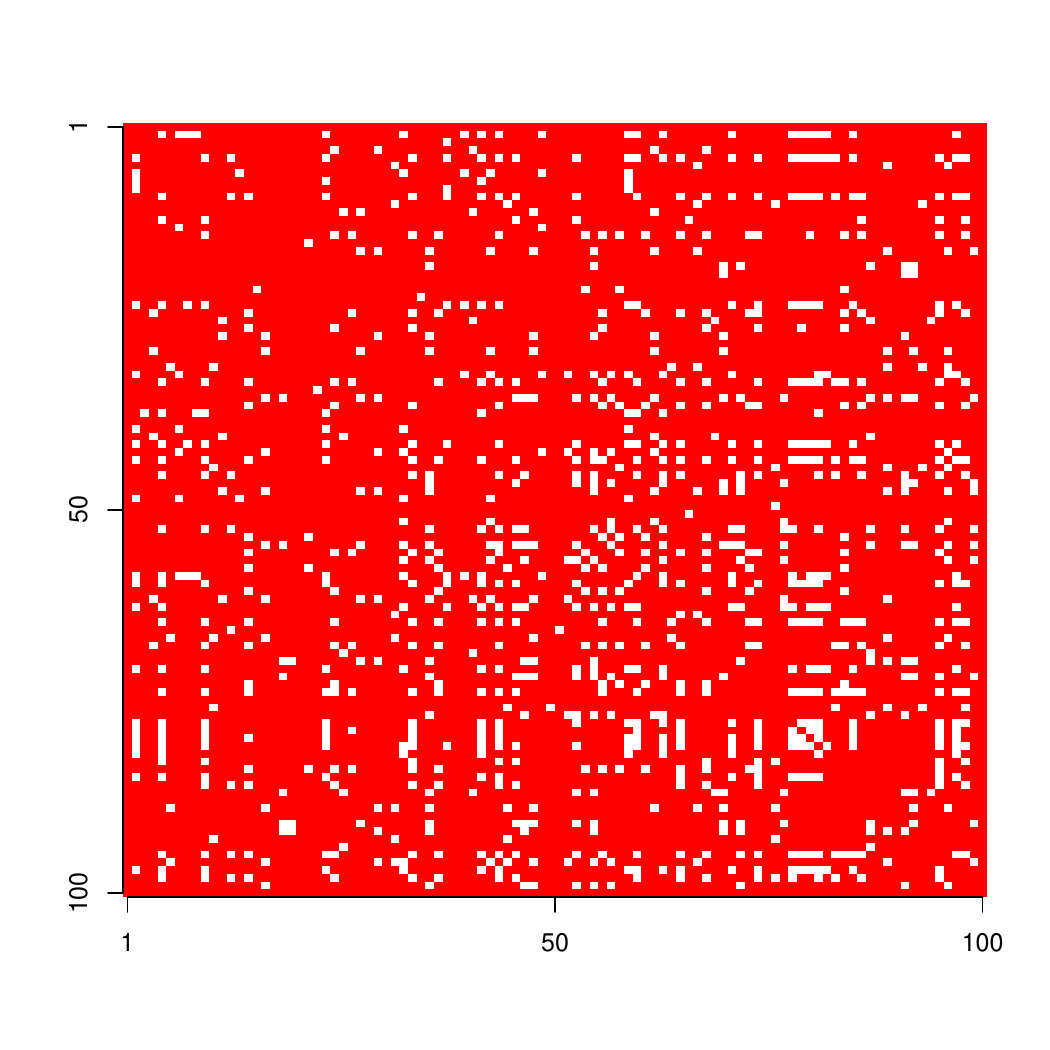}
		\includegraphics[width=0.24\textwidth]{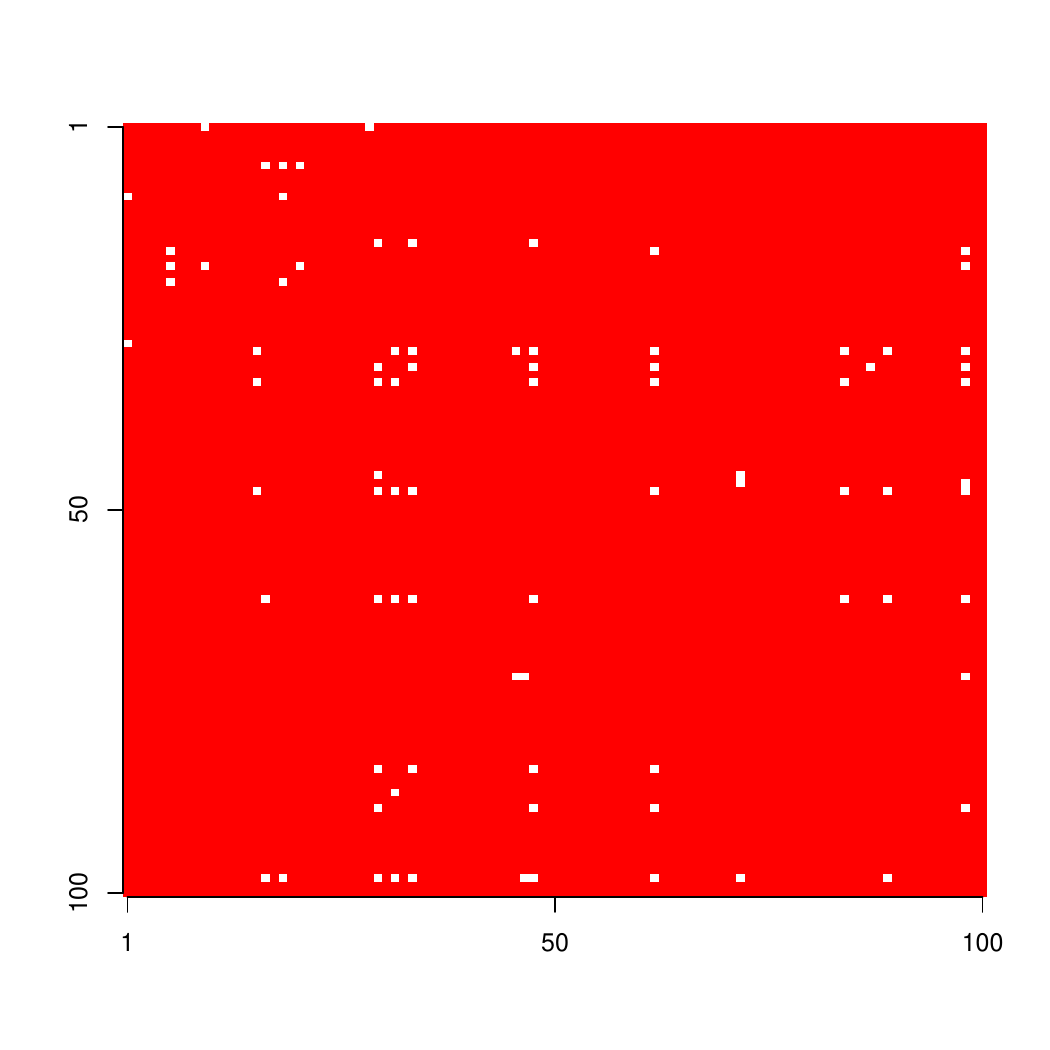}
		\includegraphics[width=0.24\textwidth]{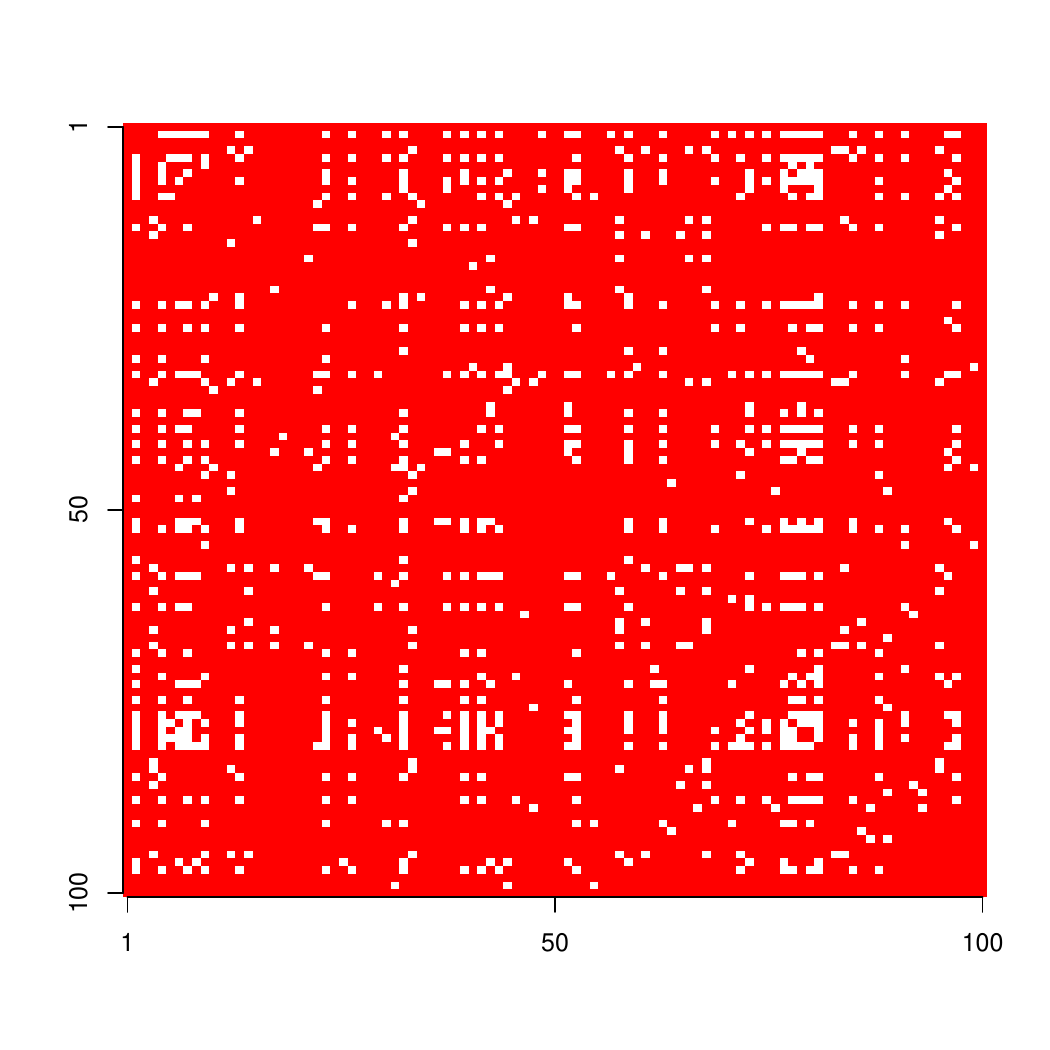} 
		\includegraphics[width=0.24\textwidth]{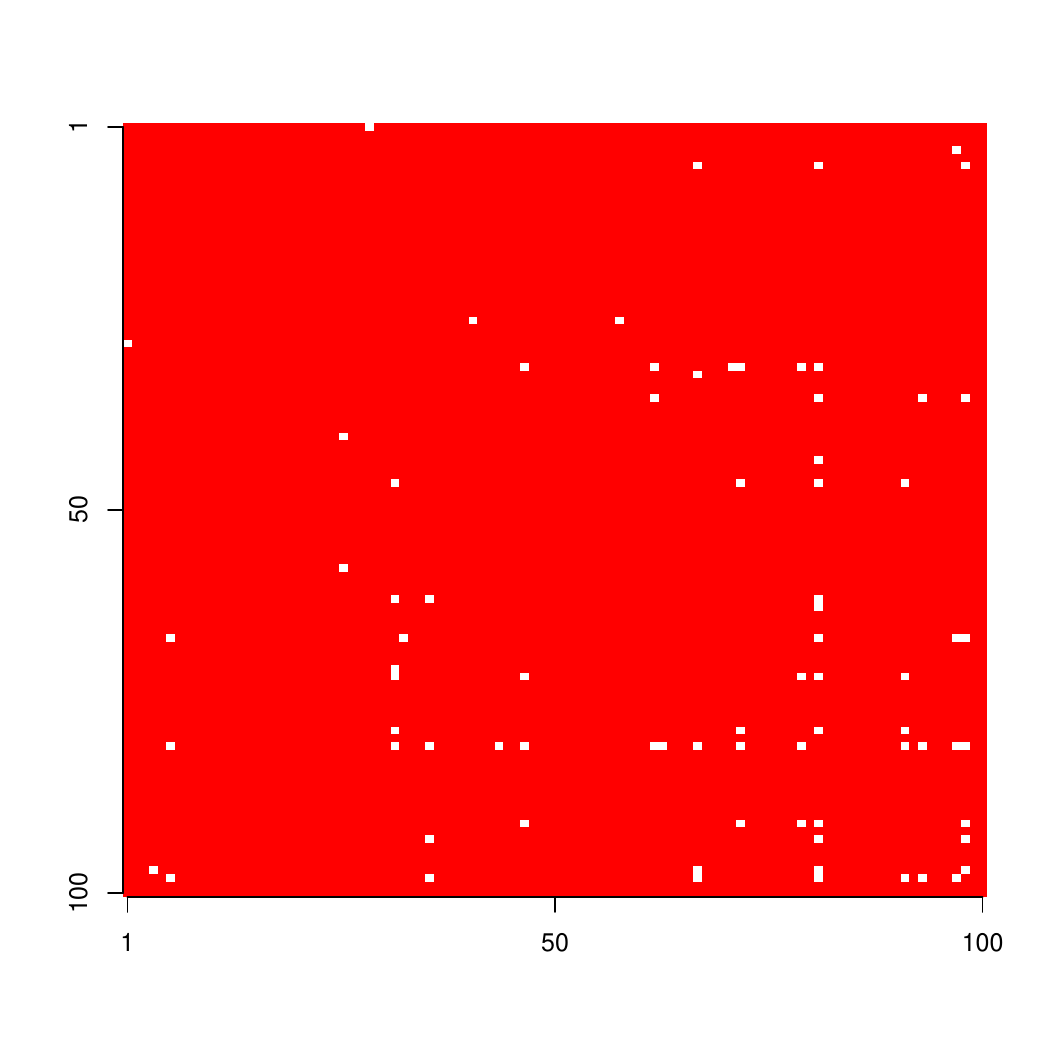}
		\includegraphics[width=0.24\textwidth]{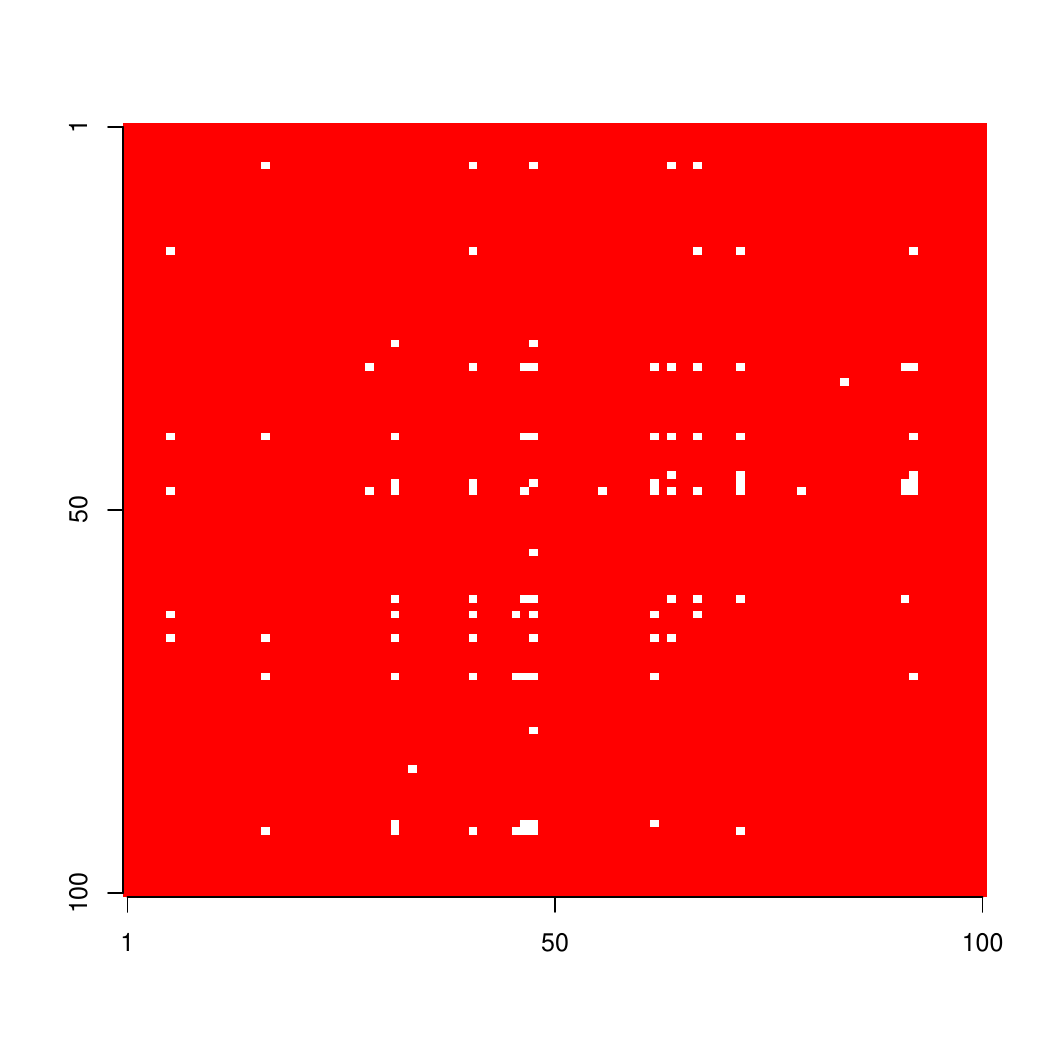}
		\includegraphics[width=0.24\textwidth]{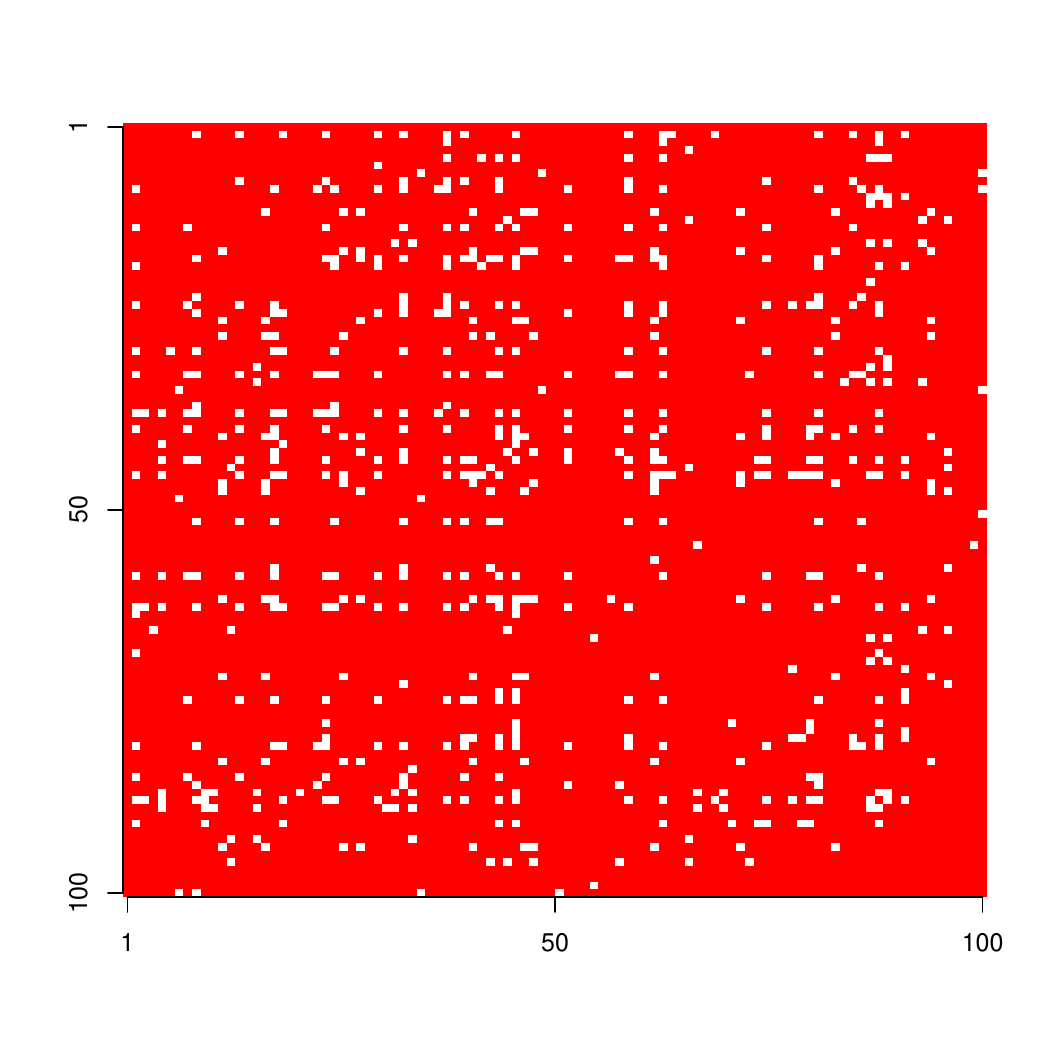} 
		\includegraphics[width=0.24\textwidth]{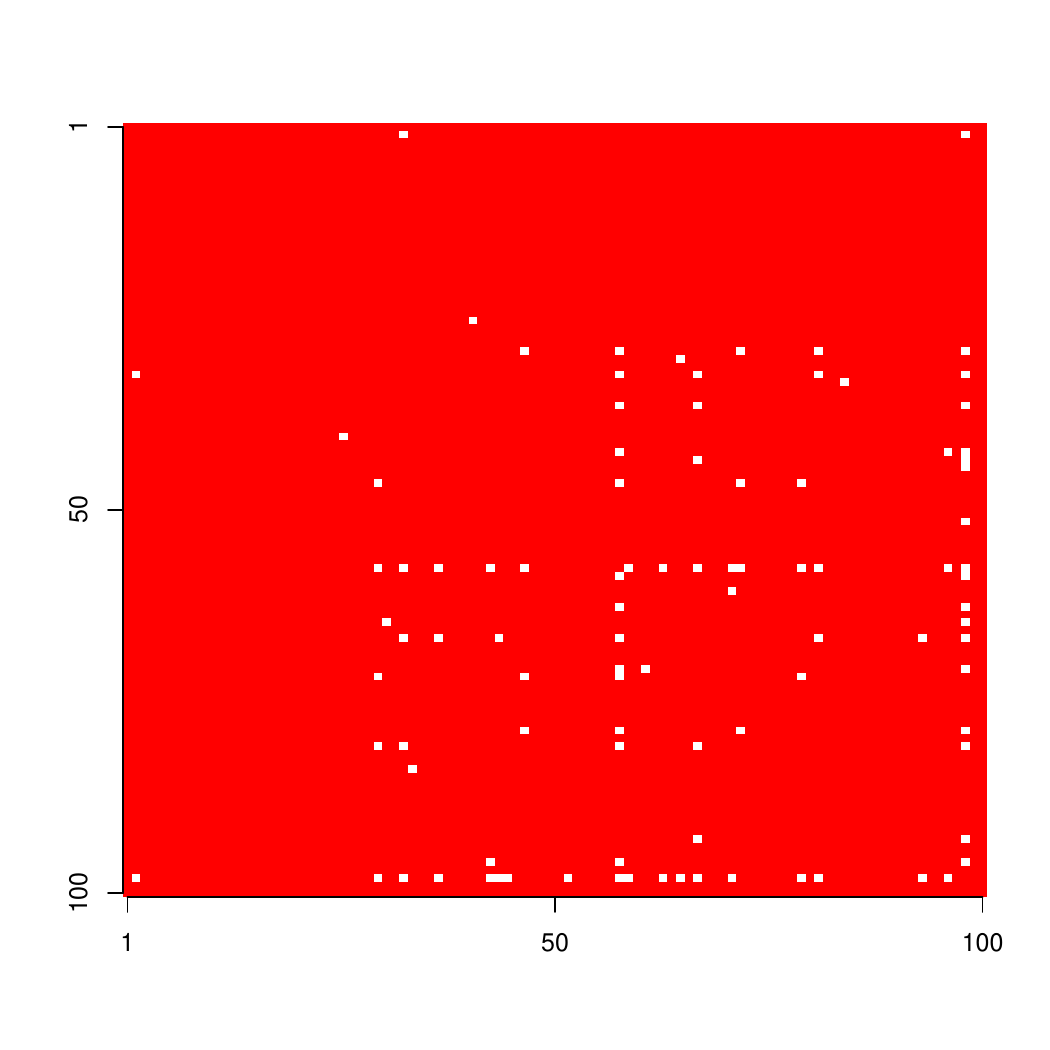}
		\includegraphics[width=0.24\textwidth]{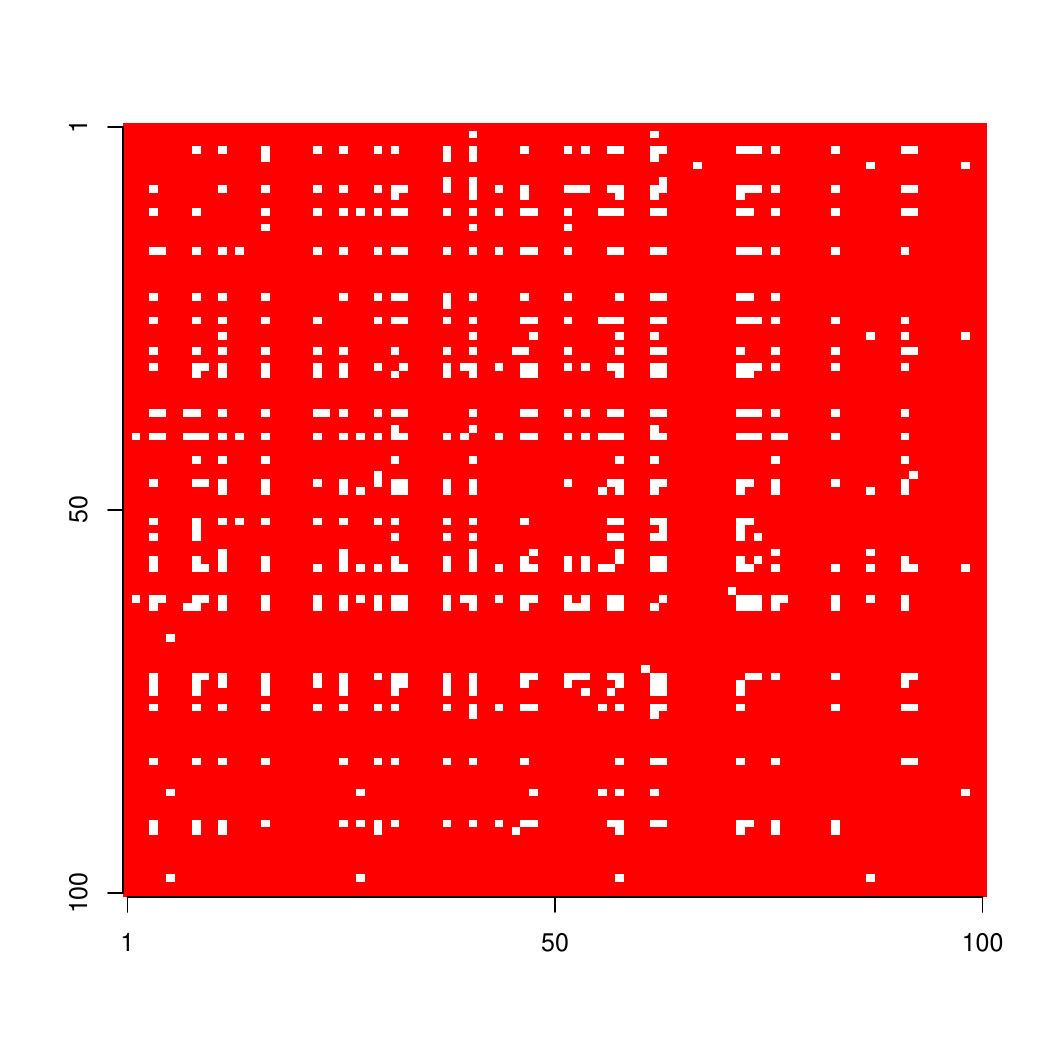}
		\includegraphics[width=0.24\textwidth]{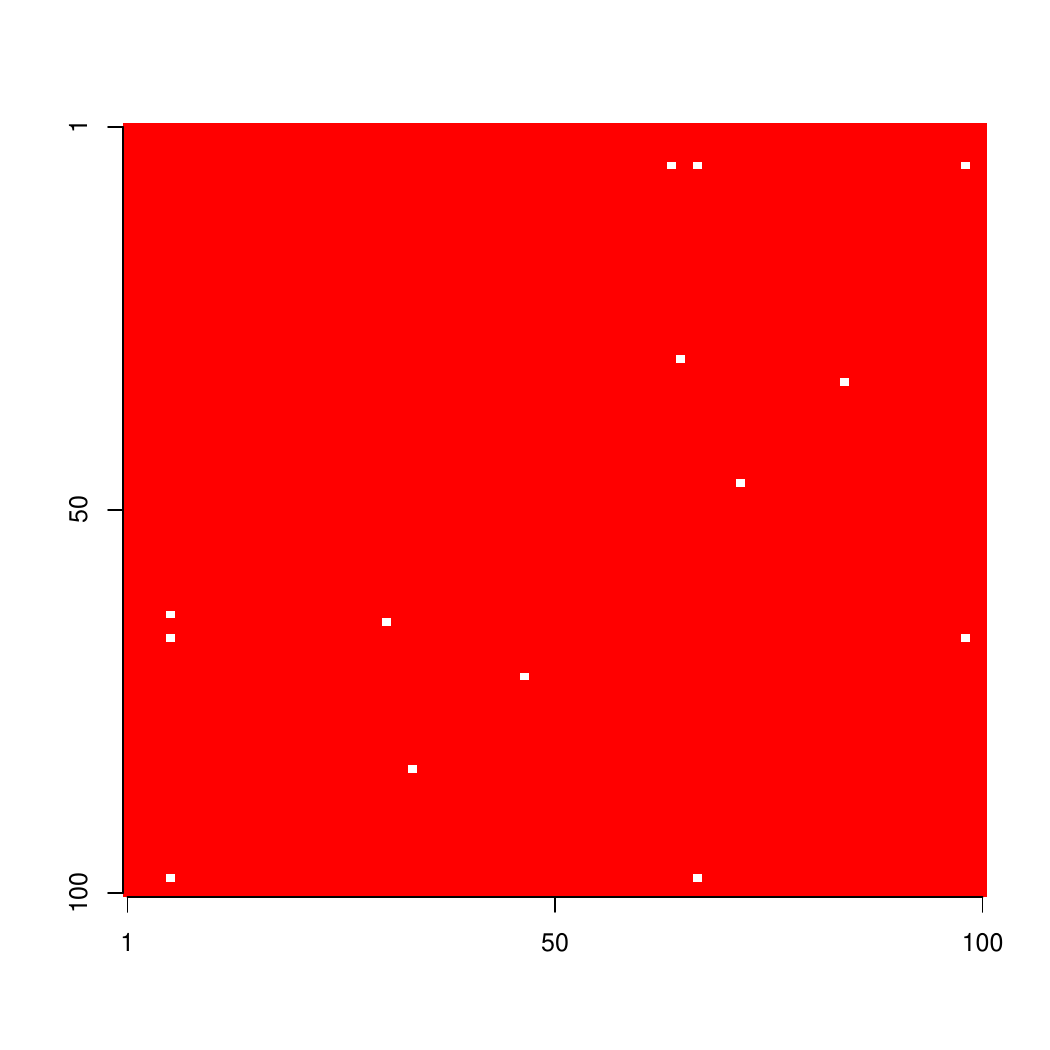} 
		\includegraphics[width=0.24\textwidth]{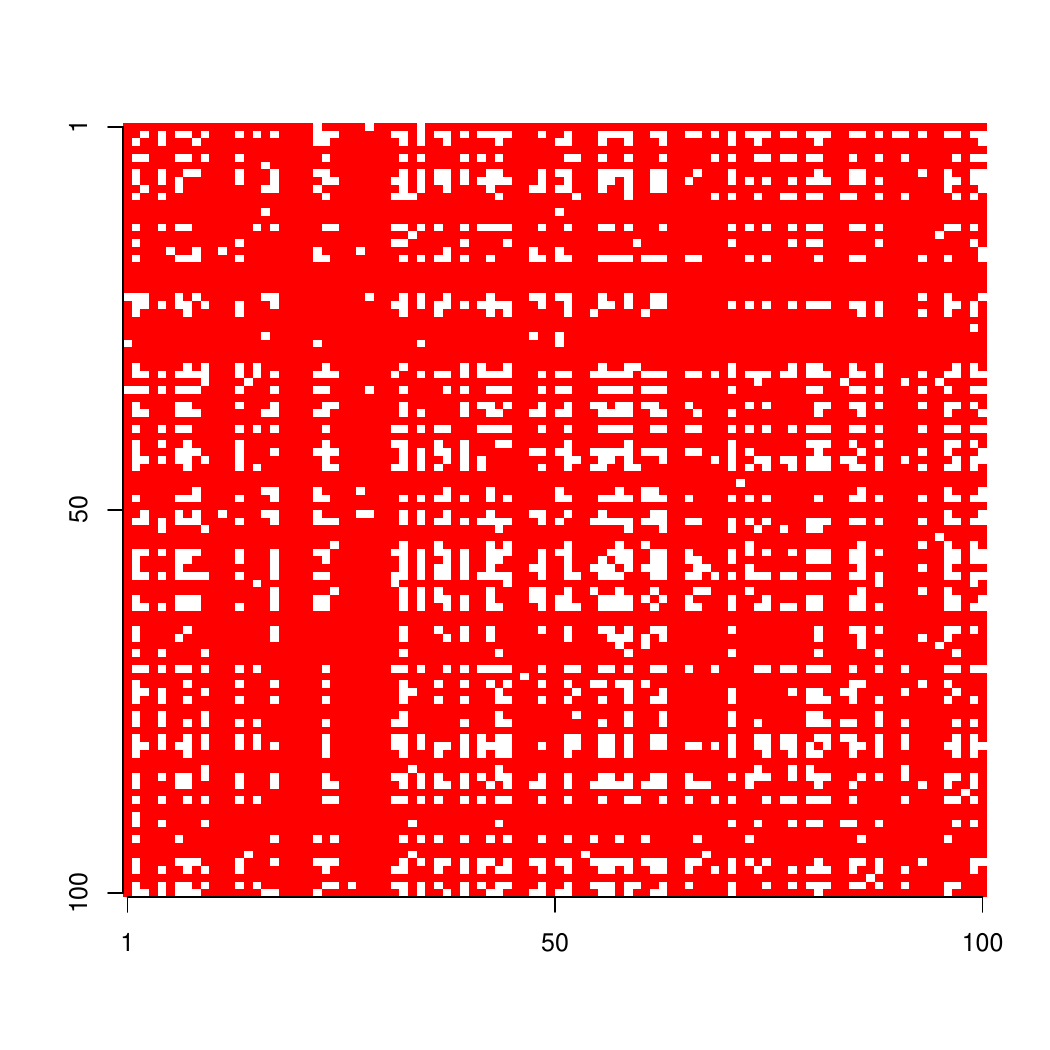}
		\includegraphics[width=0.24\textwidth]{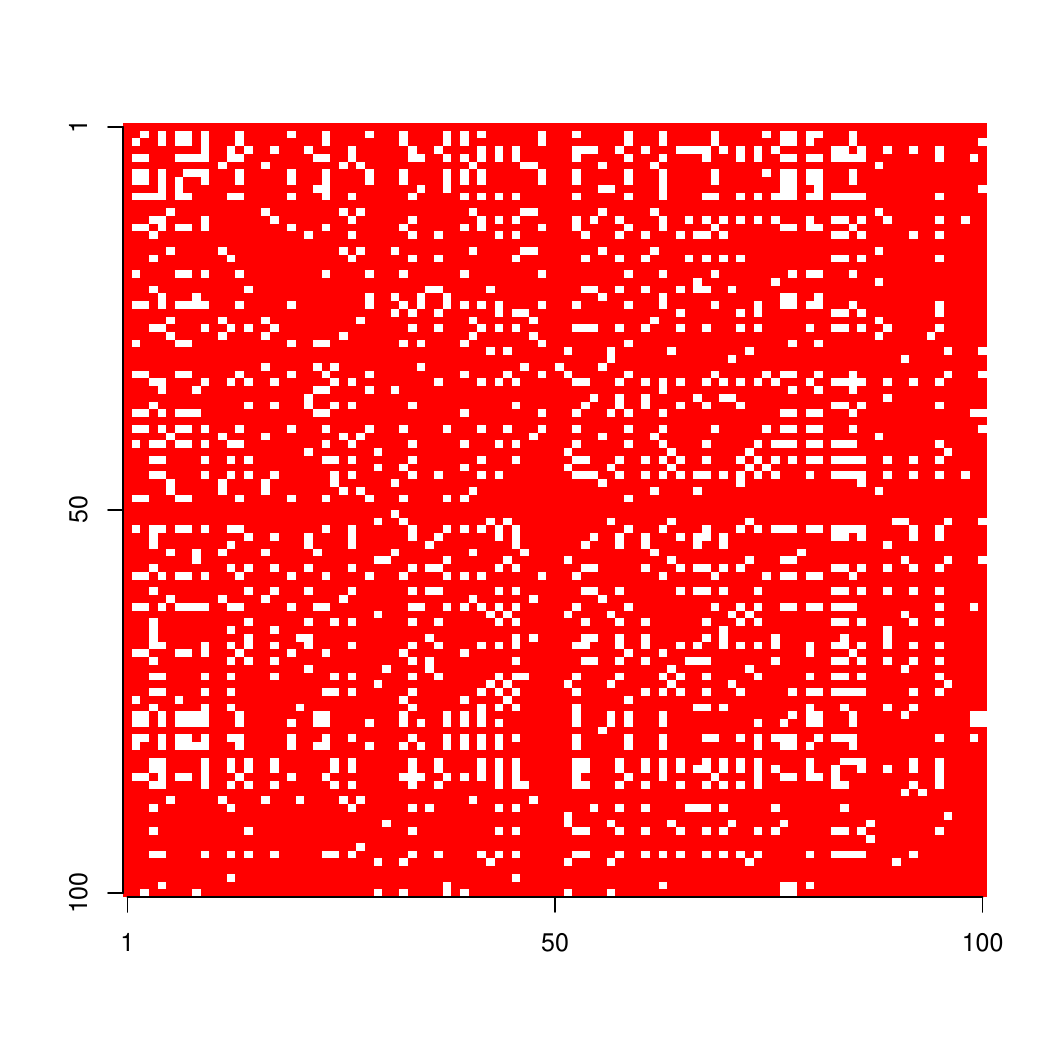}	 
				\includegraphics[width=0.24\textwidth]{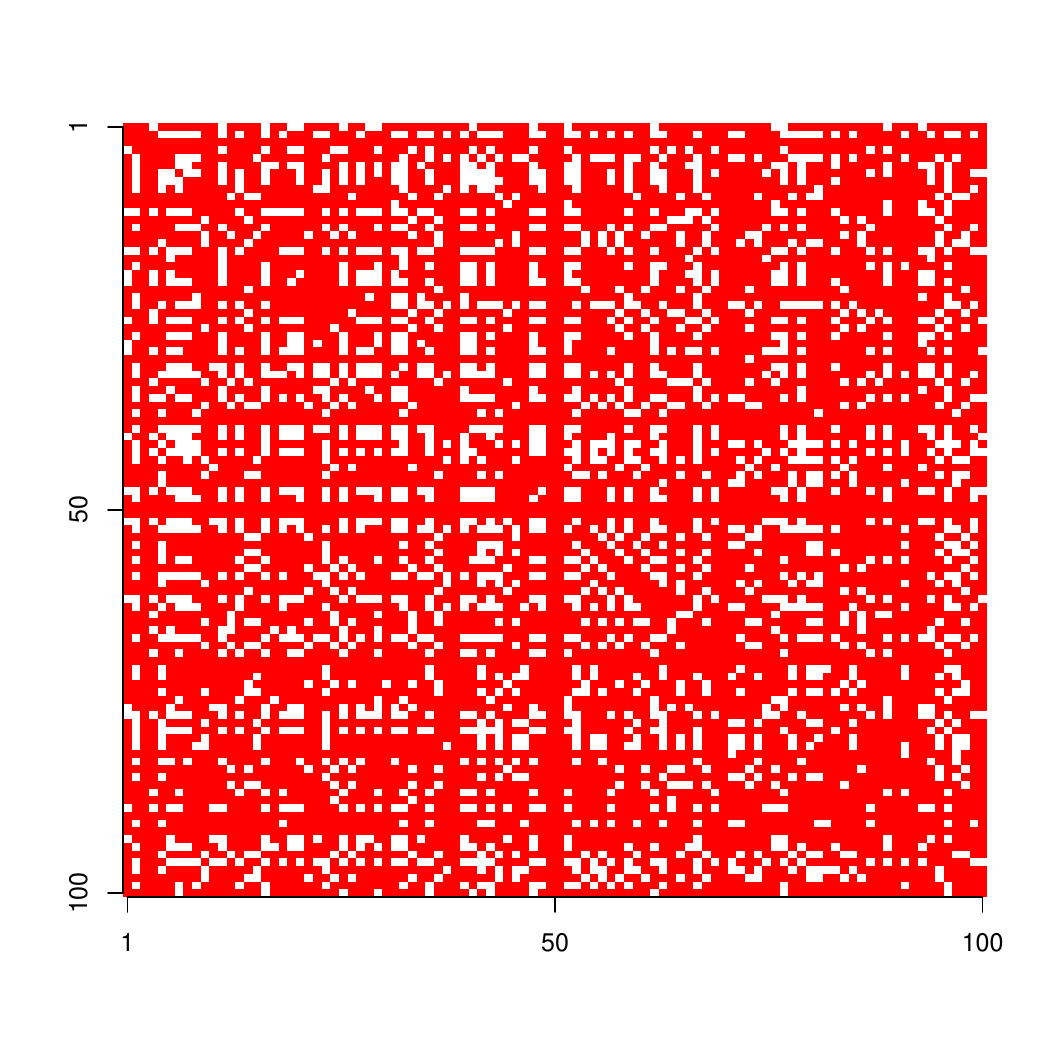}
					\includegraphics[width=0.24\textwidth]{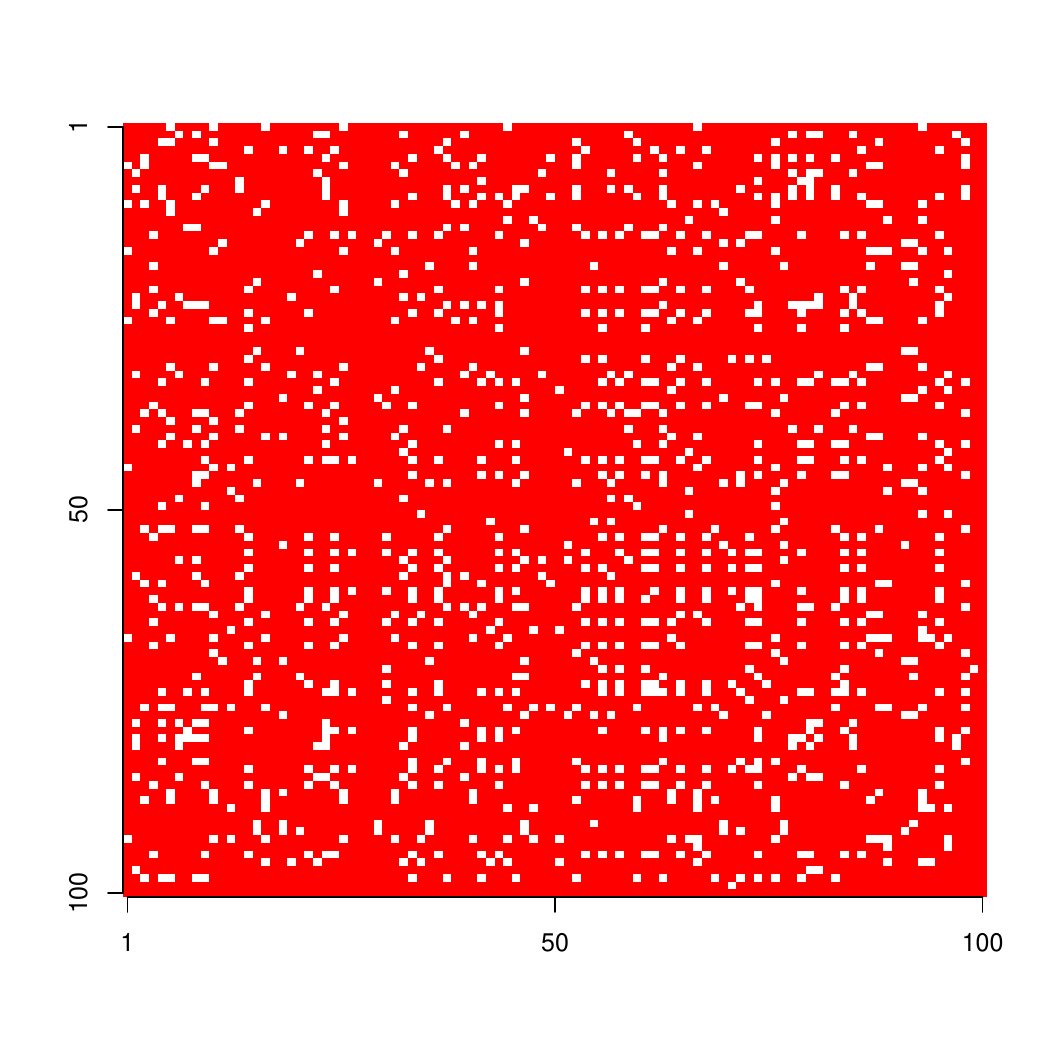}
						\includegraphics[width=0.24\textwidth]{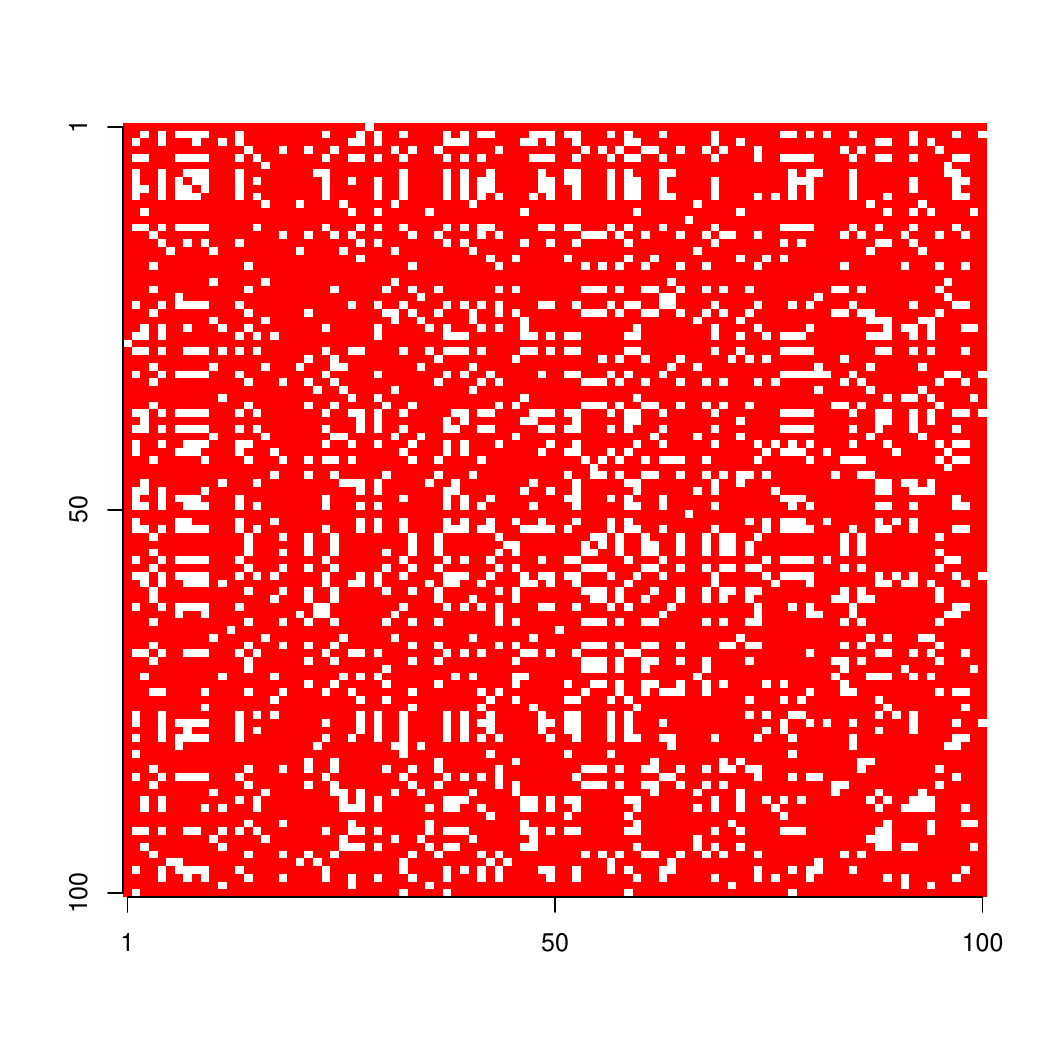}
							\includegraphics[width=0.24\textwidth]{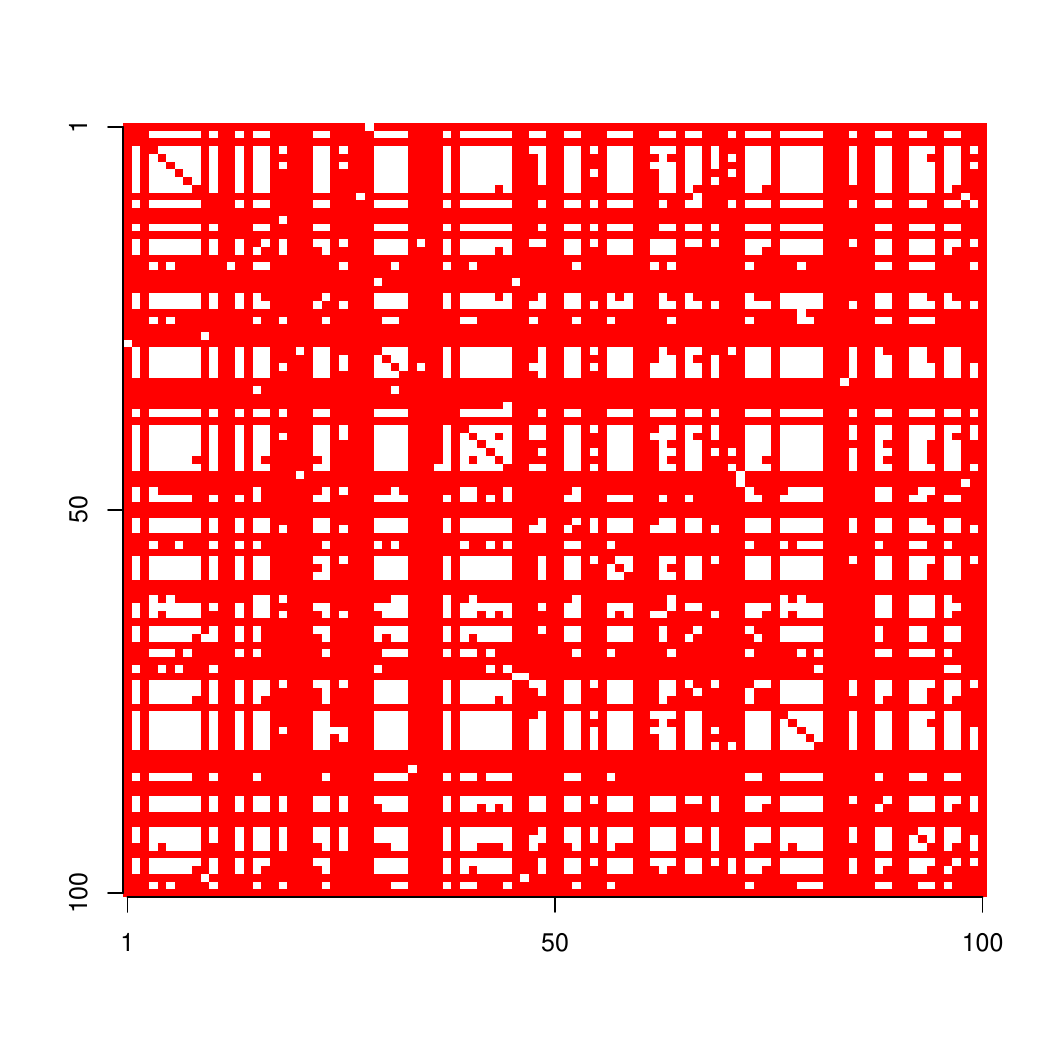}
				\includegraphics[width=0.24\textwidth]{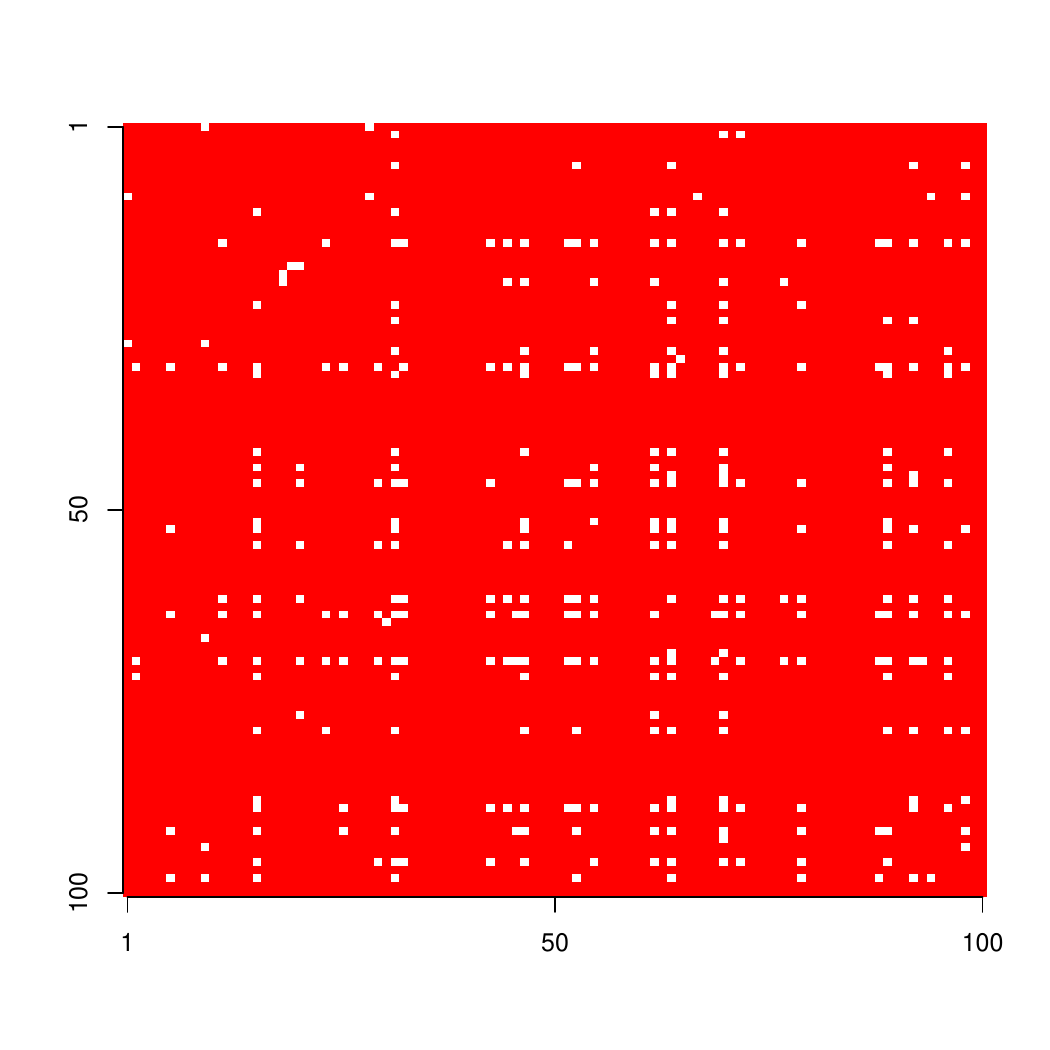}
\includegraphics[width=0.24\textwidth]{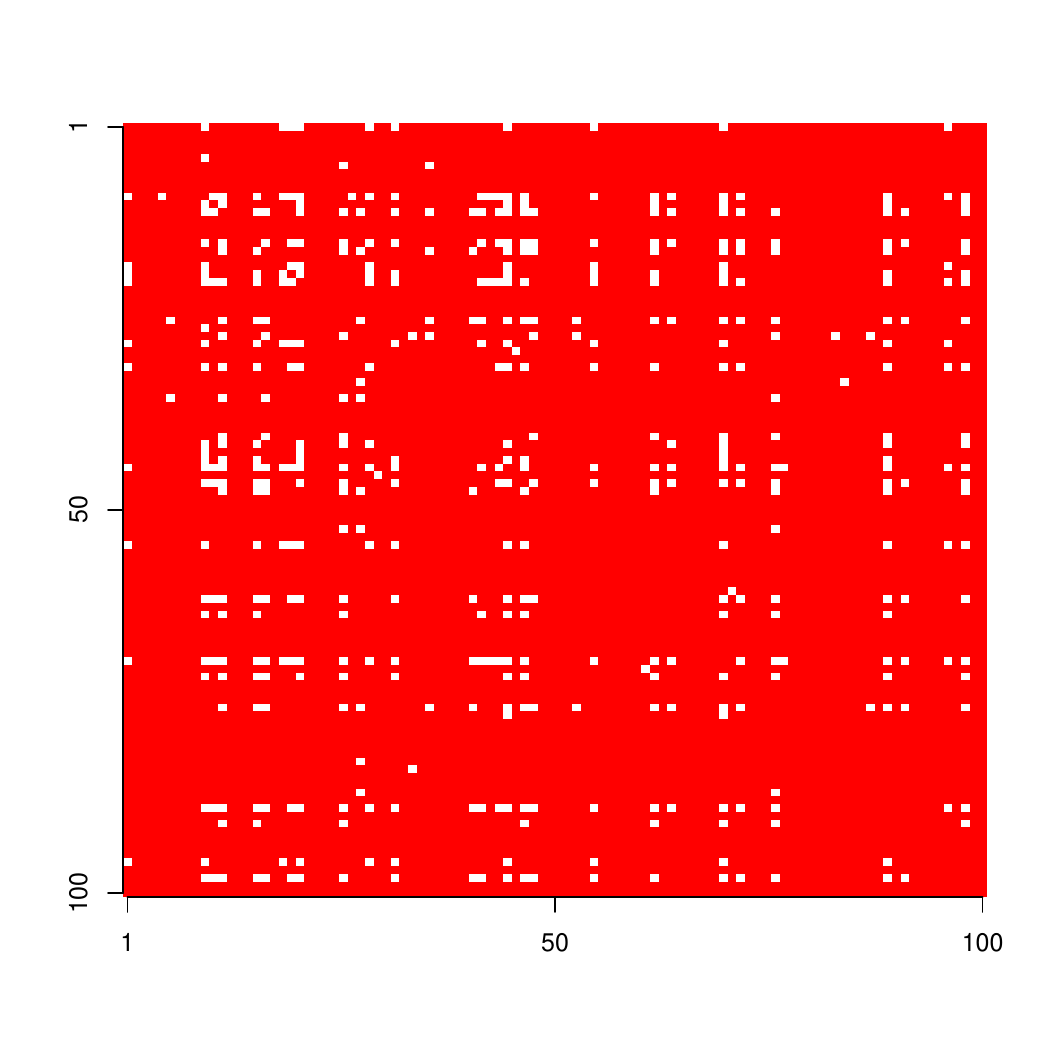}
\includegraphics[width=0.24\textwidth]{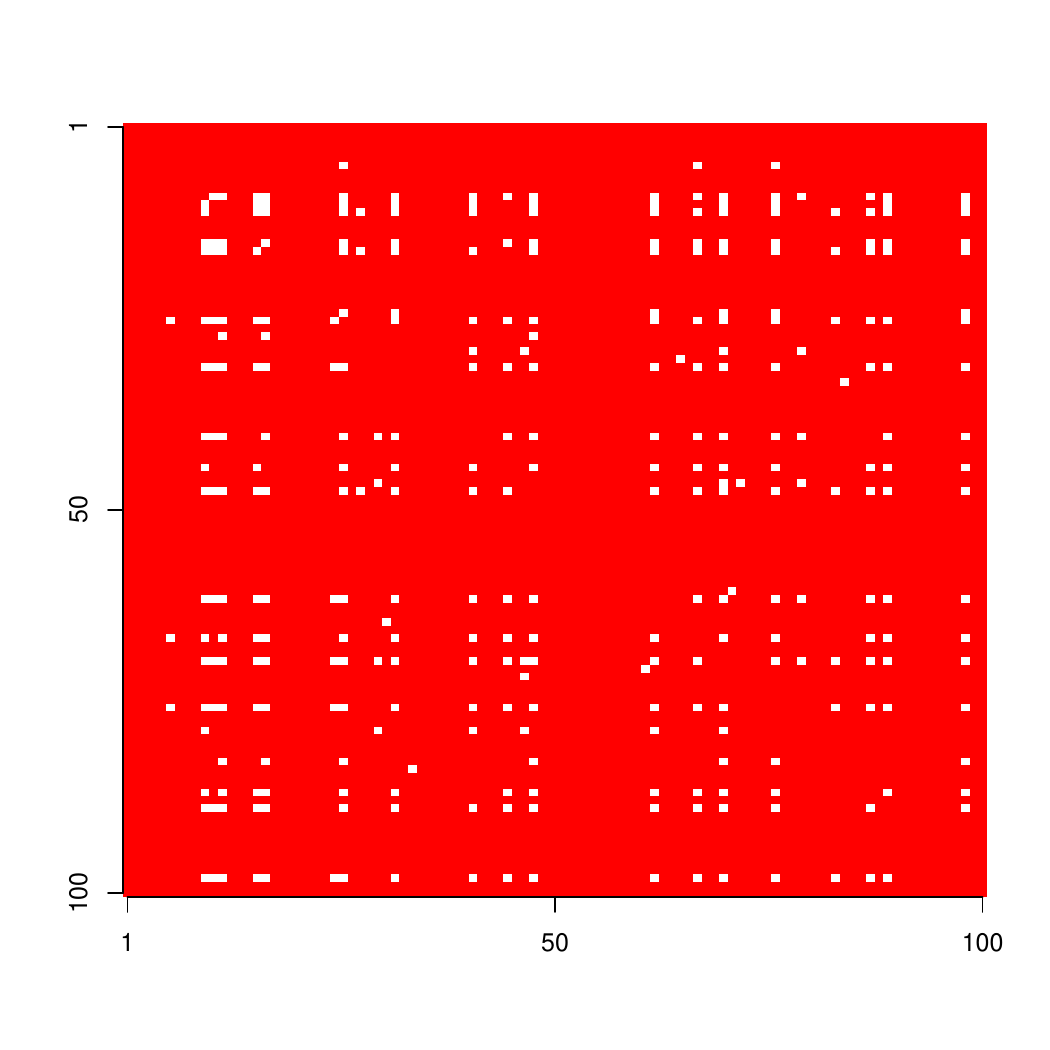}
\includegraphics[width=0.24\textwidth]{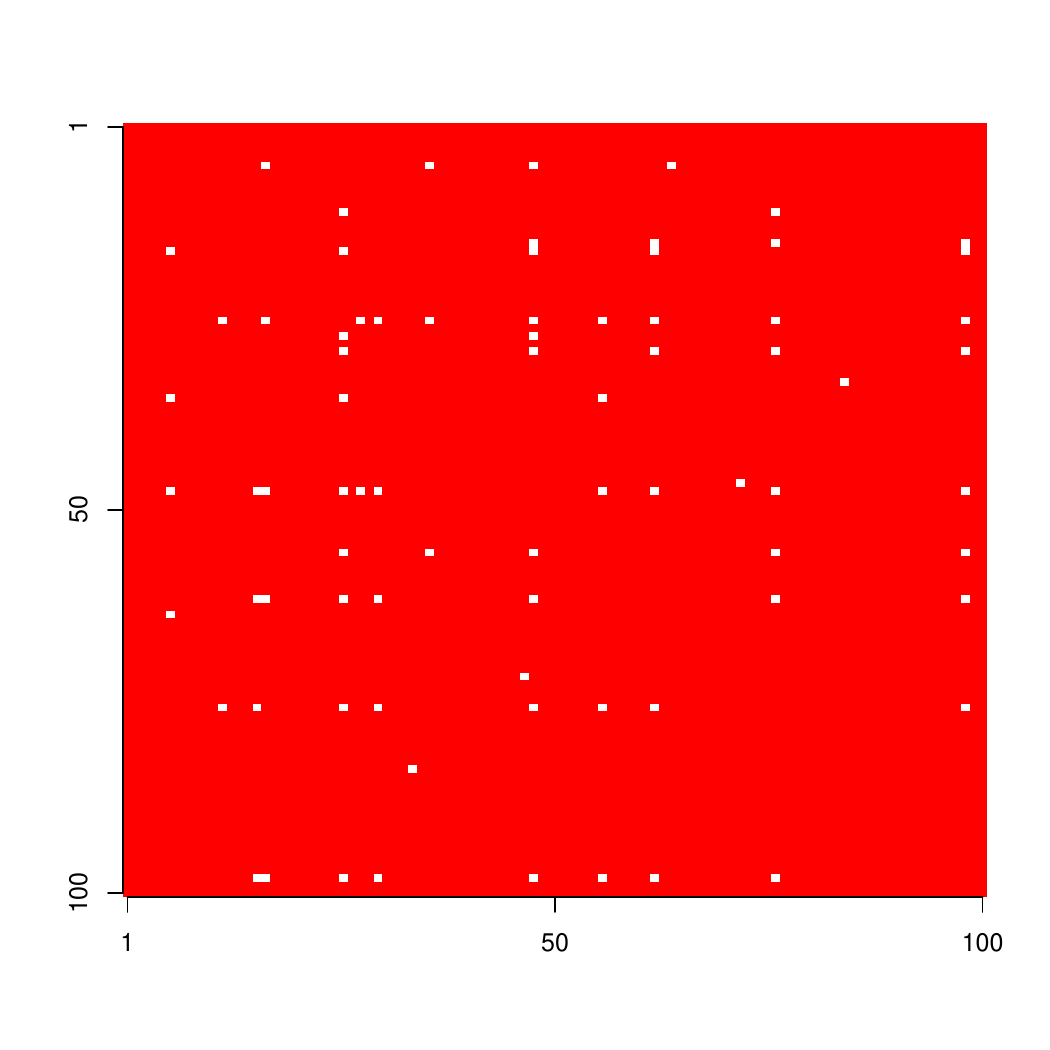}							
		\caption{\label{fig3} Examples  of  adjacency  matrices, down-sampled to  a $100\times 100$, between  the  change points  estimated  by  	NonPar-RDPG-CPD  in  the zebrafish example.  From  left  to  right and from top to bottom, the   first two rows of panels correspond to  $t = 3, 7,  15,  32,  40,  45, 52, 60, 65, 75, 80$ and  $87$.    From  left  to  right and from top to bottom, the last two rows correspond  to $t =6,8,9,10,11,12$ and $13$. }
	\end{center}
\end{figure}

\begin{table}[t!]
	\centering
	\caption{\label{tab4} Scenario 4}
	\medskip
	\setlength{\tabcolsep}{18pt}
	\begin{small}
		\begin{tabular}{ rrrrrr}
			\hline
Method &	$n$ &  $\varepsilon$  & $\vert  K - \widehat{K}\vert$   & $d(\widehat{\mathcal{C}}| \mathcal{C})$  & $d( \mathcal{C}|\widehat{\mathcal{C}})$ \\ 
\hline
NonPar-RDPG-CPD & 300 & 0.3 & \textbf{0.72}  &35.0&\textbf{12.0}\\
NBS	& 300& 0.3  &19.4& \textbf{1.0} &43.0  \\			
MNBS	& 300& 0.3  &   2.0 &  $\inf$  &  $-\inf$ \\	
\hline	
NonPar-RDPG-CPD & 200 & 0.3 & \textbf{0.84} &40.0& \textbf{10.0}\\
NBS	& 200& 0.3  & 19.4& \textbf{1.0}  &43.0  \\	
MNBS	& 200& 0.3  &2.0 &$\inf$ &$-\inf$  \\			
\hline
NonPar-RDPG-CPD & 100 & 0.3 & \textbf{1.0}   & 30.0 & 20.0\\
NBS	& 100& 0.3  & 9.44&   \textbf{3.0} &  41.0 \\
MNBS	& 100& 0.3  &2.0 &$\inf$ &$-\inf$ \\								
\hline
NonPar-RDPG-CPD & 300 & 0.15 &  \textbf{0.8} & 34.0  &  \textbf{17.0} \\
NBS	& 300& 0.15  &20.24 &   \textbf{1.0}  & 43.0 \\	
MNBS	& 300& 0.15  &2.0 &$\inf$ &$-\inf$ \\			
\hline
NonPar-RDPG-CPD & 200 & 0.15 & \textbf{0.96}  & 40.0     &  \textbf{11.0} \\
NBS	& 200& 0.15  & 17.0 &   \textbf{1.0}   & 43.0\\
MNBS	& 200& 0.15   &2.0 &$\inf$ &$-\inf$\\			
\hline
NonPar-RDPG-CPD & 100 & 0.15 & \textbf{0.84}  &  34    &  \textbf{18.0} \\
NBS	& 100& 0.15  &10.64& \textbf{1.0} &41.0 \\
MNBS	& 100& 0.15   &2.0 &$\inf$ &$-\inf$\\	
\hline
NonPar-RDPG-CPD & 300 & 0.05 & \textbf{0.80}  &33.0  &  \textbf{17.0} \\
NBS	& 300& 0.05  &   20.48   &   \textbf{1.0}  &  43.0 \\	
MNBS	& 300& 0.05   &2.0 &$\inf$ &$-\inf$\\			
\hline
NonPar-RDPG-CPD & 200 & 0.05 & \textbf{0.88}  & 38.0  &  \textbf{19.0} \\
NBS	& 200& 0.05  & 17.56& \textbf{1.0} & 43.0 \\
MNBS	& 200& 0.05   &2.0 &$\inf$ &$-\inf$\\			
\hline
NonPar-RDPG-CPD & 100 & 0.05 &\textbf{1.04}   &   32.0 & \textbf{16.0} \\
NBS	& 100& 0.05  &11.48&   \textbf{3.0}  &  41.0 \\
MNBS	& 100& 0.05   &2.0 &$\inf$ &$-\inf$ \\
\hline		
		\end{tabular}
	\end{small}
\end{table}

\subsection{Real data}\label{sec-real-data}

Our goal is to estimate change points in the context of the neuronal activity in larval zebrafish.  The data consist of simultaneous whole-brain neuronal activity data at near single cell resolution \citep{prevedel2014simultaneous}. The original data format is a matrix of size $5379 \times 5000$.  This corresponds to the neural activity of 5379 neurons over   5000 frames, where one second in time corresponds to 20 frames.

To construct the final sequence of networks, we proceed as in \cite{park2015anomaly}.  Specifically, we first remove artificial neurons leaving us with a $5105 \times 5000$ matrix.  Then we bin the data into 100 non-overlapping periods.  Each period corresponds to 2.5 seconds of the original data.  The resulting time series is then $Z(t) \in \mathbb{R}^{5105 \times  50}$ for $t \in \{1, \ldots, 100\}$.  Following \cite{lyzinski2017fast}, we finally construct the adjacency matrices $A(t) \in \mathbb{R}^{5105 \times 5105 }$ as
	\[
		A_{i,j}(t) = \mathbbm{1}\{\text{corr}(Z_i(t), Z_j(t))> 0.7\}, \quad t = 1, \ldots, T,
	\]
	where $T = 100$.

With the  time  series  $\{A(t)\}_{t=1}^T$ in hand, we  proceed  to  run  change point  detection with Algorithm \ref{algorithm:WBS}.  The implementation  details  are the same  as  those  in   Section  \ref{sec:simulations}. However,  to facilitate  computations  at every  instance  of time  we  randomly  sample   800  nodes in the network and work with a down-sampled  version of $A(t)$. After running  our method, we  estimate  change points  at   $t = 5, 10, 29, 36, 42, 50, 57, 62, 71, 79, 85$,  and  $89$. In the original 250  seconds  time stamp, the  changes correspond to 12.5  25.0,  72.5,  90.0, 105.0, 125.0, 142.5, 155.0, 177.5, 197.5, 212.5,  and  222.5  seconds. Simple  inspection suggests  that our estimated  change points  are  in agreement with   the  extracted intensity signal of Ca2+ fluorescence using spatial filters in   Figure   3 (c) in \cite{prevedel2014simultaneous}.  As remarked in \cite{park2015anomaly},  a lab scientist induced a change-point  at the 16th
second, by giving an olfactory stimulus to the zebrafish. In  the scale of our time series  $\{A(t)\}_{t=1}^T$, this change   corresponds  to $t = 6$  which   seems  to  be  captured by our algorithm  that detected  a change point  at  $t=5$.

We also  considered change point detection with the algorithm NBS  \citep{wang2018optimal}. The set  of  estimated   change points is roughly the same  to that estimated by 	NonPar-RDPG-CPD: 10, 14, 22, 26, 32, 36, 42, 50, 58, 62, 66, 72, 80, and 90. One important difference, however,  is that  NBS did not detect a change point  near  $t=6$, the change point  created  by the lab scientist. We also  tried the MNBS method \citep{zhao2019change},  but this only  detected  changes at 14, 45, 66, 80. Tuning parameters of both  NBS and MNBS are chosen as described in Section \ref{sec:simulations}.

Finally,  we  have included  \Cref{fig3}  which shows   down-sampled  versions of $A(t)$ for values  of $t$  between   estimated  change points. This reinforces  our intuition that  the     structural breaks  estimated with NonPar-RDPG-CPD  are meaningful.

\section{Discussions}\label{sec-conclusion}

In this paper, we have studied the offline change point localization problem in a sequence of dependent nonparametric random dot product graphs.  We allow for a weakly dependent process along the time and introduce the dependence within networks via latent positions.  In fact, conditional on the latent positions, the edges within a network are independent and one may wish to further allow for dependence among edges conditional on the latent positions.  We remark that this is technically feasible - one can incorporate a weak dependence version of Bernstein's inequality (Lemma~\ref{thm-delvon}) in the estimation of the latent positions (Lemma~\ref{lem-three-e}).  Such deployment requires a data generating mechanism characterizing the dependence among edges, but without a natural distance among edges, we refrain our pursuit on this direction.

\acks{The authors are very grateful to Joshua Agterberg for pointing out a mistake in the proof in a previous version and to the Editor and all referees for constructive comments which substantially improved the paper.  OHMP is partially funded by NSF DMS-2015489.  YY is partially funded by EPSRC EP/V013432/1.  CEP was supported in part by the Defense Advanced Research Projects Agency under the D3M program administered through contract FA8750-17-2-0112, the National Science Foundation HDR TRIPODS 1934979, and by funding from Microsoft Research.  This project began while YY was Lecturer and CEP was Heilbronn Distinguished Visitor in Data Science at the University of Bristol.}

\newpage

\appendix
\section{Technical details of \Cref{sec:jumps}}\label{sec-app-jumps}

\begin{proof}[Proof of Lemma \ref{lem-dist-moment}]
	For any $i, j \in \{1, \ldots, n\}$, $i \neq j$, it holds that 
	\[
	\mathbb{P}\{A_{ij} | X_i, X_j\} = X_i^{\top} X_j = X_i^{\top} U^{\top} UX_j,
	\]
	for any orthogonal operator $U \in \mathbb{R}^{d \times d}$.  In this proof, by the equivalence in terms of the distributions $F$ and $\widetilde{F}$, we mean the equivalence up to a rotation, which is detailed in Definition \ref{def-equivalence}.  Without loss of generality, if a rotation is needed, we omit it in the notation.
	
	We divide this proof into two cases: (a) $d = 1$ and (b) $d > 1$.  
	
	\vskip 1mm
	\noindent \textbf{(a) $p = 1$.}
	
	Since the entries of $A$ and $\widetilde{A}$ are Bernoulli random variables, they only take values in $\{0, 1\}^{n \times n}$.  For any symmetric matrix $v \in \{0, 1\}^{n\times n}$, we have
	\begin{align}
	\mathbb{P}\{A = v\} & = \mathbb{E}\left\{\mathbb{E}\left(\prod_{i = 1}^{n-1} \prod_{j = i+1}^n \mathbbm{1}\{A_{ij} = v_{ij}\} \Bigg | \{X_l\}_{l=1}^n\right)\right\}	\nonumber \\
	& = \mathbb{E}\left[\prod_{i = 1}^{n-1} \prod_{j = i+1}^n \left\{(X_i  X_j)v_{ij} + (1 - X_i  X_j)(1 - v_{ij})\right\}\right]. \label{eq-pav-exp}
	\end{align}
	
	If $\mathcal{L} = \widetilde{\mathcal{L}}$, then we have the following.  
	
	\begin{itemize}
		\item If $v_{ij} \equiv 1$, then 
		\begin{align*}
		\eqref{eq-pav-exp} = \mathbb{E}\left[ \prod_{i = 1}^{n-1} \prod_{j = i+1}^n (X_i X_j)\right] = \left\{\mathbb{E} (X_1^{n-1})\right\}^n,
		\end{align*}
		which implies that $\mathbb{E}_F(X_1^{n-1}) = \mathbb{E}_{\widetilde{F}}(\widetilde{X}_1^{n-1})$.  Note that in order to have an edge, $n \geq 2$, which implies that $n - 1 \geq 1$.
		
		\item If there is one and only one pair $(i, j)$, $i < j$, such that $v_{ij} = v_{ji} = 0$, and $v_{kl} = 1$, $(k, l) \notin \{(i, j), \, (j, i)\}$, then without loss of generality, we let $(i, j) = (1, 2)$.  If $n = 2$, then 
		\[
		\eqref{eq-pav-exp} = 1 - \{\mathbb{E}(X_1)\}^2,
		\]
		which implies that $\mathbb{E}_F(X_1^{n-1}) = \mathbb{E}_{\widetilde{F}}(\widetilde{X}_1^{n-1})$.  
		
		If $n \geq 3$, then 
		\begin{align*}
		\eqref{eq-pav-exp} & = \mathbb{E}\left[ \prod_{i = 3}^{n-1} \prod_{j = i+1}^n (X_i X_j) \cdot \prod_{r = 1}^ 2\prod_{l = 3}^n (X_r X_l) \cdot (1 - X_1X_2)\right] \\
		& = \left\{\mathbb{E}(X_1^{n-2})\right\}^2 \left\{\mathbb{E}(X_1^{n-1})\right\}^{n-2} - \left\{\mathbb{E} (X_1^{n-1})\right\}^n,
		\end{align*}
		which implies $\mathbb{E}_F(X_1^{n-2}) = \mathbb{E}_{\widetilde{F}}(\widetilde{X}_1^{n-2})$.
		
		\item If $n \geq 3$, then for $k \in \{2, \ldots, n-1\}$	, without loss of generality, let $v_{1j} = v_{j1} = 0$, $j \in \{2, \ldots, k+1\}$, and $v_{rs} = v_{sr} = 1$ otherwise.  We have that
		\begin{align}
		\eqref{eq-pav-exp} & = \mathbb{E}\left[\prod_{i = k+2}^{n-1} \prod_{j = i+1}^n (X_i X_j) \cdot \prod_{l = 1}^{k+1}\prod_{i = k+2}^n (X_lX_i) \cdot \prod_{r = 2}^{k+1} (1 - X_1X_r) \right]	 \nonumber \\
		& = \left\{\mathbb{E}(X_1^{n-1})\right\}^{n-k-1} \mathbb{E}\left[\prod_{l = 1}^{k+1} X_l^{n-k-1} \cdot \prod_{r = 2}^{k+1} (1 - X_1X_r) \right] \nonumber \\
		& = \left\{\mathbb{E}(X_1^{n-1})\right\}^{n-k-1} \sum_{r = 0}^k {k \choose r} (-1)^r \mathbb{E}(X_1^{n-k-1+r}) \left[ \mathbb{E}(X_1^{n-k})\right]^r.  \label{eq-general-k-case}
		\end{align}
		
		Note that, if $k = 2$, then the summands in \eqref{eq-general-k-case} include moments $n-1$, $n-2$ and $n-3$.  We have already shown that $\mathbb{E}_F(X_1^{n-1}) = \mathbb{E}_{\widetilde{F}}(\widetilde{X}_1^{n-1})$ and $\mathbb{E}_F(X_1^{n-2}) = \mathbb{E}_{\widetilde{F}}(\widetilde{X}_1^{n-2})$, therefore \eqref{eq-general-k-case} implies that $\mathbb{E}_F(X_1^{n-3}) = \mathbb{E}_{\widetilde{F}}(\widetilde{X}_1^{n-3})$.  
		
		\item By induction, for $n > k_0$ and $k_0 \geq 3$, if it holds that $\mathbb{E}_F(X_1^{n-s}) = \mathbb{E}_{\widetilde{F}}(\widetilde{X}_1^{n-s})$, $s = 1, \ldots, k_0$, then we have $\mathbb{E}_F(X_1^{n-k_0-1}) = \mathbb{E}_{\widetilde{F}}(\widetilde{X}_1^{n-k_0-1})$, due to the fact that the summands in \eqref{eq-general-k-case} include moment $n-s$, $s = 1, \ldots, k_0 + 1$.
	\end{itemize}
	
	We conclude that if $\mathcal{L} = \widetilde{\mathcal{L}}$, then $\mathbb{E}_F(X_1^{k}) = \mathbb{E}_{\widetilde{F}}(\widetilde{X}_1^{k})$, $k = 1, \ldots, n-1$. \\
	
	If $\mathbb{E}_F(X_1^k) = \mathbb{E}_{\widetilde{F}}(\widetilde{X}_1^k)$, $k = 1, \ldots, n-1$, then it follows from that for any $v$,
	\begin{align*}
	& \eqref{eq-pav-exp} = \sum_{l = 0}^{\sum_{i < j}\mathbbm{1}\{v_{ij} = 0\}} {\sum_{i < j}\mathbbm{1}\{v_{ij} = 0\} \choose l}(-1)^{\sum_{i < j}\mathbbm{1}\{v_{ij} = 0\} - l} \\
	& \hspace{2cm} \times \mathbb{E}\left\{\prod_{v_{ij} = 1}X_i X_j \prod_{\substack{r = 1 \\ v_{i_rj_r} = 0}}^{\sum_{i < j}\mathbbm{1}\{v_{ij} = 0\} - l}X_{i_r}X_{j_r}\right\},
	\end{align*}
	which is a function solely of $\mathbb{E}_F(X_1^{k})$, $k = 1, \ldots, n-1$.  We, therefore, have that $\mathcal{L} = \widetilde{\mathcal{L}}$.
	
	\vskip 3mm
	\noindent \textbf{(b) $d > 1$.}
	
	Since the entries of $A$ and $\widetilde{A}$ are Bernoulli random variables, they only take values in $\{0, 1\}^{n \times n}$.  For any symmetric matrix $v \in \{0, 1\}^{n\times n}$, we have
	\begin{align}
	\mathbb{P}\{A = v\} & = \mathbb{E}\left\{\mathbb{E}\left(\prod_{i = 1}^{n-1} \prod_{j = i+1}^n \mathbbm{1}\{A_{ij} = v_{ij}\} \Bigg | \{X_l\}_{l=1}^n\right)\right\}	\nonumber \\
	& = \mathbb{E}\left[\prod_{i = 1}^{n-1} \prod_{j = i+1}^n \left\{\left(\sum_{k = 1}^d X_{i, k} X_{j, k}\right)v_{ij} + \left(1 - \sum_{k = 1}^d X_{i, k} X_{j, k}\right)(1 - v_{ij})\right\}\right]. \label{eq-pav-exp-2}
	\end{align}
	
	If $\mathcal{L} = \widetilde{\mathcal{L}}$, then we have the following.  
	
	\begin{itemize}
		\item 	If $v_{ij} \equiv 1$, then 
		\begin{align}
		\eqref{eq-pav-exp-2} & = \mathbb{E}\left[ \prod_{i = 1}^{n-1} \prod_{j = i+1}^n \left(\sum_{k = 1}^d X_{i, k} X_{j, k} \right)\right] \nonumber \\
		& = \mathbb{E}\left[ \prod_{j = 2}^n \left(\sum_{k = 1}^p X_{1, k} X_{j, k}\right) \cdot \prod_{i=2}^{n-1}\prod_{j = i+1}^n \left(\sum_{k = 1}^d X_{i, k}X_{j, k}\right)\right] \nonumber \\
		& = \mathbb{E}\left\{\left[\sum_{k_2, \ldots, k_{n} = 1}^d \left(\prod_{l = 2}^n X_{1, k_l}\right) \cdot \left(\prod_{l = 2}^n X_{l, k_l}\right)\right] \cdot \left[\prod_{i = 2}^{n-1} \prod_{j = i+1}^n \left(\sum_{k = 1}^d X_{i, k}X_{j, k}\right)\right] \right\} \nonumber \\
		& = \sum_{k_2, \ldots, k_n = 1}^d \mathbb{E}\left(\prod_{l = 2}^n X_{1, k_l}\right) \mathbb{E}\left\{ \left(\prod_{l = 2}^n X_{l, k_l}\right) \cdot \left[\prod_{i = 2}^{n-1} \prod_{j = i+1}^n \left(\sum_{k = 1}^d X_{i, k}X_{j, k}\right)\right] \right\}, \label{eq-pav-exp-3}
		\end{align}
		where the third identity follows from the independence assumption.  Note that for any $(k_2, \ldots, k_n) \in \{1, \ldots, p\}^{\otimes (n-1)}$, the term 
		\[
		\mathbb{E}\left\{ \left(\prod_{l = 2}^n X_{l, k_l}\right) \cdot \left[\prod_{i = 2}^{n-1} \prod_{j = i+1}^n \left(\sum_{k = 1}^d X_{i, k}X_{j, k}\right)\right] \right\}
		\]
		in \eqref{eq-pav-exp-3} does not involve $X_1$, and the term 
		\[
		\mathbb{E}\left(\prod_{l = 2}^n X_{1, k_l}\right) 
		\]
		includes all possible terms of the form 
		\begin{equation}\label{eq-pav-exp-4}
		\mathbb{E}\left(\prod_{l = 1}^d X_{1, l}^{k_l}\right), \quad \sum_{l = 1}^d k_l = n-1, \quad k_l \geq 0, \, l \in \{1, \ldots, d\}.
		\end{equation}
		Due to the exchangeablility, we conclude that \eqref{eq-pav-exp-2} is solely a function of polynomials of \eqref{eq-pav-exp-4}.
		
		If $n = 2$, then due to Definition \ref{def-1}, we have that $\mathcal{L} = \widetilde{\mathcal{L}}$ implies that 
		\[
		\mathbb{E}\left(\prod_{l = 1}^d X_{1, l}^{k_l}\right) = \mathbb{E}\left(\prod_{l = 1}^d \widetilde{X}_{1, l}^{k_l}\right), \quad \sum_{l = 1}^d k_l = n-1, \quad k_l \geq 0, \, l \in \{1, \ldots, d\}.
		\]
		
		\item If $n \geq 3$, then we prove by induction.  Assume that 	
		\[
		\mathbb{E}\left(\prod_{l = 1}^p X_{1, l}^{k_l}\right) = \mathbb{E}\left(\prod_{l = 1}^p \widetilde{X}_{1, l}^{k_l}\right), \quad \sum_{l = 1}^p k_l = n-k, \ldots, n-1, \quad k_l \geq 0, \, l \in \{1, \ldots, p\},
		\]
		where $ n-1 \geq n - k \geq 2$.  We now proceed to prove that 
		\begin{equation}\label{eq-pav-exp-5}
		\mathbb{E}\left(\prod_{l = 1}^p X_{1, l}^{k_l}\right) = \mathbb{E}\left(\prod_{l = 1}^p \widetilde{X}_{1, l}^{k_l}\right), \quad \sum_{l = 1}^p k_l = n-k-1, \ldots, n-1, \quad k_l \geq 0, \, l \in \{1, \ldots, d\}.
		\end{equation}
		To show this, we assume that $v_{1j} = v_{j1} = 0$, $j \in \{2, \ldots, k+1\}$, and $v_{rs} = 1$ otherwise.  We have that
		\begin{align*}
		\eqref{eq-pav-exp-2} & = \mathbb{E}\left[\prod_{j = 2}^{k+1} \left(1 - \sum_{s = 1}^s X_{1, s}X_{j, s}\right) \cdot \prod_{l = k+2}^n \left(\sum_{s = 1}^d X_{1, s}X_{l, s}\right) \cdot \prod_{i = 2}^{n-1} \prod_{r = i+1}^n \left(\sum_{s=1}^d X_{i, s}X_{r, s}\right)\right] \\
		& = (-1)^k \mathbb{E}\left\{\prod_{l = k+2}^n \left(\sum_{s = 1}^d X_{1, s}X_{l, s}\right) \cdot \prod_{i = 2}^{n-1} \prod_{r = i+1}^n \left(\sum_{s=1}^d X_{i, s}X_{r, s}\right)\right\} + f(X)   \\
		& = (-1)^k \sum_{s_{k+2}, \ldots, s_n = 1}^d \mathbb{E}\left(\prod_{l = k+2}^n X_{1, s_l}\right) \mathbb{E}\Bigg\{\prod_{l = k+2}^n \left(\sum_{s = 1}^d X_{l, s} \right) \\
		& \hspace{3cm} \times \prod_{i = 2}^{n-1} \prod_{r = i+1}^n \left(\sum_{s=1}^d X_{i, s}X_{r, s}\right) \Bigg\}+ f(X),   
		\end{align*}
		where $f(X)$ is solely a function of 
		\[
		\mathbb{E}\left(\prod_{l = 1}^d X_{1, l}^{k_l}\right), \quad \sum_{l = 1}^d k_l = n-k, \ldots, n-1, \quad k_l \geq 0, \, l \in \{1, \ldots, d\}.
		\]
		Note that 
		\[
		\sum_{s_{k+2}, \ldots, s_n = 1}^d \mathbb{E}\left(\prod_{l = k+2}^n X_{1, s_l}\right) 
		\]
		is a function of 
		\[
		\mathbb{E}\left(\prod_{l = 1}^d X_{1, l}^{k_l}\right), \quad \sum_{l = 1}^d k_l = n-k-1, \quad k_l \geq 0, \, l \in \{1, \ldots, d\}.
		\]
		Therefore we have shown \eqref{eq-pav-exp-5}.
	\end{itemize}
	
	To this end, we have that $\mathcal{L} = \widetilde{\mathcal{L}}$ implies that
	\begin{equation}\label{eq-pav-exp-6}
	\mathbb{E}\left(\prod_{l = 1}^d X_{1, l}^{k_l}\right) = \mathbb{E}\left(\prod_{l = 1}^d \widetilde{X}_{1, l}^{k_l}\right), \quad \sum_{l = 1}^d k_l = 1, \ldots, n-1, \quad k_l \geq 0, \, l \in \{1, \ldots, d\}.
	\end{equation}
	
	To show that \eqref{eq-pav-exp-6} implies that $\mathcal{L} = \widetilde{\mathcal{L}}$, we notice that for any $v$,
	\begin{align*}
	& \eqref{eq-pav-exp-2} = \sum_{l = 0}^{\sum_{i < j}\mathbbm{1}\{v_{ij} = 0\}} (-1)^l  \Bigg[ \Bigg\{\sum_{\substack{\{(i_1, j_1), \ldots, (i_l, j_l)\} \\ \in \{(i, j):\, v_{ij} = 0,\, i < j\}}} \left[\prod_{r = 1}^l \left\{ \sum_{k = 1}^d \left(X_{i_l, k} X_{j_l, k}\right)\right\} \right] \Bigg\}\\
	& \hspace{4cm} \times \prod_{(i, j): \, v_{ij} = 1,\, i < j} \left( \sum_{k = 1}^d X_{i, k} X_{j, k}\right)\Bigg],
	\end{align*}
	which is solely a function of 
	\[
	\mathbb{E}\left(\prod_{l = 1}^d X_{1, l}^{k_l}\right), \quad \sum_{l = 1}^d k_l = 1, \ldots, n-1, \quad k_l \geq 0, \, l \in \{1, \ldots, d\}.
	\]
	The  final claim holds.
	
\end{proof}

\begin{proof}[Proof of Lemma \ref{lem-one-network}]
	For simplicity, we assume $n$ is an even number.  Let $\mathcal{O} = \{(i, n/2+i), \, i = 1, \ldots, n/2\}$.  Let
	\[
	z_* \,\in \, \argsup_{z \in [0, 1]} |G(z) - \widetilde{G}(z)|.
	\]
	Note that 
	\begin{align}
	& \left|\sqrt{\frac{2}{n}} \sum_{(i, j) \in \mathcal{O}} \left(\mathbbm{1}\{Y_{ij} \leq z_*\} -  \mathbbm{1}\{\widetilde{Y}_{ij} \leq z_*\}\right)\right| \nonumber \\
	= & \Bigg|\sqrt{\frac{2}{n}} \sum_{(i, j) \in \mathcal{O}} \left\{\left(\mathbbm{1}\{Y_{ij} \leq z_*\} - \mathbb{E}\left[\mathbbm{1}\{Y_{ij} \leq z_*\}\right]\right) - \left(\mathbbm{1}\{\widetilde{Y}_{ij} \leq z_*\} - \mathbb{E}\left[\mathbbm{1}\{\widetilde{Y}_{ij} \leq z_*\}\right]\right) \right\} \nonumber \\
	& \hspace{2cm} + \sqrt{\frac{n}{2}} \left\{\mathbb{E}\left[\mathbbm{1}\{Y_{ij} \leq z_*\}\right] - \mathbb{E}\left[\mathbbm{1}\{\widetilde{Y}_{ij} \leq z_*\}\right]\right\} \Bigg| \nonumber \\
	\geq & \sqrt{\frac{n}{2}} \left|\mathbb{E}\left[\mathbbm{1}\{Y_{ij} \leq z_*\}\right] - \mathbb{E}\left[\mathbbm{1}\{\widetilde{Y}_{ij} \leq z_*\}\right]\right| - \left| \sqrt{\frac{2}{n}} \sum_{(i, j) \in \mathcal{O}} \left(\mathbbm{1}\{Y_{ij} \leq z_*\} - \mathbb{E}\left[\mathbbm{1}\{Y_{ij} \leq z_*\}\right]\right) \right| \nonumber \\
	& \hspace{2cm} -  \left| \sqrt{\frac{2}{n}} \sum_{(i, j) \in \mathcal{O}} \left(\mathbbm{1}\{\widetilde{Y}_{ij} \leq z_*\} - \mathbb{E}\left[\mathbbm{1}\{\widetilde{Y}_{ij} \leq z_*\}\right]\right) \right| \nonumber \\
	= & \kappa_0\sqrt{n/2} - \left| \sqrt{\frac{2}{n}} \sum_{(i, j) \in \mathcal{O}} \left(\mathbbm{1}\{Y_{ij} \leq z_*\} - \mathbb{E}\left[\mathbbm{1}\{Y_{ij} \leq z_*\}\right]\right) \right| \nonumber \\
	& \hspace{2cm} -  \left| \sqrt{\frac{2}{n}} \sum_{(i, j) \in \mathcal{O}} \left(\mathbbm{1}\{\widetilde{Y}_{ij} \leq z_*\} - \mathbb{E}\left[\mathbbm{1}\{\widetilde{Y}_{ij} \leq z_*\}\right]\right) \right|. \label{eq-pf-lem-one-network-1}
	\end{align}
	
	Next, it follows from Hoeffding's inequality that
	\begin{align}
	&\mathbb{P}\Bigg\{\max\Bigg\{\left| \sqrt{\frac{2}{n}} \sum_{(i, j) \in \mathcal{O}} \left(\mathbbm{1}\{Y_{ij} \leq z_*\} - \mathbb{E}\left[\mathbbm{1}\{Y_{ij} \leq z_*\}\right]\right) \right|, \nonumber \\
	& \hspace{2cm} \left| \sqrt{\frac{2}{n}} \sum_{(i, j) \in \mathcal{O}} \left(\mathbbm{1}\{\widetilde{Y}_{ij} \leq z_*\} - \mathbb{E}\left[\mathbbm{1}\{\widetilde{Y}_{ij} \leq z_*\}\right]\right) \right|\Bigg\} > \sqrt{\log(n)}\Bigg\} \leq 2n^{-4}.	\label{eq-pf-lem-one-network-2}
	\end{align}
	
	Combining \eqref{eq-pf-lem-one-network-1} and \eqref{eq-pf-lem-one-network-2}, we have that with probability at least $1 - 2n^{-4}$, 
	\begin{equation}\label{eq-pf-lem-one-network-3}
	\left|\sqrt{\frac{2}{n}} \sum_{(i, j) \in \mathcal{O}} \left(\mathbbm{1}\{Y_{ij} \leq z_*\} -  \mathbbm{1}\{\widetilde{Y}_{ij} \leq z_*\}\right)\right| \geq  \kappa_0\sqrt{n/2} - 2\sqrt{\log(n)}.
	\end{equation}
	
	We then prove by contradiction.  If 	$\mathcal{L} = \widetilde{\mathcal{L}}$, then it follows from Hoeffding's inequality that
	\begin{equation}\label{eq-pf-lem-one-network-4}
	\mathbb{P}\left\{\left|\sqrt{\frac{2}{n}} \sum_{(i, j) \in \mathcal{O}} \left(\mathbbm{1}\{Y_{ij} \leq z_*\} -  \mathbbm{1}\{\widetilde{Y}_{ij} \leq z_*\}\right)\right| \leq \sqrt{\log(n)}\right\} \geq 1 - 2n^{-4}.
	\end{equation}
	Due to \Cref{assume-comp}, \eqref{eq-pf-lem-one-network-3} and \eqref{eq-pf-lem-one-network-4} contradict with each other, which implies that $\mathcal{L} \neq \widetilde{\mathcal{L}}$.
	
\end{proof}

\section{Large probability events}\label{sec-app-prop}

Define
	\[
		\Delta^t_{s, e}(z) = \sum_{k = s+1}^e w_k \sum_{(i, j) \in \mathcal{O}} \left(\mathbbm{1}\{\widehat{Y}^k_{ij} \leq z\} - \mathbb{E}\left\{\mathbbm{1}\{Y^k_{ij} \leq z\}\right\}\right),
	\]
	where 
	\[
		w_k = \begin{cases}
 			\sqrt{\frac{2}{n}} \sqrt{\frac{e-t}{(e-s)(t-s)}}, & k = s + 1, \ldots, t,\\
 			- \sqrt{\frac{2}{n}} \sqrt{\frac{t-s}{(e-s)(e-t)}}, & k = t + 1, \ldots, e.
 		\end{cases}
	\]

In this section, we are to show the following two events hold with probability tending to 1, as $(n \vee T) \to \infty$,
	\begin{align*}
		\mathcal{B}_1 = 	\left\{\max_{\substack{0 \leq s < t < e \leq T}} \Delta^t_{s, e} \leq C_9 \sqrt{\frac{T}{1-\rho}} \max\{d\log(n \vee T), \, d^{3/2}\sqrt{\log(n \vee T)}\}\right\}
	\end{align*}
	and
	\begin{align*}	
		\mathcal{B}_2 &= \Bigg\{\max_{0 \leq s < t < e \leq T} \sup_{z \in [0, 1]}\left|\sqrt{\frac{2}{n(e-s)}}\sum_{k = s+1}^e \sum_{(i, j) \in \mathcal{O}} \left(\mathbbm{1}\{\widehat{Y}^k_{ij} \leq z\} - \mathbb{E}\left\{\mathbbm{1}\{Y^k_{ij} \leq z\}\right\}\right)\right| \nonumber \\
		& \hspace{2cm} \leq C_9 \sqrt{\frac{T}{1-\rho}} \max\{d\log(n \vee T), \, d^{3/2}\sqrt{\log(n \vee T)}\}\Bigg\}.
	\end{align*}
	This is formally stated in Lemma \ref{lem-large-prob}.  To reach there, we denote
	\begin{align*}
		& \mathcal{E}_1 = \left\{\max_{t = 1, \ldots, T} \|U_{P_t}^{\top}(A(t) - P_t)U_{P_t}\|_{\mathrm{F}} \leq C_1 \sqrt{\log(n \vee T)}\right\}, \\
		& \mathcal{E}_2 = \left\{\max_{t = 1, \ldots, T} \|(A(t)  - P_t)U_{P_t}\|_{ 2\rightarrow \infty  } \leq C_2 \sqrt{d\log(n \vee T)}\right\}, \\
		& \mathcal{E}_3 = \left\{\max_{t = 1, \ldots, T} \|A(t)  - P_t\|_{\mathrm{op}} \leq C_3 \sqrt{n}\right\}
	\end{align*}
	and	
	\begin{align*}		
		\mathcal{E}_4 = \bigg\{2^{-1} n \min_{k = 1, \ldots, K} \mu_d^k \leq \min_{t = 1, \ldots, T} \lambda_d(P_t) \leq \max_{t = 1, \ldots, T} \lambda_1(P_t) \leq (3/2) n \max_{k = 1, \ldots, K} \mu^k_1\bigg\}, 
	\end{align*}
	where $C_1 > 4\sqrt{6}$, $C_2 > 4\sqrt{6}$, $C_3 > 0$ are universal constants. Throughout, $\|\cdot\|_{2\rightarrow \infty}$ denotes the two-to-infinity norm.  To be specific, for any matrix $M \in \mathbb{R}^{m_1 \times m_2}$,
	\[
		\|M\|_{2 \to \infty} = \max_{x \in \mathbb{R}^{m_2}: \, \|x\|_2 = 1} \|Ax\|_{\infty},
	\]
	where $\|Ax\|_{\infty}$ denotes the largest absolute value of the entries in $Ax$.

\begin{lemma}\label{lem-low-rank}
Under \Cref{assume:model-rdpg}, for any $t \in \{1, \ldots, T\}$, it holds that
	\[
		\mathbb{P}\{\lambda_{d+1}(P_t) = 0\} = 1.
	\]
\end{lemma}

\begin{proof}
For any $t \in \{1, \ldots, T\}$	, we have that
	\[
		P_t = X(t)(X(t))^{\top}.
	\]
	For any realisation of $X(t) \in \mathbb{R}^{n \times d}$, $\lambda_{d+1}(P_t) = 0$.  Thus the final claim holds.
\end{proof}

\begin{lemma}\label{lem-three-e}
Under \Cref{assume:model-rdpg}, we have that 
	\begin{align}
		& \max\left\{\mathbb{P}\left\{\mathcal{E}_1 \mid \{X(t)\}_{t = 1}^T \right\}, \, \mathbb{P}\left\{\mathcal{E}_1 \right\}\right\} \geq 1 - (n \vee T)^{-c_1}, \label{eq-E1-def}\\
		& \max\left\{\mathbb{P}\left\{\mathcal{E}_2 \mid \{X(t)\}_{t = 1}^T \right\}, \, \mathbb{P}\left\{\mathcal{E}_2 \right\}\right\} \geq 1 - (n \vee T)^{-c_2} \label{eq-E2-def}
	\end{align}
	and
	\begin{equation}\label{eq-E3-def}
		\max\left\{\mathbb{P}\left\{\mathcal{E}_3 \mid \{X(t)\}_{t = 1}^T \right\}, \, \mathbb{P}\left\{\mathcal{E}_3 \right\}\right\} \geq 1 - 4Te^{-n},
	\end{equation}
	where $c_1, c_2 > 0$ are universal constants depending on $C_1$ and $C_2$, respectively.
\end{lemma}

\begin{proof}
We start with $\mathbb{P}\left\{\mathcal{E}_1 \mid \{X(t)\}_{t = 1}^T \right\}$.  For any $(i, j) \in \{1, \ldots, d\}^{\otimes 2}$ and any $t \in \{1, \ldots, T\}$, it satisfies that
	\begin{equation}\label{eq-ua-upupua-f-2-lem2}
		[U_{P_t}^{\top}(A(t) - P_t)U_{P_t}]_{ij} = 2 \sum_{k = 1}^{n-1} \sum_{l = k+1}^n (U_{P_t})_{li}(A(t)  - P_t)_{lk}(U_{P_t})_{kj} + \sum_{k = 1}^n (U_{P_t})_{ki}(A(t)  - P_t)_{kk}(U_{P_t})_{kj}.
	\end{equation}
	For any $\varepsilon > 0$, there exists an absolute constant $c > 0$ such that
	\begin{align}
		& \mathbb{P}\left\{\left|2 \sum_{k = 1}^{n-1} \sum_{l = k+1}^n (U_{P_t})_{li}(A(t)  - P_t)_{lk}(U_{P_t})_{kj} \right| > \varepsilon \Bigg | \{X(t)\}_{t = 1}^T\right\} \nonumber \\
		\leq & 2\exp\left\{-\frac{c\varepsilon^2}{\sum_{k=1}^{n-1}\sum_{l = k+1}^n (U_{P_t})_{li}^2(U_{P_t})_{kj}^2 } \right\} \nonumber \\
		\leq & 2\exp\left\{-\frac{c\varepsilon^2}{\sqrt{\sum_{k=1}^{n}\sum_{l = 1}^n (U_{P_t})_{li}^2(U_{P_t})_{kj}^2} } \right\} = 2\exp\{-c\varepsilon^2\}, \label{eq-ua-upupua-f-3-lem2}
	\end{align}
	where the first inequality follows from Theorem~2.6.3 in \cite{vershynin2018high}, and the identity follows from the definitions of $U_P$.  Moreover,
	\begin{equation}\label{eq-ua-upupua-f-4-lem2}
		\left|\sum_{k = 1}^n (U_{P_t})_{ki}(A(t)  - P_t)_{kk}(U_{P_t})_{kj} \right| \leq \sum_{k = 1}^n \left|(U_P)_{ki}(U_P)_{kj}\right| \leq \sqrt{\sum_{k = 1}^n (U_P)_{ki}^2\sum_{k = 1}^n (U_P)_{kj}^2} = 1.
	\end{equation}
	
Combining \eqref{eq-ua-upupua-f-2-lem2}, \eqref{eq-ua-upupua-f-3-lem2} and \eqref{eq-ua-upupua-f-4-lem2}, and taking $\varepsilon$ to be $(C_1/2)\sqrt{\log(n \vee T)}$, we have that
	\begin{align*}
		\mathbb{P}\{\mathcal{E}_1^c \mid \{X(t)\}_{t = 1}^T\} \leq 2Td^2 \exp\left\{-\frac{C_1^2}{32}\log(n \vee T)\right\} \leq (n \vee T)^{-c_1},
	\end{align*}
	where $c_1 > 0$ depends on $C_1$.
	
In addition, it holds that
	\begin{align*}
		\mathbb{P}\{\mathcal{E}_1\}	= \mathbb{E}\left\{\mathbb{P}\{\mathcal{E}_1 \mid \{X(t)\}_{t = 1}^T\}\right\} \geq 1 - (n \vee T)^{-c_1},
	\end{align*}
	therefore, \eqref{eq-E1-def} follows.

\medskip	
We then show that \eqref{eq-E2-def} holds.  For $i \in \{1, \ldots, n\}$ and  $j \in \{1, \ldots, d\}$, we have that
	\[
		\left[\{A(t) - P_t\} U_{P_t}\right]_{ij} = \sum_{l \in \{1, \ldots, n\}\backslash\{i\}} \{(A(t) - P_t\}_{il}(U_{P_t})_{lj} + \{A(t) - P_t\}_{ii}(U_{P_t})_{ij}. 
	\]
	Since
	\[
		\left\vert \{A(t) - P_t\}_{ii}(U_{P_t})_{ij}\right\vert \leq 1
	\]
	and by Hoeffding's inequality that there exists a universal constant $c > 0$, for any $\varepsilon > 0$, 
	\begin{align}
		& \mathbb{P}\left\{\left\vert \sum_{l \in \{1, \ldots, n\} \backslash \{i\}} \{A(t) - P_t\}_{il}(U_{P_t})_{lj} \right\vert > \varepsilon  \Bigg| \{X(t)\}_{t = 1}^T\right\} \nonumber \\
		\leq & 2\exp\left\{-\frac{c\varepsilon^2}{\sum_{l \in \{1, \ldots, n\}\backslash\{i\}}(U_{P_t})_{lj}^2} \right\} \leq 2\exp\{-c\varepsilon^2\},    \label{eq-ua-upupua-f-5-lem2}.
	\end{align}
	we have that
	\begin{align*}
		& \mathbb{P}\left\{\max_{t = 1, \ldots, T} \|\{A(t) -P_t\}U_t\|_{2 \to \infty}^2 > (\varepsilon + \sqrt{d})^2\right\} \\
		= & \mathbb{P}\left\{\max_{t = 1, \ldots, T} \max_{i = 1, \ldots, n} \sum_{j = 1}^d \left[\sum_{l = 1}^n \{A(t) - P_t\}_{il} (U_{P_t})_{lj} \right]^2 > (\varepsilon +\sqrt{d})^2\right\}\\
		\leq & nTd\max_{t = 1, \ldots, T} \max_{i = 1, \ldots n} \max_{j = 1, \ldots, d} \mathbb{P}\left\{\left\vert\sum_{l=1}^n \{A(t) - P_t\}_{il}(U_{P_t})_{lj}\right\vert^2 > \frac{(\varepsilon +  \sqrt{d})^2}{d} \right\} \\
		\leq & nTd\max_{t = 1, \ldots, T} \max_{i = 1, \ldots n} \max_{j = 1, \ldots, d} \mathbb{P}\left\{\left\vert\sum_{l\in \{1, \ldots, n\}\setminus \{i\}}\{A(t) - P_t\}_{il}(U_{P_t})_{lj} \right\vert >   \frac{\varepsilon}{\sqrt{d}}\right\} \\
		\leq & 2nTd\exp\left\{- \frac{c\varepsilon^2}{d} \right),
	\end{align*}
	and \eqref{eq-E2-def} follows by taking $\varepsilon =  C_2/c \sqrt{d\log (n \vee T)}$.
	
Lastly, it follows from Eq.(4.18) in \cite{vershynin2018high} that there exists a universal constant $C_3 > 0$, such that 
	\[
		\mathbb{P}\{\|A(t)  - P_t\|_{\mathrm{op}} > C\sqrt{n} \mid \{X(t)\}_{t = 1}^T\} \leq 4e^{-n},
	\]
	which leads to \eqref{eq-E3-def}.
\end{proof}

\begin{lemma}\label{lem-two-e}

Under \Cref{assume:model-rdpg}, it holds that
	\begin{align*}
		\mathbb{P}\left\{2^{-1} n \min_{k = 1, \ldots, K} \mu_d^k \leq \min_{t = 1, \ldots, T} \lambda_d(P_t) \leq \max_{t = 1, \ldots, T} \lambda_1(P_t) \leq (3/2) n \max_{k = 1, \ldots, K} \mu^k_1\right\}	> 1 - (n \vee T)^{-c_5},
	\end{align*}
\end{lemma}

\begin{proof}
We first fix $t \in \{1, \ldots, T\}$ and for simplicity drop the dependence on $t$ notationally.  For $i \in \{1, \ldots, n\}$, let $Y_i = X_i \Sigma^{-1/2}$ and $Y = (Y_1, \ldots, Y_n)^{\top} = X \Sigma^{-1/2}$, satisfying $\mathbb{E}\{n^{-1}Y^{\top}Y\} = I_d$.

It follows from Lemma~4.1.5 in \cite{vershynin2018high} that for any $\varepsilon > 0$, if 
	\begin{equation}\label{eq-nyy-i}
		\|n^{-1}Y^{\top}Y - I\|_{\mathrm{op}} \leq \max\{\varepsilon, \, \varepsilon^2\},
	\end{equation}
	then the eigenvalues of $n^{-1}Y^{\top}Y$ satisfy
	\[
		(1 - \max\{\varepsilon, \, \varepsilon^2\})^2 \leq \lambda_{\min}(n^{-1}Y^{\top}Y) \leq \lambda_{\max}(n^{-1}Y^{\top}Y) \leq (1 + \max\{\varepsilon, \, \varepsilon^2\})^2,
	\]
	which implies that
	\begin{align*}
		& n(1 - \max\{\varepsilon, \, \varepsilon^2\})^2 \leq \lambda_{\min}(\Sigma^{-1/2}X^{\top}X\Sigma^{-1/2}) \\
		\leq & \lambda_{\max}(\Sigma^{-1/2}X^{\top}X\Sigma^{-1/2}) \leq n(1 + \max\{\varepsilon, \, \varepsilon^2\})^2.
	\end{align*}
	Denote $S = \Sigma^{-1/2}X^{\top}X\Sigma^{-1/2}$.  We then have
	\[
		\lambda_1(P) = \lambda_{\max}(X^{\top}X) = \lambda_{\max}(\Sigma^{1/2}S\Sigma^{1/2}) \leq n(1 + \max\{\varepsilon, \, \varepsilon^2\})^2 \max_{k = 1, \ldots, K} \mu^k_1
	\]
	and
	\begin{align*}
		& \lambda_d(P) = \lambda_{\min}(X^{\top}X) = \lambda_{\min}(\Sigma^{1/2}S\Sigma^{1/2}) = \max_{\mathrm{dim}(E) = d} \min_{v \in \mathcal{S}_E} \langle \Sigma^{1/2}S\Sigma^{1/2}v, v \rangle \\
		= & \max_{\mathrm{dim}(E) = d} \min_{v \in \mathcal{S}_E} \langle S\Sigma^{1/2}v, \Sigma^{1/2} v \rangle = \max_{\mathrm{dim}(E) = d} \min_{v \in \mathcal{S}_E} \|\Sigma^{1/2}v\|^2 \Bigg\langle S\frac{\Sigma^{1/2}v}{\|\Sigma^{1/2}v\|}, \frac{\Sigma^{1/2}v}{\|\Sigma^{1/2}v\|}\Bigg\rangle \\
		\geq & \max_{\mathrm{dim}(E) = d} \min_{v \in \mathcal{S}_E}  \Bigg\langle S\frac{\Sigma^{1/2}v}{\|\Sigma^{1/2}v\|}, \frac{\Sigma^{1/2}v}{\|\Sigma^{1/2}v\|}\Bigg\rangle \min_{k = 1, \ldots, K} \mu_d^k \geq  \max_{\mathrm{dim}(E) = d} \min_{v \in \mathcal{S}_E} \langle Sv, v \rangle \min_{k = 1, \ldots, K} \mu_d^k \\
		\geq & n(1 - \max\{\varepsilon, \, \varepsilon^2\})^2 \min_{k = 1, \ldots, K} \mu_d^k.
	\end{align*}

Now it suffices to investigate \eqref{eq-nyy-i}.  Since
	\[
		\|n^{-1}Y^{\top}Y - I\|_{\mathrm{op}} = \sup_{v \in \mathcal{S}^{d-1}} \left|\frac{1}{n} \sum_{i = 1}^n \left\{(Y_i^{\top}v)^2 - 1\right\}\right|,
	\]
	taking $\mathcal{N}$ to be a $1/4$-net on $\mathcal{S}^{d-1}$, it holds that
	\begin{align*}
		& \mathbb{P}\left\{\|n^{-1}Y^{\top}Y - I\|_{\mathrm{op}} > C \sqrt{\frac{\log(n \vee T)}{n}}\right\} \\
		& \hspace{3cm}\leq  9^d \max_{v \in \mathcal{N}}  \mathbb{P} \left\{\left|\frac{1}{n} \sum_{i = 1}^n \left\{(Y_i^{\top}v)^2 - 1\right\}\right| > C \sqrt{\frac{\log(n \vee T)}{n}} \right\} \\
		\leq & 2 \times 9^d \exp\{-c \log(n \vee T)\},
	\end{align*}
	where $C, c > 0$ are universal constants. 

Thus we have that
	\begin{align*}
		\mathbb{P}\left\{2^{-1} n \min_{k = 1, \ldots, K} \mu_d^k \leq \min_{t = 1, \ldots, T} \lambda_d(P_t) \leq \max_{t = 1, \ldots, T} \lambda_1(P_t) \leq (3/2) n \max_{k = 1, \ldots, K} \mu^k_1\right\}	> 1 - (n \vee T)^{-c_4},
	\end{align*}
	where $c_4 > 0$ is a universal constant.

\end{proof}

Lemma \ref{lem-thm8-carey} is adapted from Theorem~8 in \cite{athreya2017statistical}.

\begin{lemma}\label{lem-thm8-carey}
It holds that
	\begin{align*}
		& \mathbb{P}\Bigg\{\max_{t = 1, \ldots, T}\min_{W \in \mathbb{O}_d}\|\widehat{X}(t) - X(t)W\|_{2  \rightarrow \infty  } >  C_W \frac{\sqrt{d\log(n \vee T)} \vee d^{3/2}}{n^{1/2}}\Bigg\}  \\
		\leq & 1 - (n \vee T)^{-c_1} - (n \vee T)^{-c_2} - 4Te^{-n} - (n \vee T)^{-c_4}.
	\end{align*}
\end{lemma}

\begin{proof}[Proof of Lemma \ref{lem-thm8-carey}]
We first work on a fixed $t \in \{1, \ldots, T\}$, and then use union bounds arguments to reach the final conclusion.  For simplicity, we drop the dependence on $t$ for now. 
Recall that
	\[
		\widehat{X} = U_A S_A^{1/2} \quad \mbox{and} \quad X = U_P S_P^{1/2}.
	\]	
	Define $W^* = W_1W_2^{\top}$, where $W_1$ and $W_2$ are the left and right singular vectors of $U_P^{\top}U_A$, that $U_P^{\top}U_A = W_1 \Lambda_1 W_2^{\top}$.  Since $W^* \in \mathbb{O}_d$, we have that
	\[
		\min_{W \in \mathbb{O}_d}\|\widehat{X} - XW\|_{2 \rightarrow \infty} \leq \|\widehat{X} - XW^*\|_{2 \rightarrow \infty}.
	\]
	In the rest of this proof, denote by $\lambda_1, \ldots, \lambda_n$ as the eigenvalues of $P$, with $|\lambda_1| \geq \cdots |\lambda_n|$; denote by $\widehat{\lambda}_1, \ldots, \widehat{\lambda}_n$ the eigenvalues of $A$, with $|\widehat{\lambda}_1| \geq \cdots \geq |\widehat{\lambda}_n|$.

\vskip 3mm
\noindent \textbf{Step 1.}  We first provide a deterministic upper bound for $\|W^*S_A^{1/2} - S_P^{1/2}W^*\|_{\mathrm{F}}$.

We have,
	\begin{align*}
		W^*S_A & = (W^* - U_P^{\top}U_A)S_A + U_P^{\top}U_AS_A = (W^* - U_P^{\top}U_A)S_A + U_P^{\top}AU_A \\
		& = (W^* - U_P^{\top}U_A)S_A + U_P^{\top}(A - P)U_A + U_P^{\top}PU_A \\
		& = (W^* - U_P^{\top}U_A)S_A + U_P^{\top}(A - P)(U_A - U_PU_P^{\top}U_A) + U_P^{\top}(A - P)U_P + S_PU_P^{\top}U_A  \\
		& = (W^* - U_P^{\top}U_A)S_A + U_P^{\top}(A - P)(U_A - U_PU_P^{\top}U_A) + U_P^{\top}(A - P)U_P \\
		& \hspace{3cm} + S_P(U_P^{\top}U_A - W^*) + S_PW^*,
	\end{align*}
	where the second and the fourth inequalities are due to 
	\[
		AU_A = U_A S_A U_A^{\top}U_A = U_A S_A \quad \mbox{and} \quad U_P^{\top}P = U_P^{\top}U_PS_PU_P^{\top} = S_PU_P^{\top},
	\]
	respectively.  Therefore,
	\begin{align}
		\|W^*S_A - S_PW^*\|_{\mathrm{F}  }  & \leq \|W^* - U_P^{\top}U_A\|_{\mathrm{F}}(\|S_A\|_{\mathrm{op}} + \|S_P\|_{\mathrm{op}}) + \|U_P^{\top}(A - P)(U_A - U_PU_P^{\top}U_A)\|_{\mathrm{F}} \nonumber \\
		& \hspace{2cm} + \|U_P^{\top}(A - P)U_P\|_{\mathrm{F}} \nonumber \\
		& \leq \|I_n - \Lambda_1\|_{\mathrm{F}}\|W_1\|_{\mathrm{op}}\|W_2\|_{\mathrm{op}}(\|S_A\|_{\mathrm{op}} + \|S_P\|_{\mathrm{op}}) \nonumber \\
		& \hspace{2cm} + \|A-P\|_{\mathrm{op}}\|U_A - U_PU_P^{\top}U_A\|_{\mathrm{F}} + \|U_P^{\top}(A - P)U_P\|_{\mathrm{F}} \nonumber \\
		& \leq \|I_n - \Lambda_1\|_{\mathrm{F}}(2\lambda_1 + \|A - P\|_{\mathrm{op}}) + \|A-P\|_{\mathrm{op}}\|U_A - U_PU_P^{\top}U_A\|_{\mathrm{F}} \nonumber \\
		& \hspace{2cm} + \|U_P^{\top}(A - P)U_P\|_{\mathrm{F}}  = (I) + (II) + (III),\label{eq-ws-sw-diff}
	\end{align}
	where $\lambda_1$ is the largest singular value of $P$ and the last inequality is due to Weyl's inequality.
	
In addition, let $\{\theta_1, \ldots, \theta_d\}$ be the principal angles between the column spaces spanned by $U_A$ and $U_P$.  We thus have
	\begin{align}
		\|I_n - \Lambda_1\|_{\mathrm{F}} & = \sqrt{\sum_{i = 1}^d (1 - \cos \theta_i)^2} \leq \sqrt{d} (1 - \cos^2 \theta_1) = \sqrt{d} \sin^2 \theta_1 = \sqrt{d} \min_{W \in \mathbb{O}_d}\|U_A - U_PW\|_{\mathrm{op}}^2 \nonumber \\
		& \leq \sqrt{d} \min_{W \in \mathbb{O}_d}\|U_A - U_PW\|_{\mathrm{F}}^2 \leq \frac{4d^{3/2}\|A - P\|_{\mathrm{op}}^2}{\lambda_d^2}, \label{eq-in-lambda-diff}
	\end{align}
	where the first and second inequalities are due to $\cos \theta_i, \sin \theta_i \in [0, 1]$, and the last inequality is due to Theorem~2 in \cite{yu2014useful} and the fact that $\lambda_{d+1} = 0$.  
	
As for term $(II)$, there exists $W \in \mathbb{O}_d$ such that
	\begin{align}
		& \|U_A - U_PU_P^{\top}U_A\|_{\mathrm{F}} = \sqrt{\mathrm{tr}(U_AU_A^{\top} - U_AU_A^{\top}U_PU_P^{\top})} = \sqrt{d - \mathrm{tr}(U_A^{\top}U_PWW^{\top}U_P^{\top}U_A)} \nonumber \\
		= & \sqrt{\sum_{i = 1}^d (1 - \cos^2\theta_i)} = \sqrt{\sum_{i = 1}^d \sin^2\theta_i} \leq \frac{2 \sqrt{d}\|A - P\|_{\mathrm{op}}}{\lambda_d}. \label{eq-ua-upupua-f}
	\end{align}
	
Term $(III)$ is dealt in Lemma \ref{lem-three-e}.
	
As for $\|W^*S_A^{1/2} - S_P^{1/2}W^*\|_{\mathrm{F}}$, we note that the $ij$-th entry of $W^*S_A^{1/2} - S_P^{1/2}W^*$ satisfies that
	\[
		|W^*_{ij} (\hat{\lambda}_j^{1/2} - \lambda_i^{1/2})| = \left|\frac{W^*_{ij} (\hat{\lambda}_j - \lambda_i)}{\hat{\lambda}^{1/2}_j + \lambda_i^{1/2}}\right| = \left|\frac{(W^*S_A - S_PW^*)_{ij}}{\hat{\lambda}^{1/2}_j + \lambda_i^{1/2}}\right| \leq \frac{|(W^*S_A - S_PW^*)_{ij}|}{\lambda_d^{1/2}},
	\]
	which means
	\begin{align}\label{eq-wsa-spw-diff}
		& \|W^*S_A^{1/2} - S_P^{1/2}W^*\|_{\mathrm{F}} \leq \frac{\|W^*S_A - S_PW^*\|_{\mathrm{F}}}{\lambda_d^{1/2}} \nonumber \\
		\leq & \frac{8d^{3/2} \|A - P\|_{\mathrm{op}}^2 \lambda_1}{\lambda_d^{5/2}} + \frac{4d^{3/2} \|A - P\|_{\mathrm{op}}^3}{\lambda_d^{5/2}} + \frac{2d^{1/2} \|A - P\|^2_{\mathrm{op}}}{\lambda_d^{3/2}} + \frac{\|U_P^{\top} (A - P)U_P\|_{\mathrm{F}}}{\lambda_d^{1/2}}.
	\end{align}
\vskip 3mm
\noindent \textbf{Step 2.}  We then provide an upper bound for $\min_{W \in \mathbb{O}_d}\|\widehat{X} -XW\|_{2 \rightarrow \infty}$.  Since
	\[
		\min_{W \in \mathbb{O}_d}\|\widehat{X} - XW\|_{2 \rightarrow \infty} \leq \|\widehat{X} - XW^*\|_{2 \rightarrow \infty},
	\]
	in the rest of this step, we work on $\|\widehat{X} - XW^*\|_{2\to \infty}$.  We have that 
	\begin{align}
		& \|\widehat{X} - XW^*\|_{2 \to \infty} = \|U_A S_A^{1/2} - U_PS_P^{1/2}W^*\|_{2 \to \infty} \nonumber \\		
		= & \|U_AS_A^{1/2} - U_PW^*S_A^{1/2} + U_P(W^*S_A^{1/2} - S_P^{1/2}W^*)\|_{2 \to \infty} \nonumber \\	
		\leq & \|(U_A - U_PU_P^{\top}U_A)S_A^{1/2}\|_{2 \to \infty} + \|U_P(U_P^{\top}U_A - W^*)S_A^{1/2}\|_{2 \to \infty} \nonumber \\
		& \hspace{2cm} + \|U_P(W^*S_A^{1/2} - S_P^{1/2}W^*)\|_{2 \to \infty} \nonumber  \\ 
		= & (I) + (II) + (III).  \label{eq-x-xw-diff}
	\end{align}

As for term $(I)$, it holds that
	\begin{align*}
		& (U_A - U_PU_P^{\top}U_A)S_A^{1/2} = (A - P)U_AS_A^{-1/2} - U_PU_P^{\top}(A - P)U_AS_A^{-1/2} \\
		= & (A - P)U_PW^* S_A^{-1/2} - U_PU_P^{\top}(A - P)U_P W^*S_A^{-1/2} \\
		& \hspace{2cm} + (I - U_PU_P^{\top})(A-P) (U_A - U_PW^*)S_A^{-1/2},
	\end{align*}
	which satisfies
	\begin{align*}
		& \|(A - P)U_PW^* S_A^{-1/2}\|_{2 \to\infty} \leq \|(A - P)U_P\|_{2 \to \infty} (\widehat{\lambda}_d)^{-1/2}, \\
		& \|U_PU_P^{\top}(A - P)U_p W^*S_A^{-1/2}\|_{2 \to \infty} \leq \|U_P^{\top}(A - P)U_P\|_{\mathrm{F}}(\widehat{\lambda}_d)^{-1/2}
	\end{align*}
	and
	\begin{align*}
		& \|(I - U_PU_P^{\top})(A-P) (U_A - U_PW^*)S_A^{-1/2}\|_{2 \to \infty} \\
		\leq & \|A - P\|_{\mathrm{op}} \|U_A - U_PW^*\|_{\mathrm{F}}(\widehat{\lambda}_d)^{-1/2} \leq \frac{4d^{3/2}\|A - P\|_{\mathrm{op}}^3(\widehat{\lambda}_d)^{-1/2}}{\lambda_d^2},
	\end{align*}
	which is due to \eqref{eq-ua-upupua-f}.  Therefore we have that
	\begin{align}
		& \|(I)\|_{2\to\infty} \leq \|(A - P)U_P\|_{2 \to \infty} (\widehat{\lambda}_d)^{-1/2} + \|U_P^{\top}(A - P)U_P\|_{\mathrm{F}}(\widehat{\lambda}_d)^{-1/2}  \nonumber \\
		& \hspace{3cm} + \frac{4d^{3/2}\|A - P\|_{\mathrm{op}}^3(\widehat{\lambda}_d)^{-1/2}}{\lambda_d^2}. \label{eq-xhat-xw-2inf-1}
	\end{align}
	
As for term $(II)$, it holds that
	\begin{equation}\label{eq-xhat-xw-2inf-2}
		\|U_P(U_P^{\top}U_A - W^*)S_A^{1/2}\|_{2 \rightarrow 
			\infty  } \leq \|I - \Lambda_1\|_{\mathrm{F}} (\lambda_1 + \|A - P\|_{\mathrm{op}})^{1/2} \leq \frac{4d^{3/2}\|A - P\|_{\mathrm{op}}^2}{\lambda_d^2} \sqrt{\frac{3 \lambda_1}{2}}.
	\end{equation}

As for term $(III)$, it holds that 
	\begin{equation}\label{eq-xhat-xw-2inf-3}
		\|U_P(W^*S_A^{1/2} - S_P^{1/2}W^*)\|_{2 \to \infty} \leq \|(W^*S_A^{1/2} - S_P^{1/2}W^*)\|_{\mathrm{F}}.
	\end{equation}

Combining \eqref{eq-wsa-spw-diff}, \eqref{eq-x-xw-diff}, \eqref{eq-xhat-xw-2inf-1}, \eqref{eq-xhat-xw-2inf-2} and \eqref{eq-xhat-xw-2inf-3}, we have that
	\begin{align*}
		\min_{W \in \mathbb{O}_d}\|\widehat{X} - XW\|_{2 \to \infty} & \leq \|(A - P)U_P\|_{2 \to \infty}(\widehat{\lambda}_d)^{-1/2} + \|U_P^{\top}(A - P)U_P\|_{\mathrm{F}}(\widehat{\lambda}_d)^{-1/2} \\
		& \hspace{-2cm} + \frac{4d^{3/2}\|A - P\|_{\mathrm{op}}^3(\widehat{\lambda}_d)^{-1/2}}{\lambda_d^2} + \frac{4d^{3/2}\|A - P\|_{\mathrm{op}}^2}{\lambda_d^2} \sqrt{\frac{3 \lambda_1}{2}} \\
		& \hspace{-2cm} + \frac{8d^{3/2} \|A - P\|_{\mathrm{op}}^2 \lambda_1}{\lambda_d^{5/2}} + \frac{4d^{3/2} \|A - P\|_{\mathrm{op}}^3}{\lambda_d^{5/2}} + \frac{2d^{1/2} \|A - P\|^2_{\mathrm{op}}}{\lambda_d^{3/2}} + \frac{\|U_P^{\top} (A - P)U_P\|_{\mathrm{F}}}{\lambda_d^{1/2}} \\
		& \leq \frac{\|(A - P)U_P\|_{2 \to \infty}}{\sqrt{\lambda_d - \|A - P\|_{\mathrm{op}}}} + \frac{\|U_P^{\top}(A - P)U_P\|_{\mathrm{F}}}{\sqrt{\lambda_d - \|A - P\|_{\mathrm{op}}}}  \\
		& \hspace{-2cm} + \frac{4d^{3/2}\|A - P\|_{\mathrm{op}}^3}{\lambda_d^2 \sqrt{\lambda_d - \|A - P\|_{\mathrm{op}}}} + \frac{4d^{3/2}\|A - P\|_{\mathrm{op}}^2}{\lambda_d^2} \sqrt{\frac{3 \lambda_1}{2}} \\
		& \hspace{-2cm} + \frac{8d^{3/2} \|A - P\|_{\mathrm{op}}^2 \lambda_1}{\lambda_d^{5/2}} + \frac{4d^{3/2} \|A - P\|_{\mathrm{op}}^3}{\lambda_d^{5/2}} + \frac{2d^{1/2} \|A - P\|^2_{\mathrm{op}}}{\lambda_d^{3/2}} + \frac{\|U_P^{\top} (A - P)U_P\|_{\mathrm{F}}}{\lambda_d^{1/2}},
	\end{align*}
	where the second inequality follows from that $\widehat{\lambda}_d \geq \lambda_d - \|A - P\|_{\mathrm{op}}$.  It holds on the event $\mathcal{E}_1 \cap \mathcal{E}_2 \cap \mathcal{E}_3 \cap \mathcal{E}_4$, that
	\begin{align*}
		& \mathbb{P}\Bigg\{\max_{t = 1, \ldots, T}\min_{W \in \mathbb{O}_d}\|\widehat{X}(t) - X(t)W\|_{2 \rightarrow 
			\infty   } >  C_W \frac{\sqrt{d\log(n \vee T)} \vee d^{3/2}}{n^{1/2}}\Bigg\} \\
		& \hspace{4cm} \leq 1 - (n \vee T)^{-c_1} - (n \vee T)^{-c_2} -  (n \vee T)^{-c_4} - 4Te^{-n},
	\end{align*}
	where $C_W > 0$ is a universal constant depending only on $C_1, C_2, C_3, \max_{k = 1, \ldots, K}\mu_1^k$ and $\min_{k = 1, \ldots, K}\mu_d^k$.

\end{proof}

We first state a weakly dependent version of Bernstein inequality.  This is in fact Theorem~4 in \cite{delyon2009exponential}.  The notation in Lemma \ref{thm-delvon} only applies within Lemma \ref{thm-delvon}.

\begin{lemma} \label{thm-delvon}
Let $\{X_1, \ldots, X_T\}$ be centred random variables.  Define
	\[
		g = \sum_{t = 2}^T \sum_{s = 1}^{t-1} \|X_s\|_{\infty} \|\mathbb{E}(X_t \mid \mathcal{F}_{s})\|_{\infty}, \quad v = \sum_{t = 1}^T \|\mathbb{E}(X_t^2 \mid \mathcal{F}_{t-1})\|_{\infty}
	\]
	and
	\[
		m = \max_{t = 1, \ldots, T}\|X_t\|_{\infty},
	\]
	where $F_{s} = \sigma\{X_1, \ldots, X_s\}$, $s \geq 1$, is the natural $\sigma$-field generated by $\{X_i\}_{i=1}^s$.  For any $\varepsilon > 0$, it holds that
	\[
		\mathbb{P}\left\{\left|\sum_{t = 1}^T X_t \right| > \epsilon\right\} \leq 2\exp \left(-\frac{\varepsilon^2}{2(v + 2g) + 2\varepsilon m/3}\right).
	\]
\end{lemma}

\begin{lemma}\label{lem-large-prob-1}
Under \Cref{assume:model-rdpg}, it holds that for any $z \in \mathbb{R}$, 
	\begin{align*}
		\mathbb{P}\left\{\max_{0 \leq s < t < e \leq T} \left|\Delta^t_{s, e}(z)\right| \geq C_8 \sqrt{T} \max\{\sqrt{d\log(n \vee T)}, \, d^{3/2}\}\right\} \leq 4(n \vee T)^{-c} + 4Te^{-n},
	\end{align*}
	where $c = \min\{c_1, c_2, c_4, c_5\} - 1 > 0$ is a universal constant.
	
In addition,
	\begin{align}
		& \mathbb{P}\Bigg\{\max_{0 \leq s < t < e \leq T} \left|\sqrt{\frac{2}{n(e-s)}}\sum_{k = s+1}^e \sum_{(i, j) \in \mathcal{O}} \left(\mathbbm{1}\{\widehat{Y}^k_{ij} \leq z\} - \mathbb{E}\left\{\mathbbm{1}\{Y^k_{ij} \leq z\}\right\}\right)\right| \nonumber \\
		& \hspace{2cm} \geq C_8 \sqrt{\frac{T}{1-\rho}} \max\{\sqrt{d\log(n \vee T)}, \, d^{3/2}\}\Bigg\} \leq 4(n \vee T)^{-c} + 4Te^{-n}, \label{eq-lem5-2}
	\end{align}
	where $c = \min\{c_1, c_2, c_4, c_5\} - 1 > 0$ is a universal constant.

\end{lemma}

\begin{proof}
For any $(i, j) \in \mathcal{O}$ and $t \in \{1, \ldots, T\}$, it holds that
	\begin{align*}
		& \left|\widehat{Y}^t_{ij} - Y^t_{ij}\right| = \left|(\widehat{X}_i(t))^{\top} \widehat{X}_j(t) - (X_i(t))^{\top} X_j(t)\right| = \left|(\widehat{X}_i(t))^{\top} \widehat{X}_j(t) - (W_tX_i(t))^{\top} W_tX_j(t)\right| \\
		\leq & \left|(\widehat{X}_i(t) - W_tX_i(t))^{\top} W_tX_j(t)\right| + \left|(\widehat{X}_i(t) - W_tX_i(t))^{\top}(W_tX_j(t) - \widehat{X}_j(t))\right| \\
		& \hspace{2cm} + \left|(\widehat{X}_j(t) - W_t\widehat{X}_j(t))^{\top} W_tX_i(t)\right| \\
		\leq & 2 \max_{t = 1, \ldots, T} \min_{W \in \mathbb{O}_d}\|\widehat{X}(t) - X(t)W^{\top}\|_{ 2\rightarrow\infty  } \max_{\substack{t = 1, \ldots, T \\ i = 1, \ldots, n}}\|X_i(t)\| \\
		& \hspace{2cm} + \left(\max_{t = 1, \ldots, T} \min_{W \in \mathbb{O}_d} \|\widehat{X}(t) - X(t)W^{\top}\|_{ 2\rightarrow\infty  }\right)^2 \\
		\leq & 2 \max_{t = 1, \ldots, T} \min_{W \in \mathbb{O}_d}\|\widehat{X}(t) - X(t)W^{\top}\|_{ 2\rightarrow\infty  } + \left(\max_{t = 1, \ldots, T} \min_{W \in \mathbb{O}_d} \|\widehat{X}(t) - X(t)W^{\top}\|_{ 2\rightarrow\infty    }   \right)^2,
	\end{align*}
	where $W_t \in \mathbb{O}_d$ satisfies 
	\[
		  \|\widehat{X}(t) - X(t)W_t^{\top}\|_{2 \rightarrow \infty } = \min_{W \in \mathbb{O}_d}\|\widehat{X}(t) - X(t)W^{\top}\|_{2 \rightarrow \infty   }. 
	\]

We fix the chosen pairs $\mathcal{O} \subset \{1, \ldots, n\}^{\otimes 2}$ with $|\mathcal{O}| = n/2$, which is assumed to be an integer.  As for the sequence $\{w_k\}$, it holds that
	\begin{equation}\label{eq-wk}
		\sum_{k = s+1}^e \sum_{(i, j) \in \mathcal{O}}w_k^2 = 1.
	\end{equation}

We have for any $z \in \mathbb{R}$, it holds that
	\begin{align*}
		\left|\Delta^t_{s, e}(z)\right|	& \leq \left|\sum_{k = s + 1}^e w_k \sum_{(i, j) \in \mathcal{O}} \left(\mathbbm{1}\{\widehat{Y}^k_{ij} \leq z\} - \mathbbm{1}\{Y^k_{ij} \leq z\}\right)\right| \\
		& \hspace{2cm} + \left|\sum_{k = s + 1}^e w_k \sum_{(i, j) \in \mathcal{O}} \left(\mathbbm{1}\{Y^k_{ij} \leq z\} - \mathbb{E}\left\{\mathbbm{1}\{Y^k_{ij} \leq z\}\right\}\right)\right|  = (I) + (II).
	\end{align*}

\vskip 3mm
\noindent \textbf{Term $(II)$.}  As for $(II)$, notice that
	\[
		\mathbb{E}\left(\mathbbm{1}\{Y^k_{ij} \leq z\} - \mathbb{E}\left\{\mathbbm{1}\{Y^k_{ij} \leq z\}\right\}\right) = 0.
	\]
	In order to apply Lemma \ref{thm-delvon}, we let 
	\[
		V_i(k) = w_k \mathbbm{1}\{Y^k_{ij} \leq z\} - w_k\mathbb{E}\left\{\mathbbm{1}\{Y^k_{ij} \leq z\}\right\},
	\]
	with $i = 1, \ldots, n/2$, $k = 1, \ldots, T$.  We order $\{V_i(k)\}$ as 
	\begin{equation}\label{eq-order}
		V_1(1), \ldots, V_1(T), V_2(1), \ldots, V_2(T), \ldots, V_{n/2}(1), \ldots, V_{n/2}(T).
	\end{equation}
	Denote $\mathcal{F}_{i, t}$ as the natural $\sigma$-field generated by $V_i(t)$ and all the random variables before it in the order of \eqref{eq-order}, and denote $\mathcal{F}_{i, t, -}$ as the natural $\sigma$-filed generated by all the random variables before $Y_i(t)$ in the order of \eqref{eq-order} excluding $Y_i(t)$.  If $(i, t) = (1, 1)$, then $\mathcal{F}_{i, t, -}$ is the $\sigma$-field generated by constants.
	
In addition, for the notation in Lemma \ref{thm-delvon}, we have that	
	\begin{align}
		v & = \sum_{i = 1}^{n/2} \sum_{t = s+1}^e \left\|\mathbb{E}(V_i(t)^2 \mid \mathcal{F}_{i, t, -})\right\|_{\infty}	\nonumber \\
		& = \sum_{i=1}^{n/2} \sum_{k:\, \eta_k \in (s, e)} \left[ (w_{\eta_k + 1})^2 \mathbb{E}\left\{\mathbbm{1}\{Y^{\eta_k + 1}_{ij} \leq z\}\right\}(1 - \mathbb{E}\left\{\mathbbm{1}\{Y^{\eta_k + 1}_{ij} \leq z\}\right\})\right]\nonumber \\
		& \hspace{1cm} + \sum_{i=1}^{n/2} \sum_{\substack{t \in (s, e] \\ t \notin \{\eta_k + 1\}} } (1-\rho) (w_t)^2 \mathbb{E}\left\{\mathbbm{1}\{Y^{t}_{ij} \leq z\}\right\}(1 - \mathbb{E}\left\{\mathbbm{1}\{Y^{t}_{ij} \leq z\}\right\}) \nonumber \\
		& \hspace{1cm} + \sum_{i=1}^{n/2} \sum_{\substack{t \in (s, e] \\ t \notin \{\eta_k + 1\}} } \rho (w_t)^2 \|(\mathbbm{1}\{Y^{t-1}_{ij} \leq z\} - \mathbb{E}\left\{\mathbbm{1}\{Y^{t-1}_{ij} \leq z\}\right\})^2\|_{\infty} \nonumber \\
		& \leq 1 + \rho, \label{eq-v-bound}
	\end{align}
	where the last inequality is due to \eqref{eq-wk},
	\begin{align}\label{eq-m-bound}
		m \leq \max_{t = 1, \ldots, T}|w_t|,	
	\end{align}
	and
	\begin{align}\label{eq-g-bound}
		g = (n/2)\sum_{k: \eta_k \in (s, e)} \left(\sum_{t = \eta_k + 2}^{\min\{\eta_{k+1}, \, e\}} \sum_{u = \eta_k + 1}^t + \sum_{t = s+1}^{\eta_{k_0+1}} \sum_{u = s+2}^{t-1}\right) |w_t w_u|\rho^{t-u}.	
	\end{align}

Combining \eqref{eq-v-bound}, \eqref{eq-m-bound}, \eqref{eq-g-bound} and Lemma \ref{thm-delvon}, we have for any $\varepsilon > 0$, it holds that 
	\begin{align*}
		\mathbb{P} \left ( (II) \geq \varepsilon  \right) \leq  2\exp \left\{-C\varepsilon^2/((1-\rho)^{-1} + \varepsilon)\right\}.
	\end{align*}	
	We thus denote
	\[
		\mathcal{E}_5 = \left\{\max_{1 < s < t < e \leq T} \left|\sum_{k = s + 1}^e w_k \sum_{(i, j) \in \mathcal{O}} \left(\mathbbm{1}\{Y^k_{ij} \leq z\} - \mathbb{E}\left\{\mathbbm{1}\{Y^k_{ij} \leq z\}\right\}\right)\right| \geq C_5 \sqrt{\frac{\log(n \vee T)}{1-\rho}}\right\},
	\]
	where $C_5 > 0$ is a universal constant, and therefore it holds that
	\[
		\mathbb{P}\{\mathcal{E}_5\} \leq (n \vee T)^{-c_5},
	\]
	where $c_5 > 0$ is a universal constant.

\vskip 3mm
\noindent \textbf{Term ($I$).} As for $(I)$, we have that
	\begin{align*}
		& \mathbb{E}\left\{\left|\mathbbm{1}\{\widehat{Y}^k_{ij} \leq z\} - \mathbbm{1}\{Y^k_{ij} \leq z\}\right|\right\} \\
		\leq & \max\left\{\mathbb{P}\left\{\left(\widehat{Y}^k_{ij} \leq z\right) \cap \left(Y^k_{ij} > z\right)\right\}, \, \mathbb{P}\left\{\left(\widehat{Y}^k_{ij} > z\right) \cap \left(Y^k_{ij} \leq z\right)\right\}\right\} = \max\{(I.1), \, (I.2)\}.
	\end{align*}
	
Let 
	\begin{align*}
		\mathcal{E}_6 = \left\{\max_{t = 1, \ldots, T} \min_{W \in \mathbb{O}_d} \|\widehat{X}_t - X_tW\| \leq C_W \frac{\sqrt{d \log(n \vee T)} \vee d^{3/2}}{n^{1/2}}\right\}.	
	\end{align*}
	On the event $\mathcal{E}_6$, it holds that 
	\[
		\max_{\substack{t = 1, \ldots, T \\ (i, j) \in \mathcal{O}}} \left|\widehat{Y}^t_{ij} - Y^t_{ij}\right| \leq 3C_W \frac{\sqrt{d\log(n \vee T)} \vee d^{3/2}}{n^{1/2}} = \delta
	\]
	and
	\begin{align*}
		& \mathbb{P}\left\{\max_{\substack{t = 1, \ldots, T \\ (i, j) \in \mathcal{O}}} \left|\widehat{Y}^t_{ij} - Y^t_{ij}\right| \leq \delta \right\} \\
		\geq & 1 - (n \vee T)^{-c_1} - (n \vee T)^{-c_2} - (n \vee T)^{-c_4} - 4Te^{-n} = 1 - p_{\delta}.
	\end{align*}
	
Therefore,
	\begin{align*}
		(I.1) & = \mathbb{P}\left\{\left(\widehat{Y}^k_{ij} \leq z\right) \cap \left(Y^k_{ij} > z\right) \big | Y^k_{ij} > z + \delta \right\} \mathbb{P}\{Y^k_{ij} > z + \delta\} \\
		& \hspace{2cm} + \mathbb{P}\left\{\left(\widehat{Y}^k_{ij} \leq z\right) \cap \left(Y^k_{ij} > z\right) \big | Y^k_{ij} < z + \delta \right\} \mathbb{P}\{Y^k_{ij} < z + \delta\} \\
		& \leq p_{\delta}(1 - F_k(z + \delta)) + F_k(z + \delta) - F_k(z) \leq p_{\delta} + \delta C_F
	\end{align*}
	and
	\begin{align*}
		(I.2) & = \mathbb{P}\left\{\left(\widehat{Y}^k_{ij} > z\right) \cap \left(Y^k_{ij} \leq z\right) \big | Y^k_{ij} \leq z - \delta \right\} \mathbb{P}\{Y^k_{ij} \leq z - \delta\} \\
		& \hspace{2cm} + \mathbb{P}\left\{\left(\widehat{Y}^k_{ij} > z\right) \cap \left(Y^k_{ij} \leq z\right) \big | Y^k_{ij} > z - \delta \right\} \mathbb{P}\{Y^k_{ij} > z - \delta\} \\
		& \leq p_{\delta} F_k(z - \delta) + F_k(z) - F_k(z - \delta) \leq p_{\delta} + \delta C_F.
	\end{align*}
	Then we have,
	\begin{align*}
		\mathbb{E}\left|\sum_{k = s + 1}^e w_k \sum_{(i, j) \in \mathcal{O}} \left(\mathbbm{1}\{\widehat{Y}^k_{ij} \leq z\} - \mathbbm{1}\{Y^k_{ij} \leq z\}\right)\right| & \leq 2 \sqrt{\frac{n}{2}}\sqrt{\frac{(e-t)(t-s)}{e-s}} (p_{\delta} + \delta C_F) \\
		& \leq 2 \sqrt{\frac{n}{2}}\min\{\sqrt{e-t}, \, \sqrt{t-s}\} (p_{\delta} + \delta C_F).
	\end{align*}

Therefore, following from similar arguments as those used in bounding $(II)$, we have that for any $\varepsilon > 0$, it holds that
	\begin{align*}
		& \mathbb{P}\Bigg\{\Bigg|\sum_{k = s + 1}^e w_k \sum_{(i, j) \in \mathcal{O}} \left(\mathbbm{1}\{\widehat{Y}^k_{ij} \leq z\} - \mathbbm{1}\{Y^k_{ij} \leq z\}\right) \\
		& \hspace{4cm}- \mathbb{E}\left\{\sum_{k = s + 1}^e w_k \sum_{(i, j) \in \mathcal{O}} \left(\mathbbm{1}\{\widehat{Y}^k_{ij} \leq z\} - \mathbbm{1}\{Y^k_{ij} \leq z\}\right)\right\}\Bigg| > \varepsilon \Bigg\}	\\
		\leq & 2\exp \left\{-C\varepsilon^2/((1-\rho)^{-1} + \varepsilon)\right\},
	\end{align*}
	which implies that
	\begin{align*}
		& \mathbb{P}\Bigg\{\Bigg|\sum_{k = s + 1}^e w_k \sum_{(i, j) \in \mathcal{O}} \left(\mathbbm{1}\{\widehat{Y}^k_{ij} \leq z\} - \mathbbm{1}\{Y^k_{ij} \leq z\}\right)\Bigg| \\
		& \hspace{4cm} > \mathbb{E}\left|\sum_{k = s + 1}^e w_k \sum_{(i, j) \in \mathcal{O}} \left(\mathbbm{1}\{\widehat{Y}^k_{ij} \leq z\} - \mathbbm{1}\{Y^k_{ij} \leq z\}\right)\right| + \varepsilon/2\Bigg\}	\\
		\leq & \mathbb{P}\Bigg\{\Bigg|\sum_{k = s + 1}^e w_k \sum_{(i, j) \in \mathcal{O}} \left(\mathbbm{1}\{\widehat{Y}^k_{ij} \leq z\} - \mathbbm{1}\{Y^k_{ij} \leq z\}\right)\Bigg| \\
		& \hspace{2cm} > 2\sqrt{\frac{n}{2}}\min\{\sqrt{e-t}, \, \sqrt{t-s}\} (p_{\delta} + \delta C_F) + \varepsilon/2\Bigg\} \\
		\leq & 2\exp \left\{-C\varepsilon^2/((1-\rho)^{-1} + \varepsilon)\right\} + p_{\delta}.
	\end{align*}

Lastly, we have that
	\begin{align*}
		& \mathbb{P}\left\{\max_{0 \leq s < t < e \leq T} \left|\Delta^t_{s, e}(z)\right| \geq C_8 \sqrt{\frac{T}{1-\rho}} \max\{\sqrt{d\log(n \vee T)}, \, d^{3/2}\}\right\} \\
		\leq & \mathbb{P}\left\{\left|\Delta^t_{s, e}(z)\right| > C_5\sqrt{\frac{\log(n \vee T)}{1-\rho}} + \sqrt{2n}\min\{\sqrt{e-t}, \, \sqrt{t-s}\} (p_{\delta} + \delta C_F)\right\}  \\
		\leq & 4(n \vee T)^{-c} + 4Te^{-n},
	\end{align*}
	where $c = \min\{c_1, c_2, c_4, c_5\} - 1 > 0$ is a universal constant.
	
The result \eqref{eq-lem5-2} follows from the identical arguments.	
\end{proof}

\begin{lemma}\label{lem-large-prob}
Let
	\[
		\Delta^t_{s, e} = \sup_{z \in \mathbb{R}}|\Delta^t_{s, e}(z)|.
	\]
	It holds that
	\begin{align*}
		\mathbb{P}\Bigg\{\max_{\substack{0 \leq s < t < e \leq T}} \Delta^t_{s, e} > C_9 T^{1/2}(1-\rho)^{-1/2} \max\{\sqrt{d\log(n \vee T)}, \, d^{3/2}\}\Bigg\}  \leq 11(n \vee T)^{-c} + 8Te^{-n}.
	\end{align*}
	In addition,
	\begin{align}
		& \mathbb{P} \Bigg\{\max_{0 \leq s < t < e \leq T} \sup_{z \in \mathbb{R}}\left|\sqrt{\frac{2}{n(e-s)}}\sum_{k = s+1}^e \sum_{(i, j) \in \mathcal{O}} \left(\mathbbm{1}\{\widehat{Y}^k_{ij} \leq z\} - \mathbb{E}\left\{\mathbbm{1}\{Y^k_{ij} \leq z\}\right\}\right)\right| \nonumber \\
		\leq &  C_9 T^{1/2} (1-\rho)^{-1/2}\max\{\sqrt{d\log(n \vee T)}, \, d^{3/2}\}\Bigg\}  \leq 11(n \vee T)^{-c} + 8Te^{-n}. \label{eq-lem6-2}
	\end{align}	
\end{lemma}

\begin{proof}
Let 
	\begin{equation}\label{eq-lem7-pf-1}
		\delta = 3C_W \frac{\sqrt{d\log(n \vee T)} \vee d^{3/2}}{n^{1/2}}.
	\end{equation}
	Let $z_m = m\delta$, $m = 1, \ldots, \lfloor 1/\delta \rfloor$.  Let $I_m = [z_m - \delta, z_m + \delta]$, for $m = 1, \ldots, \lfloor 1/\delta \rfloor - 1$, and $I_{\lfloor 1/\delta \rfloor} = [z_{\lfloor 1/\delta \rfloor - 1}, 1]$.  Let $M = \lfloor 1/\delta \rfloor$.  Then 
	\begin{equation}\label{eq-lem7-pf-2}
		\sup_{z \in \mathbb{R}} |\Delta^t_{s, e}(z)| \leq \max_{j = 1, \ldots, M}\left\{|\Delta^t_{s, e}(z_j)| + \sup_{z \in I_j} |\Delta^t_{s, e}(z_j) - \Delta^t_{s, e}(z)|\right\}.
	\end{equation}

It follows from Lemma \ref{lem-large-prob-1} that 
	\begin{equation}\label{eq-lem7-pf-3}
		\mathbb{P}\left\{\max_{j = 1, \ldots, M}|\Delta^t_{s, e}(z_j)| \geq C_8 \sqrt{T} (1-\rho)^{-1/2}\max\{\sqrt{d\log(n \vee T)}, \, d^{3/2}\}\right\} \leq 4(n \vee T)^{-c} + 4Te^{-n}.
	\end{equation}	
	
For every $z \in \mathbb{R}$, on the event 
	\[
		\left\{\max_{\substack{k = 1, \ldots, T \\ (i, j) \in \mathcal{O}}} \left|\widehat{Y}^k_{ij} - Y^k_{ij}\right| \leq \delta \right\},
	\]	
	it holds that 
	\[
		\left|\mathbbm{1}\{\widehat{Y}^k_{ij} \leq z\} - \mathbbm{1}\{Y^k_{ij} \leq z\}\right| \leq \mathbbm{1}\{Y^k_{ij} \in [z - \delta, z + \delta]\}.
	\]
	For any $z \in \mathbb{R}$, there exist $z_m$ and $z_{m+1}$, $m \in \{1, \ldots, M=1\}$, such that 
	\[
		[z - \delta, z + \delta] \subset [z_m - \delta, z_m + \delta] \cup [z_{m+1} - \delta, z_{m+1} + \delta].
	\]
	Let 
	\[
		B_m = \sum_{k = s + 1}^e \sum_{(i, j) \in \mathcal{O}} \mathbbm{1}\{Y^k_{i, j} \in I_m\}, \, m = 1, \ldots, M.
	\]
	Therefore
	\begin{align}
		& \left|\Delta_{s,e}^t(z_m) - \Delta_{s,e}^t(z)\right|\nonumber \\
		\leq & \left|\sum_{k = s + 1}^e \sum_{(i, j) \in \mathcal{O}} w_k \bigl\{\mathbbm{1}\{\widehat{Y}^k_{i, j} \leq z_m\} - \mathbbm{1}\{Y^k_{i, j} \leq z_m\} \bigr\} \right| \nonumber \\
		& \hspace{0.5cm} + \left|\sum_{k = s + 1}^e \sum_{(i, j) \in \mathcal{O}} w_k \bigl\{\mathbbm{1}\{\widehat{Y}^k_{i, j} \leq z\} - \mathbbm{1}\{Y^k_{i, j} \leq z\} \bigr\} \right| \nonumber \\
		& \hspace{0.5cm} + \left|\sum_{k = s + 1}^e \sum_{(i, j) \in \mathcal{O}} w_k \bigl\{\mathbbm{1}\{Y^k_{i, j} \leq z_m\} - \mathbbm{1}\{Y^k_{i, j} \leq z\}\bigr\} \right| + \left|\sum_{k = s + 1}^e \sum_{(i, j) \in \mathcal{O}} w_k \bigl\{G_k(z_m) - G_k(z)\bigr\} \right| \nonumber  \\
		\leq & \left(\sqrt{\frac{2(e-t)}{n(e-s)(t-s)}} \vee \sqrt{\frac{2(t-s)}{n(e-s)(e-t)}}\right)\Bigg(\left|\sum_{k = s + 1}^e \sum_{(i, j) \in \mathcal{O}} \bigl|\mathbbm{1}\{\widehat{Y}^k_{i, j} \leq z_m\} - \mathbbm{1}\{Y^k_{i, j} \leq z_m\} \bigr| \right| \nonumber  \\
		& \hspace{0.5cm} + \left|\sum_{k = s + 1}^e \sum_{(i, j) \in \mathcal{O}} \bigl|\mathbbm{1}\{\widehat{Y}^k_{i, j} \leq z\} - \mathbbm{1}\{Y^k_{i, j} \leq z\} \bigr| \right| + \left|\sum_{k = s + 1}^e \sum_{(i, j) \in \mathcal{O}}\mathbbm{1}\{Y^k_{i, j} \leq I_m\}\right|\Bigg) \nonumber  \\
		& \hspace{0.5cm} + \left|\sum_{k = s + 1}^e \sum_{(i, j) \in \mathcal{O}} w_k \bigl\{G_k(z_m) - G_k(z)\bigr\} \right| \nonumber  \\
		\leq & \left(\sqrt{\frac{2(e-t)}{n(e-s)(t-s)}} \vee \sqrt{\frac{2(t-s)}{n(e-s)(e-t)}}\right)\Bigg( \left|\sum_{k = s + 1}^e \sum_{(i, j) \in \mathcal{O}} \mathbbm{1}\{Y^k_{i, j} \in I_m\} \right| \nonumber  \\
		& \hspace{0.5cm} + \sup_{m = 1, \ldots, M-1}\left|\sum_{k = s + 1}^e \sum_{(i, j) \in \mathcal{O}} \mathbbm{1}\{Y^k_{i, j} \in [z_m - \delta, z_{m+1} + \delta]\} \right|  + \left|\sum_{k = s + 1}^e \sum_{(i, j) \in \mathcal{O}}\mathbbm{1}\{Y^k_{i, j} \in I_m\}\right|\Bigg)  \nonumber  \\
		& \hspace{0.5cm} + \left(\sum_{k=s+1}^e \sum_{(i, j) \in \mathcal{O}} |w_k|\right) \max_{k = s+1, \ldots, e} |G_k(z) - G_k(z_m)| \nonumber  \\
		& \leq 4 \left(\sqrt{\frac{2(e-t)}{n(e-s)(t-s)}} \vee \sqrt{\frac{2(t-s)}{n(e-s)(e-t)}}\right) \max_{m = 1, \ldots, M}B_m + \sqrt{\frac{2n(e-t)(t-s)}{e-s}} \delta C_G. \label{eq-lem7-pf-4}
	\end{align}
	
Since
	\begin{align*}
		& \max_{m = 1, \ldots, M}B_m \\
		\leq &\left|\sum_{k = s + 1}^e \sum_{(i, j) \in \mathcal{O}}(\mathbbm{1}\{Y^k_{i, j} \in I_m\} - \mathbb{P}\{\mathbbm{1}\{Y^k_{i, j} \in I_m\}\}) \right| + \left|\sum_{k = s + 1}^e \sum_{(i, j) \in \mathcal{O}}\mathbb{P}\{\mathbbm{1}\{Y^k_{i, j} \in I_m\}\} \right|  \\
		\leq &  \left|\sum_{k = s + 1}^e \sum_{(i, j) \in \mathcal{O}}(\mathbbm{1}\{Y^k_{i, j} \in I_m\} - \mathbb{P}\{\mathbbm{1}\{Y^k_{i, j} \in I_m\}\}) \right| + (e-s)n\delta C_G,
	\end{align*}
	and
	\begin{align*}
		& \mathbb{P}\left\{\max_{m = 1, \ldots, M}\left|\sum_{k = s + 1}^e \sum_{(i, j) \in \mathcal{O}}(\mathbbm{1}\{Y^k_{i, j} \in I_m\} - \mathbb{P}\{\mathbbm{1}\{Y^k_{i, j} \in I_m\}\}) \right| \leq C_9 \sqrt{\frac{n(e-s) \log (n \vee T)}{1-\rho}} \right\} \\
		\geq & 1 - (n \vee T)^{-c_9},
	\end{align*}
	where $C_9, c_9 > 0$ are universal constants, we have that 
	\begin{align}
		& \mathbb{P}\Bigg\{\left(\sqrt{\frac{2(e-t)}{n(e-s)(t-s)}} \vee \sqrt{\frac{2(t-s)}{n(e-s)(e-t)}}\right) \max_{m = 1, \ldots, M}B_m \nonumber \\
		& \hspace{2cm} \geq C_{10}T^{1/2}(1-\rho)^{-1/2} (\sqrt{d\log(n \vee T)} \vee d^{3/2})\Bigg\}	 \leq (n \vee T)^{-c_{10}}, \label{eq-lem7-pf-5}
	\end{align}
	where $C_{10}, c_{10} > 0$ are universal constants.
	
Combining \eqref{eq-lem7-pf-1}, \eqref{eq-lem7-pf-2}, \eqref{eq-lem7-pf-3}, \eqref{eq-lem7-pf-4} and \eqref{eq-lem7-pf-5}, the proof is complete.					
\end{proof}

\section{Change point analysis lemmas}\label{sec-app-cpd}

\begin{lemma}\label{lem-lem2}
Under \Cref{assume:model-rdpg}, for any pair $(s, e) \subset (0, T)$ satisfying
	\[
	\eta_{k-1} \leq s \leq \eta_k \leq \ldots \leq \eta_{k+q} \leq e \leq \eta_{k+q+1}, \quad q \geq 0,
	\]	
	let
	\[
	   b_1 \in \underset{b = s +1, \ldots, e-1}{\arg \max} \widetilde{D}_{s,e}^b.
	\]
	Then $ b_1 \in \{\eta_1,\ldots,\eta_K \} $.
	
Let $z \in \argmax_{x \in \mathbb{R}} |\widetilde{D}^b_{s, e}(x)|$.  If $\widetilde{D}^{t}_{s, e}(z) > 0$ for some $t \in (s, e)$, then $\widetilde{D}_{s, e}^t(z)$ is either monotonic or decreases and then increases within each of the interval $(s, \eta_k)$, $(\eta_{k}, \eta_{k+1})$, $\ldots$, $(\eta_{k+q}, e)$.
\end{lemma}

This is identical to Lemma~7 in \cite{padilla2019optimal} and we omit the proof here.

\begin{lemma} \label{lem-uni-lem10}
Under \Cref{assume:model-rdpg}, let $0 \leq s < \eta_k < e \leq T$ be any interval satisfying 
	\[
	   \min\{\eta_k - s, \, e - \eta_k\} \geq c_1 \Delta,
	\]
	with $c_1 > 0$.  Then we have that 
	\[
		\max_{t = s + 1, \ldots, e - 1} \widetilde{D}_{s, e}^{t} \geq \frac{2^{-3/2}c_1 \kappa \Delta \sqrt{n}}{\sqrt{e-s}}.
	\]
\end{lemma}

\begin{proof}
Recall that 
	\[
		G_{\eta_k}(z) = \mathbb{P}\left\{(X_1(\eta_k))^{\top}X_2(\eta_k) \leq z\right\}.
	\]
	Let
	\[
		z_0 \in \argmax_{z\in [0, 1]} |G_{\eta_{k}}(z)  -  G_{\eta_{k+1}}(z)|.
	\]
	Without loss of generality, assume that $F_{\eta_k}(z_0) > F_{\eta_{k+1}}(z_0)$.  For $s < t < e$, note that 
	\begin{align*}
		\widetilde{D}_{s, e}^{t}(z_0) & = \left|\sqrt{\frac{n(e-t)}{2(e-s)(t-s)}} \sum_{k = s+1}^t G_k(z_0) - \sqrt{\frac{n(t-s)}{2(e-s)(e-t)}} \sum_{k = t+1}^e G_k(z_0)\right| \\
		& = \left|\sqrt{\frac{n(e-s)}{2(t-s)(e-t)}} \sum_{k = s+1}^t \widetilde{G}_k(z_0)\right|,
	\end{align*}
	where $\widetilde{G}_k(z_0) = G_k(z_0) - (e-s)^{-1}\sum_{k = s + 1}^e G_k(z_0)$.
	
Under \Cref{assume:model-rdpg} , it holds that $\widetilde{G}_{\eta_k}(z_0) > \kappa/2$.  Therefore
	\[
		\sum_{k = s+1}^{\eta_k} \widetilde{G}_k(z_0) \geq (c_1/2) \kappa \Delta, \quad \mbox{and} \quad \sqrt{\frac{n(e-s)}{2(t-s)(e-t)}} \geq \sqrt{\frac{n}{2(e-s)}}.
	\]
	Then
	\[
		\max_{t = s + 1, \ldots, e - 1} \widetilde{D}_{s, e}^{t} \geq \frac{2^{-3/2}c_1 \kappa \Delta \sqrt{n}}{\sqrt{e-s}}.
	\]
\end{proof}

\begin{lemma}\label{lem-uni-lem8}
Under \Cref{assume:model-rdpg},if $\eta_k$ is the only change point in $(s, e)$, then
	\begin{equation}\label{eq-lem4-2}
		\widetilde{D}_{s, e}^{\eta_k} \leq \kappa_k \sqrt{n/2} \min\{\sqrt{\eta_k-s}, \, \sqrt{e-\eta_k}\};
	\end{equation}	
	if $(s, e) \subset (0, T)$ contain two and only two change points $\eta_k$ and $\eta_{k+1}$, then we have
	\begin{equation}\label{eq-2cpt}
		\max_{t = s+1, \ldots, e-1} \widetilde{D}_{s, e}^{\eta_k} \leq \sqrt{n/2}\sqrt{e - \eta_{k+1}} \kappa_{k+1} + \sqrt{n/2} \sqrt{\eta_k - s} \kappa_k;
	\end{equation}
	if $(s, e) \subset (0, T)$ contains two or more change points, including $\eta_k$ and $\eta_{k+1}$, which satisfy that $\eta_k - s \leq c_1 \Delta$, for $c_1 > 0$, then
	\begin{equation}\label{eq-more-2cpt}
		\widetilde{D}^{\eta_k}_{s, e} \leq \sqrt{c_1} \widetilde{D}^{\eta_{k+1}}	_{s, e} + \sqrt{2(\eta_k - s) n} \kappa_k.
	\end{equation}

\end{lemma}

\begin{proof}
As for \eqref{eq-lem4-2}, it is due to that
	\begin{align*}
		\widetilde{D}_{s, e}^{\eta_k}  = \sqrt{\frac{n(\eta_k - s)(e-\eta_k)}{2(e-s)}} \sup_{z \in \mathbb{R}} \bigl|G_{\eta_k}(z) - G_{\eta_k+1}(z)\bigr|  \leq \kappa_k \sqrt{n/2} \min\{\sqrt{\eta_k-s}, \, \sqrt{e-\eta_k}\}.
	\end{align*}
	Eq.~\eqref{eq-2cpt} follows similarly.
	
As for \eqref{eq-more-2cpt}, we consider the distribution sequence $\{H_t\}_{t = s + 1}^e$ be such that
	\[
		H_t = \begin{cases}
			G_{\eta_k + 1}, & t = s+1, \ldots, \eta_k,\\
			G_t, & t = \eta_k + 1, \ldots, e.
	 	\end{cases}	
	\]	
	For any $s < t < e$, define
	\[
		\mathcal{H}_{s,e}^t = \underset{ z \in \mathbb{R}  }{\sup }\left\vert  \mathcal{H}_{s,e}^t(z)  \right\vert,  
	\]
	where
	\[	
		\mathcal{H}_{s,e}^t(z) = \sqrt{\frac{n (t - s)(e - t)}{2(e-s)} } \left\{\frac{1}{t - s}\sum_{l = s + 1}^t H_l(z) -\frac{1}{e - t}\sum_{l = t+1}^e H_l(z)\right\}.
	\]

For any $t \geq \eta_k$ and $z \in \mathbb{R}$, it holds that
	\[
		\left|\widetilde{D}^t_{s, e}(z) - \mathcal{H}^t_{s, e}(z)\right| = \sqrt{\frac{2(e-t)}{n(e-s)(t-s)}} \frac{n(\eta_k - s)}{2} \left|G_{\eta_{k+1}}(z) - G_{\eta_{k}}(z)\right| \leq \sqrt{\frac{n(\eta_k - s)}{2}} \kappa_k.
	\]
	Thus we have
	\begin{align*}
		\widetilde{D}^{\eta_k}_{s, e} & = \sup_{z \in \mathbb{R}} \left|\widetilde{D}^{\eta_k}_{s, e}(z) - \mathcal{H}^{\eta_k}_{s, e}(z) + \mathcal{H}^{\eta_k}_{s, e}(z)\right| \leq \sup_{z \in \mathbb{R}} \left| \widetilde{D}^{\eta_k}_{s, e}(z) - \mathcal{H}^{\eta_k}_{s, e}(z) \right| + \mathcal{H}^{\eta_k}_{s, e} \\
		& \leq \mathcal{H}^{\eta_k}_{s, e} + \sqrt{\frac{n(\eta_k - s)}{2}} \kappa_k \leq \sqrt{\frac{(\eta_k - s)(e - \eta_{k+1})}{(\eta_{k+1} - s)(e - \eta_{k+1})}} \mathcal{H}^{\eta_{k+1}}_{s, e} + \sqrt{\frac{n(\eta_k - s)}{2}} \kappa_k   \\
	& \leq \sqrt{c_1} \widetilde{D}^{\eta_{k+1}}_{s, e} + \sqrt{2n(\eta_k-s)} \kappa_k.
	\end{align*}	
\end{proof}

\begin{lemma}
\label{lem:local}
For any $z_0 \in \mathbb{R}$ and $(s, e) \subset (0, T)$ satisfying the following: there exits a true change point $\eta_k \in (s, e)$ such that 
	\begin{equation} \label{eqn:interval_lb}
		\min\{\eta_k - s, \, e - \eta_k\} \geq c_1 \Delta,
	\end{equation}
	\begin{equation} \label{eqn:lower_bound}
		\widetilde{D}_{s, e}^{\eta_{k}} (z_0) \geq (c_1/2)  \sqrt{n/2} \frac{\kappa \Delta }{ \sqrt{e - s}},
	\end{equation}
	where $c_1 > 0$ is a sufficiently small constant, and that	
	\begin{equation} \label{eqn:upper_bound} 
		\max_{t = s+ 1, \ldots, e}  \vert \widetilde{D}_{s, e}^{t} (z_0)\vert -  \widetilde{D}_{s, e}^{\eta_k} (z_0) \leq 2^{-3/2}c_1^3 (e-s)^{-7/2}\Delta^4\kappa\sqrt{n}, 
	\end{equation}
	for all $d \in (s, e)$ satisfying
	\begin{equation} \label{eqn:distance}
		\vert d -  \eta_k\vert \leq  c_1\Delta/32,
	\end{equation} 
	it holds that
	\[
		\widetilde{D}_{s, e}^{\eta_k} (z_0) -  \widetilde{D}_{s, e}^{d}(z_0) > c  |d - \eta_k| \Delta \widetilde{D}^{\eta_k}_{s, e}(z_0)(e-s)^{-2},
	\]
	where $c > 0$ is a sufficiently small constant.
\end{lemma}

\begin{proof}
	The proof is identical to the proof of Lemma~11 in \cite{padilla2019optimal} after letting $n_{\min} = n_{\max} = n/2$.	
\end{proof}

\begin{lemma}\label{lem-uni-15}

Under \Cref{assume:model-rdpg}, consider any generic $(s, e) \subset (0, T)$, satisfying
	\[
		\min_{l = 1, \ldots, K}\min\{\eta_l - s, e - \eta_l\} \geq \Delta/16, \quad \eta_k \in (s, e).
	\]
	and
	\[
		e - s \leq C_R \Delta.
	\]
	Let
	\[
		\kappa_{s, e}^{\max} = \max_{\substack{l = 1, \ldots, K \\ \eta_l \in (s, e)}} \kappa_l,
	\]
	and $b \in \argmax_{s < t < e} D^t_{s, e}$.  For some $c_1 > 0$ and $\gamma > 0$, suppose that
	\begin{equation}\label{eq-lem13-2}
		D^b_{s, e} \geq c_1 \kappa^{\max}_{s, e}\sqrt{\Delta n},
	\end{equation}
	\begin{equation}\label{eq-lem13-1}
		\max_{t = s + 1, \ldots, e - 1}\sup_{z \in \mathbb{R}}\left|\Delta^t_{s, e}(z)\right| \leq \gamma,
	\end{equation}
	and
	\begin{equation}\label{eq-lem13-4}
		\max_{0 \leq s < e \leq T}\sup_{z \in \mathbb{R}}\left| \sqrt{\frac{2}{n(e-s)}} \sum_{t = s+1}^e \sum_{i, j \in \mathcal{O}} \left(\mathbbm{1}\{\widehat{Y}^t{i,j} \leq z\} - G_t(z)\right)\right| \leq \gamma.
	\end{equation}
	If there exits a sufficiently small $0 < c_2 < c_1/2$ such that
	\begin{equation}\label{eq-lem13-3}
		\gamma \leq c_2 \kappa^{\max}_{s, e}\sqrt{\Delta n},
	\end{equation}
	then there exists a change point $\eta_k \in (s, e)$ such that
	\[
		\min\{e - \eta_k, \, \eta_k - s\} \geq \Delta/4 \quad \mbox{and} \quad |\eta_k - b| \leq C_{\epsilon}\frac{\gamma^2}{\kappa^2_k n},
	\]		
	where $C_{\epsilon} > 0$ is a sufficiently large constant.
\end{lemma}

\begin{proof}

Without loss of generality, assume that $\widetilde{D}^b_{s, e} > 0$ and that $\widetilde{D}^t_{s, e}$ is locally decreasing at $b$.  Observe that there has to be a change point $\eta_k \in (s, b)$, or otherwise $\widetilde{D}^b_{s, e} > 0$ implies that $\widetilde{D}^t_{s, e}$ is decreasing, as a consequence of Lemma~\ref{lem-lem2}.  Thus, if $s \leq \eta_k \leq b \leq e$, then
	\begin{equation}\label{eq-lem13-pf-1}
		\widetilde{D}^{\eta_k}_{s, e} \geq \widetilde{D}^b_{s, e} \geq D^b_{s, e} - \gamma \geq (c_1 - c_2)\kappa^{\max}_{s, e}\sqrt{\Delta n/2} \geq 2^{-3/2}c_1\kappa^{\max}_{s, e}\sqrt{\Delta n},
	\end{equation}
	where the second inequality follows from \eqref{eq-lem13-1}, and the third inequality follows from \eqref{eq-lem13-2} and \eqref{eq-lem13-3}.  Observe that $e - s \leq C_R \Delta$ and that $(s, e)$ contains at least one change point. \\
	
\noindent \textbf{Step 1.}  In this step, we are to show that 
	\begin{equation}\label{eq-lem13-pf-2}
		\min\{\eta_k - s, \, e - \eta_k\} \geq \min\{1, c_1^2\}\Delta/16.
	\end{equation}	
	Suppose that $\eta_k$ is the only change point in $(s, e)$.  Then \eqref{eq-lem13-pf-2} must hold or otherwise it follows from \eqref{eq-lem4-2} in Lemma \ref{lem-uni-lem8}, we have
	\[
		D^{\eta_k}_{s, e} \leq \kappa_k \sqrt{\Delta n}\frac{c_1}{4},
	\]
	which contradicts \eqref{eq-lem13-pf-1}.
	
Suppose $(s, e)$ contains at least two change points.  Then $\eta_k - s < \min\{1, \, c_1^2\}\Delta /16$ implies that $\eta_k$ is the most left change point in $(s, e)$.  Therefore it follows from \eqref{eq-more-2cpt} that
	\begin{align}
		\widetilde{D}^{\eta_k}_{s, e} & \leq \frac{c_1}{4} \widetilde{D}^{\eta_{k+1}}_{s, e} + \sqrt{2n(\eta_k - s)} \kappa_k \leq \frac{c_1}{4} \max_{t = s+ 1, \ldots, e} \widetilde{D}^t_{s, e} + \frac{c_1 \kappa_k \sqrt{n \Delta}}{4\sqrt{2}} \nonumber \\
		& \leq \frac{c_1}{4} \max_{t = s+1, \ldots, e} D^t_{s, e} + \frac{c_1}{4} \gamma + \frac{c_1 \kappa_k \sqrt{n \Delta}}{4\sqrt{2}} \nonumber \\
		& < \max_{t = s + 1, \ldots, e} D^t_{s, e} - \gamma,  \label{eq-dtilde-d-gamma-contra}
	\end{align}
	where the last inequality follows from that
	\[
		\max_{t = s + 1, \ldots, e} D^t_{s, e} = D^b_{s, e} \geq 2^{-3/2}c_1\kappa^{\max}_{s, e}\sqrt{\Delta n},
	\]
	as implied by \eqref{eq-lem13-pf-1}.  Therefore, \eqref{eq-dtilde-d-gamma-contra} contradicts 
	\[
		\widetilde{D}^{\eta_k}_{s, e} \geq \widetilde{D}^b_{s, e} - \gamma,
	\]
	which is also implied by \eqref{eq-lem13-pf-1}.\\

\noindent \textbf{Step 2.}  It follows from Lemma~\ref{lem:local} that 
	\begin{equation}\label{eq-lem13-pf-3}
		\widetilde{D}^{\eta_k}_{s, e} - \widetilde{D}^{\eta_k + c_1\Delta/32}_{s, e} \geq c\frac{c_1\Delta}{32} \Delta \widetilde{D}^{\eta_k}_{s, e}(e-s)^2 \geq \frac{cc_1}{32C_R^2}(c_1 \kappa \sqrt{\Delta n} - 2\gamma) \geq 2\gamma.
	\end{equation}
	We claim that $b \in (\eta_k, \eta_k + c_1\Delta/32)$.  By contradiction, suppose that $b \geq \eta_k + c_1\Delta/32$.  Then
	\begin{equation}\label{eq-lem13-pf-4}
		\widetilde{D}^b_{s, e} \leq \widetilde{D}^{\eta_k + c_1\Delta/32}_{s, e} < \widetilde{D}^{\eta_k}_{s, e} - 2\gamma \leq \max_{t = s+1, \ldots, e}\widetilde{D}^t_{s, e} - 2\gamma \leq \max_{t = s+1, \ldots, e} D^t_{s, e} - \gamma = D^b_{s, e} - \gamma,
	\end{equation}
	where the first inequality follows from Lemma~\ref{lem-lem2}, the second follows from \eqref{eq-lem13-pf-3}, and the fourth follows from \eqref{eq-lem13-1}.   Note that \eqref{eq-lem13-pf-4} shows that
	\[
		\widetilde{D}^b_{s, e} < D^b_{s, e} - \gamma,
	\]
	which is a contradiction with \eqref{eq-lem13-pf-1} showing that 
	\[
		\widetilde{D}^b_{s, e} \geq \widetilde{D}^b_{s, e} - \gamma.
	\]
	Therefore we have $b \in (\eta_k, \eta_k + c_1\Delta /32)$. \\

\noindent \textbf{Step 3.}  This follows from the identical arguments as those in \textbf{Step 3} in the proof of Lemma~15 in \cite{padilla2019optimal} by letting $n_{\min} = n_{\max} = n/2$ and translating notation appropriately.  We have that
	\[
		|b - \eta_k| \leq C_{\epsilon} \frac{\gamma^2}{n \kappa_k^2},
	\]
	where $C_{\epsilon} > 0$ is a universal constant.

\end{proof}

\section{Proof of \Cref{thm-main}}\label{sec-app-main}

\begin{proof}[Proof of \Cref{thm-main}]
	Since $\epsilon$ is the upper bound of the localisation error, by induction, it suffices to consider any interval $(s, e) \subset (1, T)$	 that satisfies
	\[
	\eta_{k-1} \leq s \leq \eta_k \leq \ldots \leq \eta_{k+q} \leq e \leq \eta_{k+q+1}, \quad q \geq -1,
	\]
	and
	\[
	\max\bigl\{\min\{\eta_k - s, \, s - \eta_{k-1}\}, \, \min \{\eta_{k+q+1} - e, \, e - \eta_{k+q}\} \bigr\} \leq \epsilon,
	\]
	where $q = -1$ indicates that there is no change point contained in $(s, e)$.
	
	By \Cref{assume-snr}, it holds that 
	\[
	\epsilon = C_{\epsilon} \frac{T \max\{d\log(n \vee T), \, d^3\}}{\kappa^2 n (1 - \rho)} < \Delta/4.
	\]
	It has to be the case that for any change point $\eta_k \in (0, T)$, either $|\eta_k - s| \leq \epsilon$ or $|\eta_k - s| \geq \Delta - \epsilon \geq 3\Delta/4$.  This means that $\min\{|\eta_k - s|, \, |\eta_k - e|\}\leq \epsilon$ indicates that $\eta_k$ is a detected change point in the previous induction step, even if $\eta_k \in (s, e)$.  We refer to $\eta_k \in (s, e)$ an undetected change point if $\min\{|\eta_k - s|, \, |\eta_k - e|\} \geq 3\Delta/4$.
	
	In order to complete the induction step, it suffices to show that we (i) will not detect any new change point in $(s, e)$ if all the change points in that interval have been previous detected, and (ii) will find a point $b \in (s, e)$ such that $|\eta_k - b| \leq \epsilon$ if there exists at least one undetected change point in $(s, e)$.
	
	Recall the definitions $Y_{ij}^k = (X_i(k))^{\top}X_j(k)$ and $\widehat{Y}_{ij}^k = (\widehat{X}_i(k))^{\top}\widehat{X}_j(k)$.  For $j = 1, 2$, define the events
	\[
	\mathcal{B}_j(\gamma) = \left\{\max_{1\leq s < b < e \leq T}\sup_{z \in [0, 1]}\left|\sum_{k = s+1}^e w_k^{(j)} \sum_{(i, j) \in \mathcal{O}} \left(\mathbbm{1}\{\widehat{Y}^k_{ij} \leq z\} - \mathbb{E}\left\{\mathbbm{1}\{Y^k_{ij} \leq z\}\right\}\right)\right| \leq \gamma\right\},
	\]
	where
	\[
	w_k^{(1)} = \begin{cases}
	\sqrt{\frac{2}{n}}\sqrt{\frac{(e-b)}{(b-s)(e-s)}}, & k = s + 1, \ldots, b, \\
	- \sqrt{\frac{2}{n}} \sqrt{\frac{(b-s)}{(e-b) (e-s)}}, & k = b+1, \ldots, e,
	\end{cases}, \quad w_k^{(2)} = \sqrt{\frac{2}{n}}\frac{1}{\sqrt{e-s}},
	\]
	and
	\[
	\gamma = C_{\gamma} T^{1/2} \max\{\sqrt{d\log(n \vee T)}, \, d^{3/2}\},
	\] 
	with a sufficiently large constant $C_{\gamma} > 0$.  
	
	Define
	\[
	\mathcal{S} = \bigcap_{k = 1}^K\left\{\alpha_s \in [\eta_k - 3\Delta/4, \eta_k - \Delta/2], \, \beta_s \in [\eta_k + \Delta/2, \eta_k + 3\Delta/4], \mbox{ for some } s = 1, \ldots, S\right\}.
	\] 
	
	It follows from Lemma \ref{lem-large-prob} that for $j = 1, 2$, it holds that
	\begin{align*}
	\mathbb{P}\{\mathcal{B}_j\} \geq 1 - 11(n \vee T)^{-c} - 8Te^{-n}.
	\end{align*}
	The event $\mathcal{S}$ is studied in Lemma~13 in \cite{wang2018univariate}.  The rest of the proof assumes the the event $\mathcal{B}_1(\gamma) \cap \mathcal{B}_2(\gamma) \cap \mathcal{S}$. \\
	
	\noindent \textbf{Step 1.}  In this step, we will show that we will consistently detect or reject the existence of undetected change points within $(s, e)$.  Let $a_m$, $b_m$ and $m^*$ be defined as in Algorithm~\ref{algorithm:WBS}.   Suppose there exists a change point $\eta_k \in (s, e)$ such that $\min \{\eta_k - s, \, e - \eta_k\} \geq 3\Delta/4$.  In the event $\mathcal{S}$, there exists an interval $(\alpha_m, \beta_m)$ selected such that $\alpha_m \in [\eta_k - 3\Delta/4, \eta_k - \Delta/2]$ and $\beta_m \in [\eta_k + \Delta/2, \eta_k + 3\Delta/4]$.  
	
	Following Algorithm~\ref{algorithm:WBS}, $(s_m, e_m) = (\alpha_m, \beta_m) \cap (s, e)$.  We have that $\min \{\eta_k - s_m, e_m - \eta_k\} \ge (1/4)\Delta$ and $(s_m, e_m)$ contains at most one true change point.

	It follows from Lemma \ref{lem-uni-lem10}, with $c_1$ there chosen to be $1/4$, that
	\[
	\max_{s_m < t < e_m} \widetilde{D}^{t}_{s_m, e_m} \geq \frac{2^{-7/2}\kappa\Delta \sqrt{n}}{\sqrt{e-s}},
	\]
	Therefore
	\begin{align*}
	a_m  = \max_{s_m < t < e_m} D^{t}_{s_m, e_m} \geq \max_{s_m < t < e_m} \widetilde{D}^{t}_{s_m, e_m} - \gamma  \geq 2^{-7/2}C_R^{-1/2}\kappa \sqrt{\Delta n} - \gamma.
	\end{align*}
	Thus for any undetected change point $\eta_k \in (s, e)$, it holds that
	\begin{equation}\label{eq:wbsrp size of population}
	a_{m^*} = \sup_{1\le m\le  S} a_m \geq 2^{-7/2}C_R^{-1/2}\kappa \sqrt{\Delta n} - \gamma \geq c_{\tau, 2}  \kappa \sqrt{\Delta n} ,   
	\end{equation}
	where the last inequality is from the choice of $\gamma$ and $c_{\tau, 2} > 0$ is achievable with a sufficiently large $C_{\mathrm{SNR}}$ in \Cref{assume-snr}.  This means we accept the existence of undetected change points.
	
	Suppose that there are no undetected change points within $(s, e)$, then for any $(s_m, e_m)$, one of the following situations must hold.
	\begin{itemize}
		\item [(a)]	There is no change point within $(s_m, e_m)$;
		\item [(b)] there exists only one change point $\eta_k \in (s_m, e_m)$ and $\min\{\eta_k - s_m, e_m - \eta_k\} \le \epsilon_k$; or
		\item [(c)] there exist two change points $\eta_k, \eta_{k+1} \in (s_m, e_m)$ and $\eta_k - s_m \leq \epsilon_k$, $e_m - \eta_{k+1} \leq \epsilon_{k+1}$.
	\end{itemize}
	
	Observe that if (a) holds, then we have
	\[
	\max_{s_m < t < e_m} D^{t}_{s_m, e_m}  \leq \max_{s_m < t < e_m}  \widetilde{D}^{t}_{s_m, e_m}+ \gamma = \gamma  <   \tau,
	\]
	so no change points  are detected.
	
	Cases (b) and (c) are similar, and case (b) is simpler than (c), so we will only focus on case (c).  It follows from Lemma \ref{lem-uni-lem8} that
	\begin{align*}
	\max_{s_m < t < e_m} \widetilde{D}^t_{s_m, e_m} & \leq \sqrt{n/2}\sqrt{e_m - \eta_{k+1}}\kappa_{k+1} + \sqrt{n/2} \sqrt{\eta_k - s_m} \kappa_k \\
	& \leq \sqrt{2C_{\epsilon}} T^{1/2} \max\{\sqrt{d\log (n \vee T)},\, d^{3/2}\},
	\end{align*}
	therefore
	\begin{align*}
	\max_{s_m < t < e_m} D^t_{s_m, e_m} \leq \max_{s_m < t < e_m} \widetilde{D}^t_{s_m, e_m} + \gamma \leq 2\gamma < \tau.
	\end{align*}
	Under \eqref{eq-tau-cond}, we will always correctly reject the existence of undetected change points. \\
	
	\noindent \textbf{Step 2.}  Assume that there exists a change point $\eta_k \in (s, e)$ such that $\min\{\eta_k - s, \eta_k - e\} \ge 3\Delta/4$.  Let $s_m$, $e_m$ and $m^*$ be defined as in Algorithm~\ref{algorithm:WBS}.  To complete the proof it suffices to show that, there exists a change point $\eta_k \in (s_{m*}, e_{m*})$ such that $\min\{\eta_k - s_{m*}, \eta_k - e_{m*}\} \geq \Delta/4$ and $|b_{m*} - \eta_k| \leq \epsilon$.	
	
	To this end, we are to ensure that the assumptions of Lemma \ref{lem:local} are verified.  Note that \eqref{eq-lem13-2} follows from \eqref{eq:wbsrp size of population}, \eqref{eq-lem13-1} and \eqref{eq-lem13-4} follow from the definitions of events $\mathcal{B}_1(\gamma)$ and $\mathcal{B}_2(\gamma)$, and \eqref{eq-lem13-3} follows from \Cref{assume-snr}.
	
	Thus, all the conditions in Lemma \ref{lem:local} are met. Therefore, we conclude that there exists a change point $\eta_{k}$, satisfying
	\begin{equation}
	\min \{e_{m^*}-\eta_k,\eta_k-s_{m^*}\} > \Delta /4 \label{eq:coro wbsrp 1d re1}
	\end{equation}
	and
	\[
	| b_{m*}-\eta_{k}| \leq C_{\epsilon}\frac{\gamma^2}{n \kappa_k^2} \leq \epsilon,
	\]
	where the last inequality holds from the choice of $\gamma$ and \Cref{assume-snr}.
	
	The proof is completed by noticing  that \eqref{eq:coro wbsrp 1d re1} and  $(s_{m^*}, e_{m^*}) \subset (s, e)$ imply that
	\[
	\min \{e-\eta_k,\eta_k-s\}  >  \Delta /4 > \epsilon.
	\]
	As discussed in the argument before {\bf Step 1}, this implies that $\eta_k $ must be an undetected change point.	
	
\end{proof}

\section{Sensitivity analysis of the input $d$}
\label{sec:sensitivity}

We proceed with the same setting as in \Cref{sec:simulations}, focusing on Scenario 3. The only difference with \Cref{sec:simulations} is that now we explore the sensitivity of  NonPar-RDPG-CPD   to the choice $d$, by considering the performance of our algorithm  for $d \in \{7,9,11,13,15,17\}$. The results in \Cref{tab5}  show that, overall,  NonPar-RDPG-CPD   is not sensitive to  $d$, when it is not smaller than the true dimension of the latent positions.

\begin{table}[t!]
\centering
\caption{\label{tab5} Performance of  NonPar-RDPG-CPD  with data generated under  Scenario 3  for varying values of $d$.}
\medskip
\setlength{\tabcolsep}{18pt}
\begin{small}
	\begin{tabular}{ rrrrr}
		\hline
	$d$ &	$n$  & $\vert  K - \widehat{K}\vert$   & $d(\widehat{\mathcal{C}}| \mathcal{C})$  & $d( \mathcal{C}|\widehat{\mathcal{C}})$ \\ 
		\hline
	    7	   & 300 &   0.3      &   0.0 & 0.0 \\
	     9   & 300&     0.2    &0.0    &0.0  \\
		11	 & 300&   0.2      & 0.0   &  0.0\\			
		13	 & 300&     0.2    &  0.0  &  0.0\\			
		15	 & 300&     0.3    &   0.0 &0.0  \\			
		17	 & 300&       0.2  & 0.0   &0.0  \\									
		\hline	
			    7	   & 200 &      0.2   &    0.0& 0.0 \\
		9   & 200&      0.2   &   0.0 &  0.0\\
		11	 & 200&   0.1      &  0.0 &  0.0\\			
		13	 & 200&   0.1      & 1.0   & 0.0\\			
		15	 & 200&   0.1      &  1.0  &0.0  \\			
		17	 & 200&       0.1  &1.0    & 1.0 \\							
\hline		
	    7	   & 100 &    0.3     &    2.0& 2.0 \\
9   & 100&      0.2   &  1.0  &  3.0\\
11	 & 100&    0.5     &    5.0&  5.0\\			
13	 & 100&     0.5    &  3.0  &5.0  \\			
15	 & 100&     0.5    & 6.0   & 7.0 \\			
17	 & 100&    0.6     & 10.0   &  10.0\\			
\hline
	\end{tabular}
\end{small}
\end{table}

\section{Additional experiment on community structure changes only}

In this section we consider an additional scenario to the ones described in Section \ref{sec:experiments}.  Keeping everything as in Section \ref{sec:experiments}, with $\rho =0.5$, we modify Scenario 1 by setting  the matrix $Q$  as
	\[
Q_{i,j} = \begin{cases}
	0.5,  & i,j \in \mathcal{C}_l, \, l \in \{1, 2, 3\},\\
	0.3,  & \text{otherwise}.
\end{cases}
\]
where  $\mathcal{C}_1 = \mathcal{B}_1$, $\mathcal{C}_2 = \mathcal{B}_2$ and  $\mathcal{C}_3 = \mathcal{B}_3\cup \mathcal{B}_4$, with $\mathcal{B}_1,\ldots, \mathcal{B}_4$ as in Scenario 1. 
We refer to the resulting model as  Scenario 5, which consists of an example where the change happens in the community structure.  

The results in Table \ref{tab6}  show that in the setting of Scenario 5 our proposed approach still outperforms the competing methods. 
\newpage

\begin{table}[t!]
	\centering
	\caption{\label{tab6} Scenario 5}
	\medskip
	\setlength{\tabcolsep}{18pt}
	\begin{small}
		\begin{tabular}{ rrrrr}
			\hline
			Method &	$n$  & $\vert  K - \widehat{K}\vert$   & $d(\widehat{\mathcal{C}}| \mathcal{C})$  & $d( \mathcal{C}|\widehat{\mathcal{C}})$ \\ 
			\hline
			NonPar-RDPG-CPD & 300 &\textbf{0.0}  &1.0 & \textbf{1.0} \\
			NBS	                                 & 300&            21.6       &         1.0     & 43\\
			MNBS	                            & 300&          0.9        &          \textbf{0.0}    &20.0\\				
			\hline		
			NonPar-RDPG-CPD & 200 &\textbf{0.1}  &3.0 & \textbf{3.0} \\
           NBS	                                 & 200&        21.2           &    \textbf{1.0}           & 43\\
           MNBS	                            & 200&         1.1         &     4.0         &19.0\\	
			\hline
			NonPar-RDPG-CPD & 100 &\textbf{0.6}  &15.0 & \textbf{15.0} \\
			NBS	                                 & 100&       22.0            &  \textbf{2.0}             & 44.0\\
			MNBS	                            & 100&      1.1            &      5.0        &20.0\\	
			\hline
		\end{tabular}
	\end{small}
\end{table}

\section{Additional simulation results on varying $\kappa$}

We now consider the setting of Scenario 3 and allow for an extra parameter $\sigma^2$. Specifically, the data are now generated as follows. For $t \in \{1, 101\}$, we generate $Z_i(t) \stackrel{\mbox{ind}}{\sim} \mathcal{N}(0, \sigma^2  I_3)$, and for  $t \in \mathcal{A}_1 \cup \mathcal{A}_3 \backslash \{1, 101\}$, we generate
	\[
	Z_i(t)\begin{cases}
		\stackrel{\mathrm{ind}}{\sim}	\mathcal{N}(0,\sigma^2 I_3), & \text{with probability } 0.9,\\
		=	Z_i(t-1), & \text{with probability } 0.1. 
	\end{cases} 
	\]
	We then set
	\[
	P_{i,j}(t) = \frac{\exp\left\{Z_i(t)^{\top}Z_{j}(t)\right\}}{1 + \exp\left\{Z_i(t)^{\top}Z_{j}(t)\right\}}.
	\]

Furthermore, we generate $P_{i,j}(51) \sim \text{Beta}(100, 100)$, and for $t \in \{52, \ldots, 100\}$ we generate
\[
P(t) \begin{cases}
	= P(t-1),  &  \text{with probability } 0.9,\\
	\sim \text{Beta}(100,100), & \text{with probability } 0.1.\\
\end{cases}
\] 
Once the mean matrices $\{P(t)\}_{t=1}^T \mathbb{R}^{n\times n}$ have been constructed, we independently draw $A_{i,j}(t) \sim \text{Ber}(P_{i, j}(t))$, for all $i, j \in \{1, \ldots, n\}$ and $t \in \{1, \ldots, T\}$.   We consider experiments with $\sigma^2 \in \{1.5,2,2.5\}$.  This additional parameter is meant to capture different levels of jump sizes $\kappa$. 

For the model above and with the same setting from Section \ref{sec:simulations}, the results in Tables  \ref{tab8}--\ref{tab10}  show that our method once again outperforms the competing approaches.

\begin{table}[t!]
	\centering
	\caption{\label{tab8} $\sigma^2 = 1.5.$}
	\medskip
	\setlength{\tabcolsep}{18pt}
	\begin{small}
		\begin{tabular}{ rrrrr}
			\hline
			Method &	$n$  & $\vert  K - \widehat{K}\vert$   & $d(\widehat{\mathcal{C}}| \mathcal{C})$  & $d( \mathcal{C}|\widehat{\mathcal{C}})$ \\ 
			\hline
			NonPar-RDPG-CPD & 300 &\textbf{0.2}  &\textbf{0.0} & \textbf{0.0} \\
			NBS	& 300& 15.1 &1.0 &43.0 \\
			MNBS	& 300&1.1&28.0& 36.0\\				
			\hline		
			NonPar-RDPG-CPD & 200 & \textbf{0.52} &\textbf{0.0} &\textbf{0.0}  \\
			NBS	& 200& 14.0 & 1.0&44.0\\
			MNBS	& 200& 1.0&25.0 &35.0 \\	
			\hline
			NonPar-RDPG-CPD & 100 & \textbf{0.32} &\textbf{1.0}& \textbf{1.0} \\
			NBS	& 100 &14.0 &1.0 &45.0\\
			MNBS	& 100& 1.0& 25.0 &34.0\\			
			\hline
		\end{tabular}
	\end{small}
\end{table}

\begin{table}[t!]
	\centering
	\caption{\label{tab9} $\sigma^2 = 2.0.$}
	\medskip
	\setlength{\tabcolsep}{18pt}
	\begin{small}
		\begin{tabular}{ rrrrr}
			\hline
			Method &	$n$  & $\vert  K - \widehat{K}\vert$   & $d(\widehat{\mathcal{C}}| \mathcal{C})$  & $d( \mathcal{C}|\widehat{\mathcal{C}})$ \\ 
			\hline
			NonPar-RDPG-CPD & 300 &\textbf{0.2}  &\textbf{0.0} & \textbf{0.0} \\
			NBS	& 300&  14.9& 1.0&43.0 \\
			MNBS	& 300&1.2&21.0&37.0 \\				
			\hline		
			NonPar-RDPG-CPD & 200 & \textbf{0.3} &\textbf{0.0} &\textbf{0.0}  \\
			NBS	& 200&  14.4& 1.0&43.0\\
			MNBS	& 200&0.9 &26.0 &35.0 \\	
			\hline
			NonPar-RDPG-CPD & 100 & \textbf{0.3} &\textbf{1.0}& \textbf{2.0} \\
			NBS	& 100 &13.8 & 1.0&45.0\\
			MNBS	& 100& 0.9&22.0  &35.0\\			
			\hline
		\end{tabular}
	\end{small}
\end{table}

\begin{table}[t!]
	\centering
	\caption{\label{tab10} $\sigma^2 = 2.5.$}
	\medskip
	\setlength{\tabcolsep}{18pt}
	\begin{small}
		\begin{tabular}{ rrrrr}
			\hline
			Method &	$n$  & $\vert  K - \widehat{K}\vert$   & $d(\widehat{\mathcal{C}}| \mathcal{C})$  & $d( \mathcal{C}|\widehat{\mathcal{C}})$ \\ 
			\hline
			NonPar-RDPG-CPD & 300 &\textbf{0.2}  &\textbf{0.0} & \textbf{0.0} \\
			NBS	& 300&15.4  & 2.0& 43.0\\
			MNBS	& 300&1.2&21.0&35.0 \\				
			\hline		
			NonPar-RDPG-CPD & 200 & \textbf{0.4} &\textbf{0.0} &\textbf{0.0}  \\
			NBS	& 200& 14.7 &\textbf{1.0}&43.0\\
			MNBS	& 200& 0.9&21.0 &35.0 \\	
			\hline
			NonPar-RDPG-CPD & 100 & \textbf{0.3} &\textbf{1.0}& \textbf{1.0} \\
			NBS	& 100 & 13.8& \textbf{1.0}&43.0\\
			MNBS	& 100& 0.9& 27.0 &36.0\\			
			\hline
		\end{tabular}
	\end{small}
\end{table}

\bibliography{citation} 

\end{document}